\documentclass[conference,compsoc]{IEEEtran}

\ifCLASSOPTIONcompsoc
\usepackage[nocompress]{cite}
\else
\usepackage{cite}
\fi

\ifCLASSOPTIONcompsoc
\usepackage[caption = false, font = footnotesize, labelfont = sf, textfont = sf]{subfig}
\else
\usepackage[caption = false, font = footnotesize]{subfig}
\fi

\usepackage{algorithm}
\usepackage[noend]{algpseudocode}
\usepackage{amsfonts}
\usepackage{amsmath}
\usepackage{amsthm}
\usepackage{array}
\usepackage{booktabs,makecell,tabularx} 
\usepackage{enumitem}
\usepackage{fontawesome5}
\usepackage[misc]{ifsym} 
\usepackage{multirow}
\usepackage{stmaryrd}
\usepackage{tablefootnote}
\usepackage[most]{tcolorbox}
\usepackage{xspace}

\usepackage{tikz}
\usetikzlibrary{arrows}

\usepackage{graphicx}
\graphicspath{{figure/}}

\usepackage [
]{cryptocode} 
\usepackage{hyperref}

\usepackage[available, functional, reproduced]{ieeebadges}

\def\eg{\textit{e.g.}\xspace}
\def\ie{\textit{i.e.}\xspace}
\def\etc{\textit{etc.}\xspace}
\def\etal{\textit{et~al.}\xspace}
\def\cf{\textit{cf.}\xspace}

\def\odd{\textit{odd}}
\def\even{\textit{even}}

\newcommand{\sys}{\ensuremath{\mathsf{AnonGBDT}}\xspace}
\newcommand{\sysfinal}{\ensuremath{\sys^{\mathsf{OTSA}}}\xspace}

\newcommand{\sk}{\mathsf{sk}}
\newcommand{\pk}{\mathsf{pk}}
\newcommand{\sample}{\leftarrow_{\$}}
\newcommand{\concat}{\|\!}

\newtheorem{theorem}{Theorem}[section]

\theoremstyle{definition}
\newtheorem{definition}{Definition}[section]
\theoremstyle{remark}

\algrenewcommand\algorithmicrequire{\textbf{Input:}}
\algrenewcommand\algorithmicensure{\textbf{Output:}}
\algnewcommand\algorithmicsenderinput{\textbf{Sender's Input:}}
\algnewcommand\SenderInput{\item[\algorithmicsenderinput]}
\algnewcommand\algorithmicreceiverinput{\textbf{Receiver's Input:}}
\algnewcommand\ReceiverInput{\item[\algorithmicreceiverinput]}
\algnewcommand\algorithmicpainput{\textbf{$P_0$'s Input:}}
\algnewcommand\PaInput{\item[\algorithmicpainput]}
\algnewcommand\algorithmicpbinput{\textbf{$P_1$'s Input:}}
\algnewcommand\PbInput{\item[\algorithmicpbinput]}
\algnewcommand\algorithmicpcinput{\textbf{$P_{c}$'s Input:}}
\algnewcommand\PcInput{\item[\algorithmicpcinput]}
\algnewcommand\algorithmicpdinput{\textbf{$P_{1 - c}$'s Input:}}
\algnewcommand\PdInput{\item[\algorithmicpdinput]}
\algnewcommand\algorithmicpp{\textbf{Public parameters:}}
\algnewcommand\PP{\item[\algorithmicpp]}

\AtBeginDocument{%
}

\begin{document}

\title{Practical Anonymous Two-Party Gradient Boosting Decision Tree}

%
\author{\IEEEauthorblockN{Chenyu Huang\IEEEauthorrefmark{1}\IEEEauthorrefmark{5}
Fan Zhang\IEEEauthorrefmark{1}\IEEEauthorrefmark{5}\textsuperscript{\Letter},
Minxin Du\IEEEauthorrefmark{2}\textsuperscript{\Letter}, 
Sherman S. M. Chow\IEEEauthorrefmark{3},}
Huangxun Chen\IEEEauthorrefmark{4},
Huaming Rao\IEEEauthorrefmark{1},
Danqing Huang\IEEEauthorrefmark{1},
Bo Qian\IEEEauthorrefmark{1} and
Peng Chen\IEEEauthorrefmark{1}
\IEEEauthorblockA{\{chenyuhuang, zxfanzhang\}@tencent.com, minxin.du@polyu.edu.hk}
\IEEEauthorblockA{\IEEEauthorrefmark{1}Tencent, \IEEEauthorrefmark{2}Hong Kong Polytechnic University, \IEEEauthorrefmark{3}Chinese University of Hong Kong, \IEEEauthorrefmark{4}HKUST-GZ
}
}


\maketitle

\begingroup
\def\@makefntext#1{\noindent #1}%
\makeatletter
\renewcommand{\thefootnote}{} 
\renewcommand{\@makefnmark}{} 
\renewcommand{\@makefntext}[1]{\noindent #1} 
\makeatother
\footnotetext{%
\IEEEauthorrefmark{5}Co-first Authors, \textsuperscript{\Letter}\, Co-corresponding
\\
Minxin Du is supported in part by the HK PolyU Start-up Fund (P0052722).
His work is partly done at CUHK, supported by Direct Grant (4055238).
\\
Sherman Chow is supported in part by the General Research Funds (CUHK 14210825, 14210621) from the Research Grant Council, the Hong Kong Special Administrative Region of the People's Republic of China;
and by Direct Grant (4055238) and the Strategic Impact Enhancement Fund
(3135517) from CUHK.
Chow contributes to scholarly work, but administrative, business, or compliance matters fall outside his purview.
}
\makeatother
\endgroup

\begin{abstract}
Structured data is well handled by gradient-boosted decision trees (GBDT), which are usually trained on vertically partitioned features across mutually distrustful parties.
High~speed and interpretability make GBDTs popular in finance and healthcare, where neural networks may fall~short.
Enabling secure computation for GBDTs poses unique challenges, requiring secure record alignment for comparison.
Relying on private set intersection (PSI) is a \mbox{\emph{de facto} approach}.
Mistaking PSI for a safety measure actually exposes which record identifiers (IDs) are shared between the datasets.
\mbox{Although} circuit-PSI could help, 
it is costly for generic uses.
New~ideas are needed to efficiently train in a ``dark forest.''

Aiming to hide the IDs, we initiate the study of anonymous GBDT training on split data held by two parties.
Dual circuit-PSI in our design lets the parties alternate as receiver to run pick-then-sum over local features.
Via oblivious programmable pseudorandom functions, we propagate circuit-PSI outputs as shared state across runs.
Avoiding universal alignment, we resolve the neglected dilemma that ID hiding incurs a cost that scales with domain size.
Next, we halve the cost of ciphertext packing used to convert single-instruction multiple-data homomorphic encryption from (ring) learning with errors in prior secure GBDT (Usenix Security~'23) and related secure machine-learning computations.
Comparative experiments show our protocol remains competitive with leaky approaches in efficiency.
Enabling ID-hiding aggregation, our techniques can extend to other vertically partitioned analytics.
\end{abstract}




\section{Introduction}
\label{sect:intro}
Structured data is well served by gradient-boosted decision trees (GBDT), an interpretable, state-of-the-art learner for domains like finance~\cite{pvldb/CaoYCZLQ19}, healthcare~\cite{nejm/hunter2023medical}, advertising~\cite{www/LingDGZLS17}, and beyond, even amid the rise of deep neural networks~\cite{sp/NgC23}.
However, privacy regulations preclude plaintext centralization across distrustful silos, motivating private \emph{vertical} training on disjoint feature sets over overlapping populations.
Even with secure multiparty computation, different datasets seldom align record-by-record, for example, a bank seeking payment-provider signals often finds that rosters that only partially overlap and differ in order
(specifically, row $i$ on one side need not be the same as row $i$ on the other).
Reconciling them requires privately computing the identifier (\eg, national ID) intersection and agreeing on a common order, and the mapping must be refreshed as rosters churn.
Most secure two-party protocols~\cite{uss/LuHZWH23,tifs/ChenLWHXZ23,cikm/FangZT0YWWZZ21,pvldb/WuCXCO20} address this by running private set intersection (PSI)~\cite{eurocrypt/FreedmanNP04,ccs/KolesnikovKRT16} for \emph{pre-alignment}, a setup step that determines which identifiers are shared across the datasets while hiding others.
Although common, revealing intersection membership is sensitive: linking credit-card transaction IDs to advertising data enables re-identification~\cite{Science/AlexandreLKA15,eurosp/IonKNPSS0SY20}, and revealing which patients two hospitals share can leak geographic or medical information, \eg, oncology referrals in a small suburb.
Noting that GBDT needs only \emph{features}, not identifiers, we seek to keep \emph{all} IDs and their alignment private throughout collaborative training, a regime we call \emph{anonymous training}.

\subsection{Securely Sewing Metadata
via \emph{Secret Match}}
Shielding identifiers entirely, we introduce the first \emph{anonymous} two-party GBDT protocol that never reveals them, intersecting or not.
Matching that goal, circuit-PSI~\cite{eurocrypt/PinkasSTY19,eurocrypt/RindalS21,ccs/RaghuramanR22} is a natural starting point because it can evaluate arbitrary circuits on payloads attached to the hidden intersection.
Current uses, however, are limited to order‑independent aggregates like sums~\cite{eurosp/IonKNPSS0SY20} or simple statistics~\cite{acns/YingCPXL22};
even counting the intersection size incurs ${\sim}20\times$ more cost than vanilla PSI~\cite{pkc/GarimellaMRSS21}, let alone modeling an entire GBDT as an arithmetic circuit.

The difficulty stems from an optimization in circuit-PSI: the parties employ asymmetric hashing.
The receiver inserts each ID into a \emph{cuckoo} table with \emph{one} bucket per ID, and a stash handles overflows from collisions.
In contrast, the sender places each ID in \emph{all} corresponding buckets by \emph{simple hashing}, creating a one-to-many mapping.
PSI then reduces to per-bucket private set membership (PSM) testing~\cite{asiaccs/MaC22}.
Such asymmetry poses two challenges for GBDT training:

\smallskip
\noindent I) \textbf{Gradient histograms}:
Each party holds $n$ samples with~$m$ real-valued features, discretized into an $mB \times n$~binary~matrix $\mathbf{M}$ ($B$ bins per feature).
Each node maintains a \emph{sample indicator} $\mathbf{b} \in \{0,1\}^n$ tracking which samples reach it.
Gradients are filtered by $\mathbf{b}$ and aggregated by multiplying $\mathbf{M}$ to form histograms.
The receiver, knowing the one-to-one mapping from columns to gradients, can run a secure two-party functionality $\mathcal{F}_{\mathsf{BinMatVec}}$ for binary matrix-vector multiplication to multiply 
$\mathbf{M}$ with gradient shares.
However, the sender cannot because one-to-many hashing obscures which gradient aligns to which column.

\smallskip
\noindent II) \textbf{Indicator synchronization}:
A Boolean indicator vector tracks which samples reach each node, and PSI‑based schemes start with $\mathbf{1}^{n}$.
To keep it secret while allowing efficient local updates (Section~\ref{sect:base_design_revisit}), Squirrel resorts to~AND-sharing~\cite{uss/LuHZWH23}, which is leaky in this setting.
More importantly, the indicator synchronization for children fails when the sender owns the split due to the one-to-many mapping.

These challenges prompt a broader question:
\emph{``Can we perform order-dependent computation on circuit-PSI outputs without generic secure multiparty computation?''}
The techniques here answer affirmatively, enabling practical anonymous GBDT and informing other secure computations.

\subsection{Technical Overview}
\noindent
\textbf{Base solution}:
Circuit‑PSI~\cite{eurocrypt/PinkasSTY19,eurocrypt/RindalS21,ccs/RaghuramanR22} enables oblivious processing of \emph{payloads} attached to intersecting IDs.
In our base design, the \emph{sender} attaches its binary feature matrix as the payload.
After PSI, the matrix is secret‑shared and aligned with the joint ID list.
Gradient histograms are then built by a pick‑then‑sum procedure with secure multiplexer $\mathcal{F}_{\mathsf{mux}}$~\cite{ccs/RatheeR0CGR020};
this entails $O(nmB)$ oblivious transfers (OTs).

To update the per‑node sample indicator, we introduce a new OT‑based \emph{oblivious indicator synchronization} (OIS) protocol.
If the \emph{receiver} owns the best split, it directly shares one column of its feature matrix;
otherwise, the two parties execute a $1$‑out‑of-$mB$ OT for every row so that the sender can retrieve the required column without revealing the index.
In either case, secure AND~\cite{ccs/RatheeR0CGR020} and (plaintext) string XOR then derive the children's indicator shares (Section~\ref{sect:base_OIS}).
While feasible and anonymous, the base design is OT‑heavy.

\smallskip
\noindent
\textbf{Dual-circuit-PSI (which ``symmetrizes'' alignment)}:
Our final design, with codename OTSA, reduces OTs and symmetrizes alignment via a new dual-circuit-PSI framework:
two circuit-PSI instances run with swapped roles,
so each party learns an ID-to-bucket mapping once.
No feature matrix needs to be sent as a payload;
both remain local, and gradient aggregation switches from $\mathcal{F}_{\mathsf{mux}}$ to communication‑friendly
$\mathcal{F}_{\mathsf{BinMatVec}}$.
We use \sys to refer to both the base design and the final extension.

\smallskip
\noindent
\textbf{Fast ciphertext packing (halving costs)}:
Our instantiation uses homomorphic encryption instantiated from the learning-with-error (LWE) assumption and its ring variant, ring LWE (RLWE) assumption.\footnote{%
Hereinafter, we may use shorthand LWE/RLWE to refer to ciphertexts produced by an encryption scheme under the LWE/RLWE assumption.
}
In $\mathcal{F}_{\mathsf{BinMatVec}}$, gradient shares are converted to RLWE ciphertexts, homomorphically picked-and-summed, and output as multiple LWE ciphertexts that are packed into one RLWE ciphertext.
Our $\mathsf{FastPackLWEs}$ subroutine (Section~\ref{sec:fast_pack_rlwe}) first packs each adjacent LWE pair directly into one RLWE (eliminating lifting), then recursively merges RLWEs.
This streamlines the raw $\mathsf{LWEDimLift}$-then-$\mathsf{PackLWEs}$ approach~\cite{uss/LuHZWH23}, hence halving inputs and runtime,
where the former is for lifting an LWE ciphertext to a larger lattice dimension 
and the latter is for
homomorphically merging LWE ciphertexts into one RLWE ciphertext that decrypts to a related polynomial~\cite{acns/ChenDKS21}.



\smallskip
\noindent
\textbf{Batched OPPRF-based OIS}:
Keeping both feature matrices local invalidates the row-wise OT-based OIS in the base protocol.
An OIS based on oblivious programmable pseudorandom function (OPPRF) is introduced (Section~\ref{sec:opprf_indicator_sync}): the split owner programs an OPPRF so the counterparty obtains exact assignment shares for all its cuckoo buckets.
Level-wise batching yields a batched variant $\mathcal{F}_{\mathsf{BOIS}}$ with $O(D)$ cost instead of $O(2^{D})$ for tree depth~$D$.

We present independent optimizations (Sections~\ref{sect:our_sigmoid} and~\ref{sect:adv_optimizations}), \eg, better $\mathsf{sigmoid}$ approximation, gradient packing, and lightweight $\mathsf{argmax}$, which also accelerate Squirrel.
These refinements help the dual-circuit-PSI design outperform the base design and match \emph{non-anonymous}~Squirrel, despite partial per-instance recomputation, \eg, gradients.

Our main contributions are summarized below.

\smallskip
\noindent 
1) To our knowledge, \sys is the first two‑party protocol that enables \emph{anonymous} GBDT training (and inference~\cite{ndss/MaTZC21}).
\sys builds on circuit‑PSI~\cite{eurocrypt/RindalS21,ccs/RaghuramanR22}, which keeps every common ID hidden and maintains private alignment throughout training, for anonymity that prior PSI‑based approaches did not provide.

\smallskip
\noindent 
2) Circuit-PSI asymmetric hashing (cuckoo {vs.} simple)~creates two challenges, namely, gradient aggregation and indicator synchronization.
In our base design, the sender's feature matrix is carried as a payload.
Gradients are selected via the secure multiplexer~\cite{ccs/RatheeR0CGR020}, and indicators are reconciled via a custom OIS protocol.
Both are functional yet OT‑intensive.

\smallskip
\noindent
3) \sysfinal uses a new dual-circuit-PSI framework to symmetrize the workflow.
\emph{Both} parties now build histograms with the OT-free $\mathcal{F}_{\mathsf{BinMatVec}}$, optimized by our $\mathsf{FastPackLWEs}$ subroutine.
Since the base OIS no longer applies, we propose an OPPRF-based variant with level-wise batching, which reduces communication logarithmically.

\smallskip
\noindent 4)
We evaluate on real-world and synthetic datasets.
All executable code is open source\footnote{\url{https://zenodo.org/records/17373936}}.
$\sysfinal$ is up to $20\times$ faster and $44\times$ lighter than the base version, and it matches Squirrel under local area network (LAN) settings.
Its runtime and traffic never exceed $1.8\times$ those of Squirrel for wide area networks (WAN), while providing strictly stronger privacy.
(We exclude the na\"ive baseline that embeds the \emph{entire} GBDT training in circuit-PSI, which is orders of magnitude slower.)

\section{Preliminaries}
$|S|$ is the size of the set $S$.
Sampling from a finite set $S$ 
or a distribution $\mathcal{D}$
is denoted by
$x \sample S$ or
$x \leftarrow \mathcal{D}$.
$\perp$ is the error symbol.
The computational and statistical security parameters are $\lambda$ and $\kappa$, respectively.
The indicator function $\mathbf{1}\{x\}$ equals~$1$ if $x$ is true; $0$, otherwise.

Bold lowercase letters denote vectors (\eg, $\mathbf{a}$) and bold uppercase letters denote matrices (\eg, $\mathbf{M}$).
For indices, $\mathbf{a}[i]$ is the $i$-th entry of $\mathbf{a}$, $\mathbf{M}[i]$ (or $\mathbf{M}_i$) is the $i$-th row of $\mathbf{M}$, and $\mathbf{M}[i, j]$ is the entry in row~$i$, column~$j$.
Polynomials are written as $\widehat a$, where $\widehat a[i]$ is its $i$-th coefficient.
Define the polynomial ring $\mathcal{R} = \mathbb{Z}[X] / (X^N + 1)$, where $N$ is a power of~$2$.
For a modulus $Q \in \mathbb{Z}$, set $\mathcal{R}_Q
= \mathcal{R} / Q\mathcal{R}
= \mathbb{Z}_Q[X] / (X^N + 1)$.
Table~\ref{tab:notation} lists all notations.

\begin{table}[!t]
	\caption{Notations with descriptions}
	\begin{tabular}{l||l}
		\hline
		Symbol & Description
		\\
		\hline
		$\langle \cdot \rangle^\mathsf{A}_l, \langle \cdot \rangle^\mathsf{B}_l$ & $P_l$'s arithmetic and boolean shares\\
		$\llbracket \cdot \rrbracket$ & RLWE HE-encrypted variable\\
		\hline
		$\mathbf{ID}_l$, $\mathbf{X}_l$ & Sample IDs and features owned by party $P_l$
		\\
		$n_l$, $m_l$ & Number of samples and features in $\mathbf{X}_l$
		\\
		$\mathbf{M}_l$ & Discrete feature binary matrix of $\mathbf{X}_l$
		\\
		$B$ & Number of bins in $\mathbf{M}_l$
		\\ 
		$\mathbf{y}$ & Label set
		\\
		$T$ & Number of trees
		\\
		$D$ & Depth of tree
		\\		
		$z_*, u_*$ & Best split:
		$u_*$-th bin of $z_*$-th feature
		\\
		\hline
		$\mathbf{q}$ & Result of circuit-PSI
		\\
		$\mathbf{g}^{(k)}$, $\mathbf{h}^{(k)}$ & $1$st and $2$nd order of gradients of node $k$
		\\
		$\mathbf{b}^{(k)}$ & Sample indicator of node $k$
		\\ 
		$\mathbf{w}$ & Leaf weight
		\\ 
		\hline
		$\mathbf{q}^{l}$ & Result of $l$-th circuit-PSI
		\\
		$\mathbf{g}^{(k, l)}, \mathbf{h}^{(k, l)}$ & Gradients of node $k$ based on $\mathbf{q}^{(l)}$
		\\
		$\mathbf{b}^{(k, l)}$ & Sample indicator of node $k$ based on $q^{(l)}$
		\\
		$\mathbf{T}_{\mathsf{ch}}^l, \mathbf{T}_{\mathsf{sh}}^l$ & Cuckoo and simple hash tables of $l$-th circuit-PSI
		\\
		\hline
	\end{tabular}
	\label{tab:notation}
\end{table}


\subsection{Gradient Boosting Decision Tree}
\label{sect:gbdt_training}
A GBDT is an ensemble of $T$ decision trees 
$\{\mathcal{T}_t\}^T_{t=1}$ that are trained sequentially via gradient boosting~\cite{kdd/ChenG16}.
Each tree partitions the feature space at internal nodes and assigns a real-valued weight to each leaf.
For input $\mathbf{x}$, the ensemble prediction is the sum of leaf weights:
$\tilde{y} = \sum^T_{t = 1} \mathcal{T}_t(\mathbf{x})$.

\smallskip
\noindent
\textbf{Setup and gradient computation}:
Let $\mathcal{D} = \{(\mathbf{X}[i], \mathbf{y}[i])\}_{i = 1}^{n}$ be a dataset of $n$ samples with $m$ real-valued features 
$\mathbf{X}[i] \in \mathbb{R}^{m}$ and ground-truth label $\mathbf{y}[i]$.
We focus on binary classification with logistic cross-entropy loss, which is twice differentiable and well suited to second-order boosting~\cite{kdd/ChenG16}.

After $t{-}1$ trees, for sample $i$, let $\tilde{\mathbf{y}}[i] = \sum^{t - 1}_{j = 1} \mathcal{T}_j(\mathbf{x})$ and 
$\mathbf{p}[i] = \mathsf{sigmoid}\!\bigl(\tilde{\mathbf{y}}[i]\bigr)$.
The corresponding first- and second-order statistics used to fit the $t$-th tree are respectively
\begin{equation}
\mathbf{g}[i] = \mathbf{y}[i] - \mathbf{p}[i],
\quad
\mathbf{h}[i] = \mathbf{p}[i] \cdot (1 - \mathbf{p}[i]),
\label{eq:gradients_plaintext}
\end{equation}
\ie, the canonical gradient and Hessian of logistic loss~\cite{kdd/ChenG16}.

\smallskip
\noindent
\textbf{Node splitting}:
Consider a node with index set $\mathbf{I}$.
For feature~$z$ and threshold $\phi$, define the partition 
$\mathbf{I}_{L} = \{i \mid \mathbf{X}[i, z] \le \phi\}$ and
$\mathbf{I}_{R} = \mathbf{I}\setminus\mathbf{I}_{L}$.
With $\alpha \geq 0$ denoting $L_2$ regularization on leaf weights and $\gamma \geq 0$ the complexity penalty per split, the \emph{split gain} (or ``score'') $\mathbf{L}_{\mathsf{sp}}$ is 
\begin{equation}
\frac12\!\left[
\frac{\bigl(\sum_{i \in \mathbf{I}_{L}}\mathbf{g}[i]\bigr)^{2}}{\sum_{i \in \mathbf{I}_{L}} \mathbf{h}[i] + \alpha}
\!+\!
\frac{\bigl(\sum_{i \in \mathbf{I}_{R}}\mathbf{g}[i]\bigr)^{2}}{\sum_{i \in \mathbf{I}_{R}} \mathbf{h}[i] + \alpha}
\!-\!
\frac{\bigl(\sum_{i \in \mathbf{I}} \mathbf{g}[i]\bigr)^{2}}{\sum_{i \in \mathbf{I}} \mathbf{h}[i] + \alpha}\!
\right]\!- \gamma.
\label{eq:split_gain_plaintext}
\end{equation}
The best split maximizes $\mathbf{L}_{\mathsf{sp}}$ over \emph{all} candidates $\phi$ for \emph{all}~$z$.

\smallskip
\noindent\emph{Histogram-based Search}:
Exhaustive threshold scanning per feature is costly.
We apply histogram binning~\cite{kdd/ChenG16,nips/KeMFWCMYL17} to discretize each feature into $B$ ordered \emph{bins}.
Let $\mathbf{M} \in \{0, 1\}^{B m \times n}$ be the ``one-hot bin indicator'' matrix,
where $\mathbf{M}[zB + u, i] = 1$ iff sample $i$ falls into bin $u$ of feature $z$.\footnote{
For example, consider features 
age: [${<}30$, ${\geq}30$] and
weight: [${<}70$kg, ${\geq}70$kg],
discretized into $B = 2$ bins.
When given $n = 4$ samples:
$\{(25, 65), (35, 75), (45, 85), (28, 68)\}$,
$\mathbf{M} = 
[1010|0101|1010|0101]^\top$.
}

Aggregating $\{\mathbf{g}[i], \mathbf{h}[i]\}$ into per-bin histograms supports efficient split evaluation by scanning bin boundaries, and the optimal one is denoted by $(z_\star, u_\star)$.
For convenience, we pre-compute an exclusive-prefix matrix $\widetilde{\mathbf{M}} \in \{0, 1\}^{B m \times n}$ with rows $\widetilde{\mathbf{M}}[zB + u, i] = \bigvee_{v\le u} \mathbf{M}[zB + v, i]$
so that any candidate left-child assignment
$\mathbf{I}_{L}$ directly equals a row of $\widetilde{\mathbf{M}}$, avoiding runtime comparisons with bin boundaries.

\smallskip
\noindent
\textbf{Leaf weights}:
At maximum depth $D$, a leaf 
$k$ has weight 
\begin{align}
\mathbf{w}[k] = 
{-{\textstyle\sum}_{i \in \mathbf{I}_k}{\mathbf{g}[i]}}
/{({\textstyle\sum}_{i \in \mathbf{I}_k} \mathbf{h}[i] + \alpha)},
\label{eq:weight_plaintext}
\end{align}
which is the closed-form minimizer of the second-order objective with $L_2$ regularization~\cite{kdd/ChenG16}.
The prediction for sample $i$ is then updated by adding $\mathbf{w}[k]$ if $i \in \mathbf{I}_k$.

\subsection{Cryptographic Primitives}
\label{sect:crypto_tools}

\subsubsection{Secret Sharing}
We adopt $(2, 2)$-arithmetic additive secret sharing~\cite{ndss/Demmler0Z15}.
To share an $\ell$‑bit integer $x \in \mathbb{Z}_{2^{\ell}}$,
party $P_{0}$ samples
$r \sample \mathbb{Z}_{2^{\ell}}$, sends $r$ to $P_{1}$, and sets
$
\langle x\rangle^{\mathsf{A}}_{0} = x - r, \;
\langle x\rangle^{\mathsf{A}}_{1} = r.
$
Reconstruction is
$\langle x \rangle^\mathsf{A}_0 + \langle x \rangle^\mathsf{A}_1 \mod \mathbb{Z}_{2^\ell}$.

A real $\widetilde{x} \in \mathbb{R}$ is encoded as
$x = \lfloor \widetilde{x} \, 2^{f} \rceil \in [-2^{\ell-1}, 2^{\ell-1})$ using fixed‑point representation with precision~$f$.

\smallskip
\noindent $\bullet$
Addition:
$\langle z \rangle^\mathsf{A} = \langle x \rangle^\mathsf{A} + \langle y \rangle^\mathsf{A}$ is computed by $P_c$ locally.

\smallskip
\noindent $\bullet$
Multiplication:
$\langle z \rangle^\mathsf{A}= \mathcal{F}_{\mathsf{mul}}(\langle x \rangle^\mathsf{A}, \langle y \rangle^\mathsf{A})$, where $\mathcal{F}_{\mathsf{mul}}$ is a secure two-party protocol using Beaver's triples~\cite{crypto/Beaver91a}.

\smallskip
\noindent $\bullet$
Division:
$\langle z \rangle^\mathsf{A}= \mathcal{F}_{\mathsf{div}}(\langle x \rangle^\mathsf{A}, \langle y \rangle^\mathsf{A})$, where $\mathcal{F}_{\mathsf{div}}$ is a secure two-party protocol using the Goldschmidt method~\cite{fc/CatrinaS10}.

\smallskip
\noindent $\bullet$
Sigmoid:
$\langle y \rangle^\mathsf{A} = \mathcal{F}_{\mathsf{sigmoid}}(\langle x \rangle^\mathsf{A})$ is a secure two-party protocol for efficiently evaluating the sigmoid function via three‑segment Fourier approximation.

\smallskip
\noindent $\bullet$
ArgMax:
$\langle z \rangle^\mathsf{A} = \mathcal{F}_{\mathsf{argmax}}(\langle \mathbf{x} \rangle^\mathsf{A})$ is a secure two-party protocol to output the shared index $z$ of $\max \mathbf{x}$~\cite{cans/KolesnikovSS09}.

\smallskip
\noindent 
$(2, 2)$-boolean sharing~\cite{ndss/Demmler0Z15} is also used.
A boolean value $x \in \{0, 1\}$ is XOR-shared as $x = \langle x \rangle^\mathsf{B}_0 \oplus \langle x \rangle^\mathsf{B}_1$:
Party $P_0$ samples $r \in \mathbb{Z}_2$, 
sets $\langle x \rangle^\mathsf{B}_0 = x \oplus r$, and sends $\langle x \rangle^\mathsf{B}_1=r$ to $P_1$.
The same generalizes to $\ell$-bit strings with $r \sample \mathbb{Z}_{2^\ell}$.

\smallskip
\noindent $\bullet$
XOR:
$\langle z \rangle^\mathsf{B} = \langle x \rangle^\mathsf{B} \oplus \langle y \rangle^\mathsf{B}$ is computed by $P_c$ locally.

\smallskip
\noindent $\bullet$
NOT:
$\langle y \rangle^\mathsf{B} = \neg \langle x \rangle^\mathsf{B}$, where one party locally computes $\neg$ on its share, while the other keeps its share unchanged.
	
\smallskip
\noindent $\bullet$
AND:
$\langle z \rangle^\mathsf{B} = \mathcal{F}_{\mathsf{and}}(\langle x \rangle^\mathsf{B}, \langle y \rangle^\mathsf{B})$ is achieved by a secure two-party protocol using precomputed boolean triples~\cite{ndss/Demmler0Z15}.

\smallskip
\noindent $\bullet$
OR:
$\langle z \rangle^\mathsf{B} = \mathcal{F}_{\mathsf{or}}(\langle x \rangle^\mathsf{B}, \langle y \rangle^\mathsf{B})$ via $x \land y\oplus(\neg x \oplus \neg y)$.

\smallskip
\noindent$ \bullet$
Multiplexer: 
$\langle z \rangle^\mathsf{A} = \mathcal{F}_{\mathsf{mux}}(\langle x \rangle^\mathsf{B}, \langle y \rangle^\mathsf{A})$ takes \emph{arithmetic} shares $\langle y \rangle^\mathsf{A}$ and \emph{boolean} shares of a choice bit $\langle x \rangle^\mathsf{B}$ as~input.
It outputs $\langle z \rangle^\mathsf{A}$ with $z = y$ if $x = 1$;
else $z = 0$.
It can be realized via two $\mathcal{F}_{\mathsf{OT}}$ calls~\cite{ccs/RatheeR0CGR020} (introduced below).

\smallskip
\noindent $\bullet$
Comparison:
$\langle z \rangle^\mathsf{B} = \mathcal{F}_{\mathsf{greater}}(\langle x \rangle^\mathsf{A}, \langle y \rangle^\mathsf{A})$.
If $x > y$, then $z$ is $1$;
else $z$ is $0$.
When the bitlength of $x$ and $y$ is small, the secure comparison can be efficiently realized with only $\mathcal{F}_{\mathsf{OT}}$~\cite{ccs/RatheeR0CGR020}.
Otherwise, one can recursively divide them as $x = x_1 || x_2, y = y_1 || y_2$ and compares them based on 
$\mathbf{1}\{x > y\} = \mathbf{1}\{x_1 > y_1\} \oplus (\mathbf{1}\{x_1 = y_1\}\wedge \mathbf{1}\{x_2 > y_2\}),$
where $\wedge$ and $>,=$ are respectively realized by $\mathcal{F}_{\mathsf{and}}$ and $\mathcal{F}_{\mathsf{OT}}$~\cite{ccs/RatheeR0CGR020}.

\subsubsection{Oblivious Transfer (OT)}
A $1$-out-of-$n$ OT~\cite{iacr/Rabin05} runs between a sender with $n$ messages and a receiver with choice $i \in [0, n - 1]$.
The ideal functionality is $\mathcal{F}_{\mathsf{OT}}$;
the receiver only learns the $i$-th message,
and the sender learns nothing about~$i$.
Ferret OT~\cite{ccs/YangWLZW20} is used for its low communication.

\subsubsection{Circuit-PSI}
\label{sect:pre_cpsi}
Private set intersection (PSI) allows two parties to learn $X \cap Y$ of their input sets while hiding non-intersecting elements.
Circuit-PSI~\cite{eurocrypt/PinkasSTY19,eurocrypt/RindalS21} further enables evaluating an arbitrary circuit $f$ on payloads attached to items in $X \cap Y$, without revealing the intersection itself.

\smallskip
\noindent
\textbf{Hashing phase}:
Fix $e$ hash functions $h_1, \ldots, h_e$.
The \emph{receiver} inserts each 
$\mathbf{x}\!\in\!X$ into a cuckoo-hash table $\mathbf{T}_{\mathsf{ch}}$ of $(1 + \varepsilon)n$ buckets:
first empty among $\{h_i(\mathbf{x})\}$, with evictions as needed.
With suitable $(\varepsilon, e)$, no stash is required~\cite{eurocrypt/PinkasSTY19}.
The \emph{sender} inserts each $\mathbf{y}\!\in\!Y$ into all corresponding buckets of a simple-hash table $\mathbf{T}_{\mathsf{sh}}$, permitting multiple items per bucket (vs. at most one in $\mathbf{T}_{\mathsf{ch}}$).
Both tables are padded to a common maximum bucket size.
Shared hashes ensure any common item lands in one identical bucket on both sides, reducing PSI to per-bucket \emph{private set membership} (PSM).

\smallskip
\noindent
\textbf{PSM via OPPRF}:
An oblivious programmable pseudorandom function functionality $\mathcal{F}_{\mathsf{OPPRF}}$ realizes PSM~\cite{eurocrypt/PinkasSTY19,eurocrypt/RindalS21}.
For bucket $i$, the sender programs PRF $F_i$ such that $F_i(\mathbf{y}) = r_i$ for every~$\mathbf{y}$ in $\mathbf{T}_{\mathsf{sh}}[i]$ and outputs random values elsewhere.
Both parties query $\mathcal{F}_{\mathsf{OPPRF}}$ on the receiver's items from $\mathbf{T}_{\mathsf{ch}}$;
the receiver obtains $r_i'$.
A secure equality test $\mathcal{F}_{\mathsf{equal}}$ then yields shared bits of $r_i \stackrel{?}{=} r_i'$.

\smallskip
\noindent
\textbf{Evaluation with payloads}:
Each item may carry a payload, enabling $f$ to be computed on payloads of the (hidden) intersection~\cite{eurocrypt/PinkasSTY19}. 
Since a sender's bucket can contain multiple payloads, two batched OPPRFs select the correct one~\cite{eurocrypt/PinkasSTY19}.

Looking ahead, our base design attaches the receiver's labels and the sender's (binary-matrix encoded) features as payloads to $\mathcal{F}_{\mathsf{CPSI}}$.
$\sysfinal$ alternates receiver roles and uses OPPRF to synchronize intermediate states;
only labels are carried as payloads, improving efficiency.

\subsubsection{Additive Homomorphic Encryption (HE)}
Additive HE supporting homomorphic additions is instantiated from LWE and the ring variant, RLWE~\cite{crypto/BrakerskiV11}.
Public parameters $(N, Q)$ specify the lattice dimension and modulus.
Define $\mathcal{R}_{Q} = \mathbb{Z}_{Q}[X] / (X^{N} + 1)$.
Let $\chi_{\sk}$ and~$\chi_{\mathsf{err}}$ denote the secret and error distributions (\eg, $\{-1, 0, 1\}^N$).

\smallskip
\noindent $\bullet$
\textbf{KeyGen}:
It generates a key pair $(\sk, \pk)$ for $\mathsf{RLWE}$ 
($\sk \leftarrow \chi_\sk$ and $\pk \in \mathcal{R}_Q^2$).
Identify the $\mathsf{LWE}$ secret key $\mathbf{s} \in \mathbb{Z}_Q^N$ as the coefficient vector of $\sk$, \ie, $\mathbf{s}[i] = \sk[i]$,~$\forall i \in [N]$.

\smallskip
\noindent $\bullet$
\textbf{Encryption}:
$\mathsf{LWE}$ encryption of $m \in \mathbb{Z}$ is given as
$\mathsf{LWE}_{Q, \mathbf{s}}(m) = (\mathbf{a}, b) = (\mathbf{a}, \langle\mathbf{a}, \mathbf{s} \rangle + e + m) \in \mathbb{Z}_Q^{N + 1}$ with $\mathbf{a} \sample \mathbb{Z}_Q^N$ 
and 
$e \leftarrow \chi_{\mathsf{err}}$.
$\mathsf{RLWE}$ encryption of $\widehat{m} \in \mathcal{R}_Q$ is 
$\mathsf{RLWE}_{Q, \sk}(\widehat{m}) = (\widehat{a}, -\widehat{a} \cdot \sk + \widehat{e} + \widehat{m}) \in \mathcal{R}_Q\times \mathcal{R}_Q$ (or simply $\llbracket \widehat{m} \rrbracket$), 
where $\widehat{a} \sample R_Q$ and 
$\widehat{e}[i] \leftarrow \chi_{\mathsf{err}}$.

\smallskip
\noindent $\bullet$
\textbf{Decryption}:
An $\mathsf{RLWE}$ ciphertext $\widehat{ct} = \mathsf{RLWE}_{Q, \sk}(\widehat{m}) = (\widehat{a}, \widehat{b}) \in \mathcal{R}_Q\times \mathcal{R}_Q$ is decrypted by 
$\mathsf{RLWE}_{Q, \sk}^{-1}(\widehat{a}, \widehat{b}) = \widehat{a}\cdot \sk + \widehat{b} = \widehat{m} + \widehat{e} \in \mathcal{R}_Q$.
(We omit the subscript $Q$ when clear from context.)

\smallskip
\noindent $\bullet$
\textbf{PackLWEs}:
Given a set of $\mathsf{LWE}$ ciphertexts $\{\mathsf{LWE}_{\mathbf{s}}(m_{i})\}$, an FFT-style automorphism-and-keyswitch procedure can pack them into a single $\mathsf{RLWE}$ ciphertext~\cite{acns/ChenDKS21}.
For $\widehat{m}(X)$, apply the automorphism $\mathsf{EvalAuto}(\mathsf{RLWE}_{\widehat{\sk}}(\widehat{m}(X)), t)$ that applies $X \mapsto X^t$ to both ciphertext and key, then perform
$$
\mathsf{KeySwitch}_{\widehat{\sk}_t \to \widehat{\sk}}\!\colon
\mathsf{RLWE}_{\widehat{\sk}_t}\bigl( \widehat{m}_t \bigr)
\!\longmapsto
\mathsf{RLWE}_{\widehat{\sk}}\bigl(\widehat{m}_t\bigr),
$$
where
$\widehat{\sk}_t = \widehat{\sk}(X^t)$ and $\widehat{m}_t = \widehat{m}(X^t)$.

\smallskip
\noindent $\bullet$
\textbf{Arithmetic share to HE}:
$\mathcal{F}_{\mathsf{A2H}}$.
It converts $\langle \mathbf{a} \rangle^\mathsf{A}$ to an RLWE ciphertext $\llbracket \mathbf{a} \rrbracket$.
$P_l$ transforms $\langle\mathbf{a}[i]\rangle^\mathsf{A}_l$ from $\mathbb{Z}_{2^\ell}$ to $\mathbb{Z}_Q$, then encodes it to the polynomial $\widehat{a}_l = \sum_{i = 0}^{N - 1} \langle \mathbf{a}[i] \rangle^\mathsf{A}_l X^i$.
$P_c$ sends $\mathsf{RLWE}_{\sk_c}(\widehat{a}_c)$ to $P_{1 - c}$ encrypted under $\sk_c$.
$P_{1 - c}$ gets $\llbracket \mathbf{a} \rrbracket$ by homomorphically adding $\widehat{a}_{1 - c}$ to $\mathsf{RLWE}_{\pk_c}(\widehat{a})$~\cite{uss/LuHZWH23}.

\smallskip
\noindent $\bullet$
\textbf{HE to Arithmetic share}:
$\mathcal{F}_{\mathsf{H2A}}$.
It converts RLWE $\llbracket \mathbf{a} \rrbracket$ encrypted by $P_c$ to $\langle \mathbf{a} \rangle^\mathsf{A}_l$.
$P_{1 - c}$ samples a random $\widehat{r} \sample \mathcal{R}_Q$, and gets $\widehat{ct}$ by homomorphically subtracting it from $\llbracket \mathbf{a} \rrbracket$.
$P_c$ decrypts $\widehat{ct}$ and sets its coefficient vector as $\langle \mathbf{a} \rangle^\mathsf{A}_c$.
$P_{1 - c}$ sets $\langle \mathbf{a}[i] \rangle^\mathsf{A}_{1 - c}$ by transforming $\widehat{r}[i]$ from $Q$ to $\mathbb{Z}_{2^\ell}$~\cite{uss/LuHZWH23}.

\begin{figure}[!t]
	\centering
	\includegraphics[width = \linewidth]{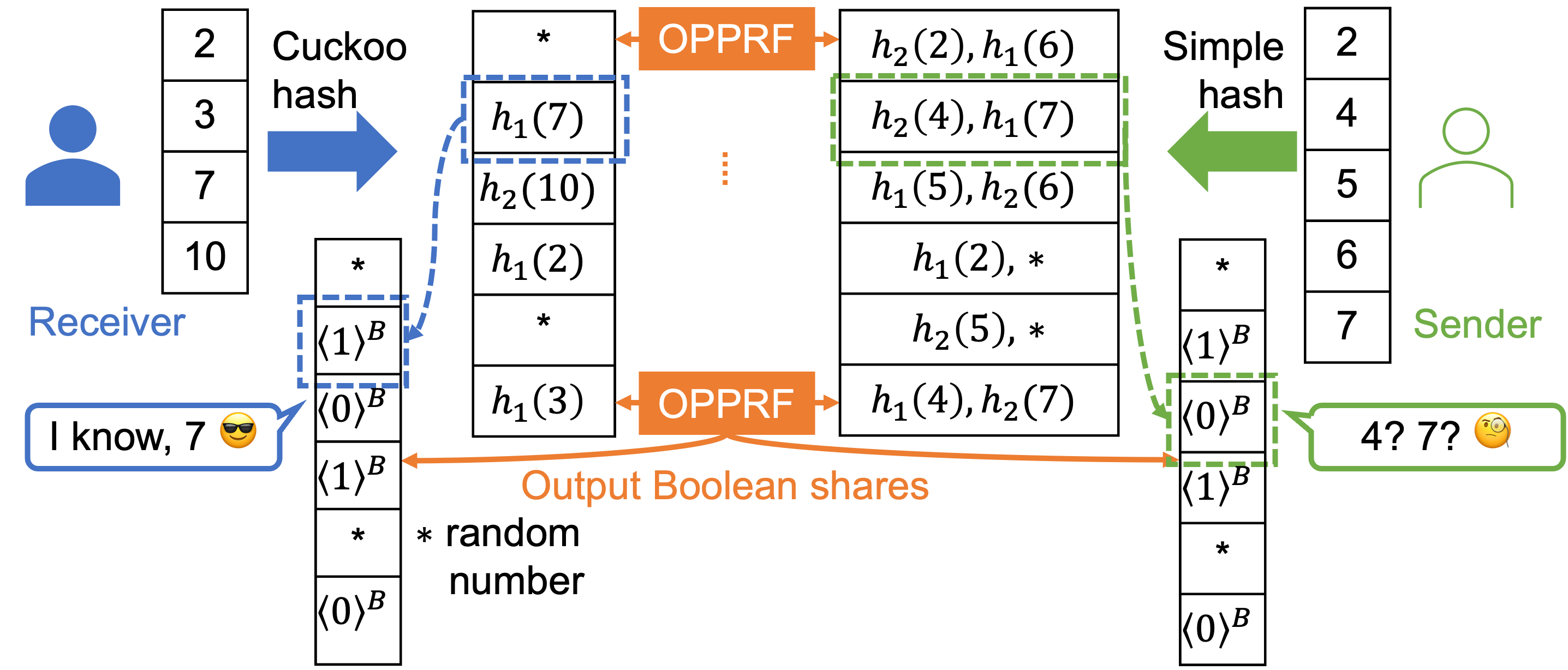}
	\caption{Circuit-PSI asymmetry (cuckoo vs.~simple hash):
	the receiver learns $7$ is in the second bucket, but the sender cannot tell which item ($4$ or $7$) is in the second one.}
	\label{fig:cpsi_asymmetric}
	\vspace{-5pt}
\end{figure}
\section{System and Security Models}

Two parties $P_0$ and $P_1$ hold vertically partitioned data $(\mathbf{ID}_l, \mathbf{X}_l) $ for $l \in \{0, 1\}$, where $\mathbf{ID}_l$ draws from a common identifier domain.
Identifier lists may overlap, while feature spaces are disjoint:
$\mathbf{X}_l \in \mathbb{R}^{n_l \times m_l}$
with distinct column sets.
For simplicity, we assume $m_0 = m_1 = m$ and $n_0 = n_1 = n$ (unbalanced $n_0 \neq n_1$ are considered in experiments).
Let $\mathbf{ID}_l[i]$ be the $i$-th identifier of $P_l$ and $\mathbf{X}_l[i]$ its associated feature vector.
The parties run a two-party protocol $\varPi$ to train a GBDT on the (implicit) join of their datasets.

All prior work~\cite{uss/LuHZWH23,cikm/FangZT0YWWZZ21,pvldb/WuCXCO20} assumes pre-aligned data, \ie, $\mathbf{X}_0[i]$ and $\mathbf{X}_1[i]$ refer to the same sample.
This implicitly discloses $\mathbf{ID}_0 \cap \mathbf{ID}_1$, (in the worst case, all identifiers), enabling re-identification.
In contrast, our designs reveal \emph{no} identifiers, \emph{no} membership in the intersection, \emph{not} even the cardinality $|\mathbf{ID}_0 \cap \mathbf{ID}_1|$ nor whether the intersection equals the entire dataset;
it also hides all alignment information (\eg, for a common item at $\mathbf{ID}_0[i]$, the counterparty's index~$j$ such that $\mathbf{ID}_0[i] = \mathbf{ID}_1[j]$ is never learned).

Figure~\ref{fig:functionality_gbdt} specifies the ideal functionality $\mathcal{F}_{\mathsf{GBDT}}$ for~one tree in $\mathcal{T}$ with two parties' data and public parameters $\mathbf{pp}$.
Without loss of generality, $P_0$ supplies labels $\mathbf{y}$.
Each tree is complete and balanced with depth~$D$;
internal nodes are indexed by $1 \le k < 2^{D - 1}$ and leaves by $2^{D - 1} \le k < 2^{D}$.
Internal computation is done on the implicit join $\mathbf{ID}_0 \cap \mathbf{ID}_1$:
Gradients/Hessians are computed with sigmoid approximation and then aggregated as histograms per bin (Section~\ref{sect:gbdt_training}).
For each internal node $k$ and candidate $(z, u)$, the split gain $\mathbf{L}_{\mathsf{sp}}$ is evaluated as in Eq.~\eqref{eq:split_gain_plaintext};
the best split $(z^{(k)}_*, u^{(k)}_*)$ maximizing $\mathbf{L}_{\mathsf{sp}}$ is revealed to $P_l$ owning~$z^{(k)}_*$.
Given the split, sample indicators update via exclusive-prefix matrices.
Upon completion, $P_l$ receives its own best splits $\mathbf{C}_l[k]$ for $k < 2^{D-1}$, and the leaf weights are additively secret-shared.

We consider static and semi-honest probabilistic polynomial time (PPT) adversaries~\cite{uss/LuHZWH23,cikm/FangZT0YWWZZ21,pvldb/WuCXCO20}, which follow the protocol without deviation but try to learn more beyond the outputs.
Formally, let $\varPi$ be a two-party protocol computing a deterministic functionality $\mathcal{F}\colon \{0, 1\}^* \times \{0, 1\}^* \mapsto \{0, 1\}^* \times \{0, 1\}^*$.
The output of $\mathcal{F}$ is a pair, $\mathcal{F}_0(x_0, x_1)$ to $P_0$ 
and $\mathcal{F}_1(x_0, x_1)$ to $P_1$.
For $i \in \{0, 1\}$, the view of $P_i$ during an execution of $\varPi$ on $(x_0, x_1)$ is denoted by $\mathcal{V}^{\varPi}_i(x_0, x_1)$.

\begin{definition}
	For a function $\mathcal{F}$, $\varPi$ privately computes $\mathcal{F}$ if there are PPT algorithms $\mathcal{S}_0$ and $\mathcal{S}_1$ such that, $i\!\in\!\{0, 1\}$, 
$$
\mathcal{S}_i(x_i, \mathcal{F}_i(x_0, x_1)) \approx_c \mathcal{V}^{\varPi}_i(x_0, x_1),
\forall x_0, x_1 \in \{0, 1\}^*, 
$$
	where $\approx_c$ denotes computational indistinguishability.
\label{def:private_function}
\end{definition}


\createprocedureblock{procb}{boxed}{}{}{}
\begin{figure}
	\begin{tcolorbox}[colback = white!5!white, left = 0mm, right = 0mm, top = -4mm, bottom = 0mm, title = \textbf{Functionality} $\mathcal{F}_{\mathsf{GBDT}}(\mathbf{Input}_0^{\varPi}\mbox{, }\mathbf{Input}_1^{\varPi}\mbox{, }\mathbf{pp})$]
	\pseudocodeblock[width = 23em, mode = text]{}{$\mathbf{Input}^\varPi_0 = \{\mathbf{ID}_0, \mathbf{X}_0, \mathbf{y}\}, \mathbf{Input}^\varPi_1 = \{\mathbf{ID}_0, \mathbf{X}_1\}$\\
			$\mathbf{Output}_0^{\varPi} = \{\mathbf{C}_0, \{\langle \mathbf{w}[k] \rangle_0^\mathsf{A}\}_{k\ge 2^{D - 1}}\}$,
			where \\ $\mathbf{C}_0[k] = (z^{(k)}_*, u^{(k)}_*)$ if $z^{(k)}_* < m$;
			otherwise, $\mathbf{C}_0[k] = \perp$ \\
			$\mathbf{Output}_1^{\varPi} = \{\mathbf{C}_1, \{\langle \mathbf{w}[k] \rangle_1^\mathsf{A}\}_{k\ge 2^{D - 1}}\}$, where 
			\\
			$\mathbf{C}_1[k] = (z^{(k)}_*, u^{(k)}_*)$ if $z^{(k)}_* \ge m$;
			otherwise, $\mathbf{C}_1[k] = \perp$ 
		}
	\end{tcolorbox}
	\vspace{-5pt}
	\caption{Ideal private GBDT training functionality: Outputs reveal only per-party best splits $\mathbf{C}$ and secret-shared weights $\langle \mathbf{w} \rangle^\mathsf{A}$.}
	\label{fig:functionality_gbdt}
\vspace{-3pt}
\end{figure}


\begin{theorem}[Modular Composition~{\cite{joc/Canetti00}}]
Let 
$\mathcal{F}_{1}, \ldots, \mathcal{F}_{m}$ be ideal two-party functionalities, and 
let $\varPi$ be a two-party protocol in the hybrid model that can invoke at most one of 
$\mathcal{F}_{1}, \ldots, \mathcal{F}_{m}$ per round.
For $i \in [m]$, let $f_{i}$ be a real two-party protocol that realizes $\mathcal{F}_{i}$ in the computational setting.
Then, for any PPT passive adversary $\mathcal{A}$ and any PPT environment $\mathcal{Z}$, there exists a PPT passive simulator $\mathcal{S}$ such that
$\mathsf{Sim}_{\varPi, \mathcal{S}, 
\mathcal{Z}}^{\mathcal{F}_1, \mathcal{F}_2, \ldots, \mathcal{F}_m} \approx_c \mathsf{Exec}_{\varPi^{(f_1, f_2, \ldots, f_m)}, \mathcal{A}, \mathcal{Z}}$.
\end{theorem}

\section{\texorpdfstring{$\sys$}{AnonGBDT}: Base Design}
Squirrel~\cite{uss/LuHZWH23}, a recent two-party secure GBDT protocol, presumes a priori alignment:
$\mathbf{X}_{0}[i]$ and $\mathbf{X}_{1}[i]$ are the same sample~$i$, enabling \emph{direct} training on $\mathbf{X}_{0}\,\|\,\mathbf{X}_{1}$.
In practice, independently maintained data only partially overlap;
perfect alignment is uncommon, rendering Squirrel impractical without a preceding private-matching step.

Circuit-PSI is utilized to reconcile $\mathbf{ID}_{0}$ and $\mathbf{ID}_{1}$ before training, without revealing identifiers.
Subsequent training then follows Squirrel's workflow (using primitives in Section~\ref{sect:crypto_tools}), with targeted adaptations to accommodate circuit-PSI's asymmetric hashing and payload interface.
A brief recap of Squirrel is first given, preceding our base design.
We follow their convention of underlining, \eg, $\underline{sk}$ denotes a secret key under the lifted HE parameters such as $\underline{N}$.


\subsection{Revisit of Squirrel}
\label{sect:base_design_revisit}

\noindent
\textbf{Data pre-processing}:
Each party $P_l$ discretizes raw features $\mathbf{X}_l$ into a one-hot bin-indicator matrix $\mathbf{M}_l$ and its exclusive-prefix variant $\widetilde{\mathbf{M}}_l$, as detailed in Section~\ref{sect:gbdt_training}.

\smallskip
\noindent
\textbf{Private sample tracking}:
At each node $k$, an $n$-bit \emph{indicator} $\mathbf{b}^{(k)}$ (initialized to $\mathbf{1}^{n}$ at the root) is maintained, where $\mathbf{b}^{(k)}[i] = 1$ iff sample~$i$ reaches node $k$.
To keep $\mathbf{b}^{(k)}$ private, Squirrel uses ``AND-style'' sharing\footnote{%
This sharing cannot ``perfectly'' hide $\mathbf{b}^{(k)}$:
If a party's share equals~$0$, it immediately learns the corresponding secret bit is $0$.}:
$\mathbf{b}^{(k)} = \mathbf{b}^{(k)}_{0}\wedge\mathbf{b}^{(k)}_{1}$.

If party $P_{c}$\footnote{%
$P_c$: asymmetric steps;
only $P_c$ acts.
$P_l$: symmetric steps;
both act.
}
owns the best split $(z_*, u_*)$, it locally forms the (plaintext) assignment 
$\mathbf{b}^{(k)}_{*} = \widetilde{\mathbf{M}}_{c}[z_* B + u_*]$ and updates
\begin{align}
\label{sample_indicator_update}
\mathbf{b}^{(2k)}_{c} = \mathbf{b}^{(k)}_{c} \wedge \mathbf{b}^{(k)}_{*}, \quad
\mathbf{b}^{(2k + 1)}_{c} = \mathbf{b}^{(k)}_{c} \oplus \mathbf{b}^{(2k)}_{c},
\end{align} 
while $P_{1 - c}$ copies its share without $\mathbf{b}^{(k)}_{*}$.
This AND-sharing update is incompatible with XOR shares produced after circuit-PSI alignment;
Section~\ref{sect:base_OIS} introduces an \emph{oblivious indicator synchronization} (OIS) for XOR-shared $\mathbf{b}^{(k)}$.

\smallskip
\noindent \textbf{Gradient computation}:
Fourier-series sigmoid~\cite{acsac/ZhengZCPTLY23} and other nonlinear functions~\cite{ndss/KeiC25} have been studied for efficient and accurate secure computation.
To obtain arithmetic shares
$\langle \mathbf{g} \rangle^{\mathsf{A}}_l$ and 
$\langle \mathbf{h} \rangle^{\mathsf{A}}_l$ (Eq.~\eqref{eq:gradients_plaintext}),
Squirrel adopts a Fourier-series sigmoid approximation that reduces secure multiplications (hence communication) compared to spline-based methods~\cite{cikm/FangZT0YWWZZ21,eurosp/Amit24}, yet still contributes ${\approx}50\%$ of total traffic.

Specifically, for $x \in [-\mu, \mu]$, it evaluates 
\(
a_{0} + \sum_{j = 1}^{J}{a_{j} \sin{\bigl(2\pi j x / 2^{L + 1}\bigr)}}
\), with $\mu = 5.6$, degree $J = 8$, scaling factor $L = \lceil \log_{2} \mu \rceil$, and maximal error $0.022$;
outside this range, outputs are clipped to $0$ (for $x < -\mu$) or $1$ (for $x > \mu$).
This requires only secure comparisons with $\pm \mu$ and secure multiplications for the series.
Section~\ref{sect:our_sigmoid} optimizes such an approximation by lowering $J$ and reusing primitives in $\pm \mu$~comparisons.

\smallskip
\noindent \textbf{Gradient aggregation}:
We describe the processing for $\mathbf{g}$ (it applies identically to $\mathbf{h}$).
Non-reachable samples are filtered by $\mathbf{g}^{(k)} = \mathbf{b}^{(k)} \odot \mathbf{g}$, where $\odot$ is element-wise multiplication, realized via a secure multiplexer:
\begin{align}
\langle \mathbf{g}^{(k)} \rangle_l^\mathsf{A}
= \mathcal{F}_{\mathsf{mux}}(\langle \mathbf{b}^{(k)} \rangle^\mathsf{B}_l, \langle \mathbf{g} \rangle^\mathsf{A}_l).
\label{eq:filter_gradient}
\end{align}
For the left child $2k$, Squirrel halves the usage of OT based on 
$\mathbf{g}^{(2k)} = \mathbf{b}^{(2k)} \odot \mathbf{g} = (\mathbf{b}^{(k)} \wedge \mathbf{b}^{(k)}_*) \odot \mathbf{g} = \mathbf{b}^{(k)}_* \odot \mathbf{g}^{(k)}$ given non-shared $\mathbf{b}^{(k)}_*$ (vs.~simply using $\mathcal{F}_{\mathsf{mux}}$ over AND-shared $\mathbf{b}^{(2k)}$).
The right child is $ \mathbf{g}^{(2k + 1)} = \mathbf{g}^{(k)} - \mathbf{g}^{(2k)}$.

Filtered gradients are then binned by \emph{pick-then-sum}:
$\widetilde{\mathbf{g}}^{(k, c)}[j] = \sum_{i | \mathbf{M}_c[j, i] = 1}{\mathbf{g}^{(k)}[i]}$ (for $i \in [n], j \in [1, mB]$), or equivalently $\mathbf{M}_c \cdot \mathbf{g}^{(k)}$, which is realized by the secure binary-matrix-vector multiplication $\mathcal{F}_{\mathsf{BinMatVec}}$:
\\
i) convert $\langle \mathbf{g}^{(k)} \rangle^\mathsf{A}_l$ to $\mathsf{RLWE}$-encrypted $\llbracket \mathbf{g}^{(k)}\rrbracket$ by $\mathcal{F}_{\mathsf{A2H}}$,
\\
ii) achieve efficient homomorphic aggregation on extracted LWE ciphertext from $\llbracket \mathbf{g}^{(k)}\rrbracket$, and pack to $\llbracket \widetilde{\mathbf{g}}^{(k)}\rrbracket$ within the RLWE domain via $\mathsf{PackLWEs}$.
which we optimize;
\\
iii) 
convert $\llbracket \widetilde{\mathbf{g}}^{(k)}\rrbracket$ back to shares $\langle \widetilde{\mathbf{g}}^{(k)} \rangle^\mathsf{A}_l$ via $\mathcal{F}_{\mathsf{H2A}}$.

In our anonymous setting, the circuit-PSI sender does not know the correspondence between $\mathbf{M}[:, i]$ and $\mathbf{g}^{(k)} [i]$, making $\mathcal{F}_{\mathsf{BinMatVec}}$ inapplicable as is;
Section~\ref{sect:hist_agg} addresses this by treating $\mathbf{M}$ as an additional payload.

\smallskip
\noindent \textbf{Best-split selection}:
For each candidate $(z, u)$, the split gain $\mathbf{L}_{\mathsf{sp}}$ (Eq.~\eqref{eq:split_gain_plaintext}) is securely evaluated using $\mathcal{F}_{\mathsf{div}}$ and $\mathcal{F}_{\mathsf{mul}}$.
$\mathcal{F}_{\mathsf{argmax}}$ reveals the maximizer
$(z_*, u_*)$ to its owner.

\subsection{Our Base Design}
\label{sect:base}
We begin with the stash-less circuit-PSI~\cite{ccs/RaghuramanR22}, using $P_{0}$ as receiver and $P_{1}$ as sender, to privately reconcile $\mathbf{ID}_0$ and $\mathbf{ID}_1$.
The output is XOR-shared $\mathbf{b}^{(1)} \in \{0, 1\}^{(1 + \varepsilon)n}$ over the receiver's cuckoo-hash order (Section~\ref{sect:pre_cpsi}), replacing the $\mathbf{1}^n$ root indicator in Squirrel.
Subsequent training avoids leaky AND-sharing by maintaining XOR shares.

Encoding the entire training pipeline as one circuit with payloads $(\mathbf{M}_{0}, \widetilde{\mathbf{M}}_{0}, \mathbf{y})$ and $(\mathbf{M}_{1}, \widetilde{\mathbf{M}}_{1})$ is impractical.
We minimize payloads and exploit circuit-PSI's order alignment to enable efficient, order-dependent steps.

\subsubsection{Gradient Histogram Aggregation}
\label{sect:hist_agg}
Order after circuit-PSI aligns elements for the receiver:
the cuckoo-hash permutation fixes a one-to-one correspondence between the intersection positions and local samples.
Concretely:

\smallskip
\noindent \emph{Receiver side $P_{0}$}.
As $\mathbf{M}_{0}$ is aligned (in cuckoo order) with samples' gradient shares from $\mathcal{F}_{\mathsf{sigmoid}}$ and $\mathcal{F}_{\mathsf{mul}}$ (Line~\ref{alg:base_sys:gradient_comp}, Figure~\ref{alg:base_sys}), $P_{0}$ can build histograms with the communication-friendly $\mathcal{F}_{\mathsf{BinMatVec}}$ as in {Squirrel};
only $\mathbf{y}$ is payload.

\smallskip
\noindent \emph{Sender side $P_{1}$}.
The sender's simple hashing creates a one-to-many mapping, so it does not learn which gradient index matches which column of $\mathbf{M}_{1}$.
To restore compatibility, $\mathbf{M}_{1}$ is carried as a circuit-PSI payload, yielding XOR-shared entries $\langle \mathbf{M}_{1} \rangle^{\mathsf{B}}$ aligned to the receiver's order.
Gradients are then aggregated by pick-then-sum using secure multiplexers:
\begin{align}
\label{our_gradient_agg}
\forall (z, u) \colon
 {\textstyle\sum}_i \mathcal{F}_{\mathsf{mux}} \left(\langle \mathbf{M}_1[zB + u, i] \rangle^\mathsf{B}_l, \langle \mathbf{g}[i] \rangle^\mathsf{A}_l \right),
\end{align}
and likewise for $\mathbf{h}$.
This avoids generic secure matrix multiplication (which would invoke $\mathcal{F}_{\mathsf{mul}}$ per entry) while using aligned order: $\langle \mathbf{M}_1[\cdot, i] \rangle^\mathsf{B}_l$, $\langle \mathbf{g}[i] \rangle^\mathsf{A}_l$ refer to the same sample.

Root filtering uses $\mathbf{b}^{(1)}$:
both parties compute $\langle \mathbf{g}^{(1)} \rangle^\mathsf{A}_l = \mathcal{F}_{\mathsf{mux}}(\langle \mathbf{b}^{(1)} \rangle_l^\mathsf{B}, \langle \mathbf{g} \rangle^\mathsf{A}_l)$ and similarly for $\mathbf{h}$ (Line~\ref{alg:base_sys:filter}, Figure~\ref{alg:base_sys}).

\smallskip
\noindent
\textbf{Complexity}: 
For $\mathbf{M}_0$, $\mathcal{O}(Bmn)$ invocations of $\mathcal{F}_{\mathsf{mux}}$, each using two $\mathcal{F}_{\mathsf{OT}}$ calls (that send $\mathcal{O}(\ell)$ bits per OT); the total communication round is $\mathcal{O}(1)$.
For $\mathbf{M}_1$, $\mathcal{F}_{\mathsf{BinMatVec}}$ incurs $\mathcal{O}(n\log{Q})$ bits via $\mathcal{F}_{\mathsf{A2H}}$ and $\mathcal{O}(2Bm\log{Q})$ bits via $\mathcal{F}_{\mathsf{H2A}}$; the communication uses two rounds.



\begin{figure}[!t]
\begin{tcolorbox}[colback = white!5!white, left = 0mm, right = 0mm, top = 0mm, bottom = 0mm,
title = {Functionality $\mathcal{F}_{\mathsf{OIS}}$}]
\begin{algorithmic}[1]
	\PaInput $\langle \widetilde{\mathbf{M}}_1 \rangle^\mathsf{B}_{0}$, $\langle \mathbf{b} \rangle^\mathsf{B}_{0}$, $\widetilde{\mathbf{M}}_0$, and $\{(z_*, u_*) \ / \perp\}$
	\PbInput $\langle \widetilde{\mathbf{M}}_1 \rangle^\mathsf{B}_1$, $\langle \mathbf{b} \rangle^\mathsf{B}_1$, and $\{\perp \ / (z_*, u_*)\}$
	\Ensure $\langle \mathbf{b}^L \rangle^\mathsf{B}_l$ and $\langle \mathbf{b}^R \rangle^\mathsf{B}_l$ for `L'eft and `R'ight children.
\end{algorithmic}
\end{tcolorbox}

\begin{tcolorbox}[colback = white!5!white, left = -1mm, right = -1mm, top = 0mm, bottom = 0mm, title = Protocol $\varPi_{\mathsf{OIS}}$]
\begin{algorithmic}[1]
	\State $P_0$:
	draw a random mask $\langle \mathbf{b}_* \rangle^\mathsf{B}_0 \coloneqq \mathbf{r} \sample \{0, 1\}^{(1 + \epsilon)n}$.\label{alg:ois_naive_cpsi:L1}
	\If{$P_0$ owns $(z_*, u_*)$}
	\State $P_0$:
	set $\mathbf{b}_* = \widetilde{\mathbf{M}}_0[z_* B + u_*]$.
	\label{alg:ois_naive_cpsi:L3}
	\State $P_0$:
	share $\langle \mathbf{b}_* \rangle^\mathsf{B}_1 = \mathbf{b}_* \oplus \mathbf{r}$ with $P_1$.
	\label{alg:ois_naive_cpsi:L4}
	\Else \For{$i \in [(1 + \epsilon) \cdot n]$} \label{alg:ois_naive_cpsi:L5}
		\State $P_0$:
		$\widetilde{\mathbf{r}}[j] = \mathbf{r}[i]\oplus \langle \widetilde{\mathbf{M}}_1 \rangle^\mathsf{B}_0[j, i], \forall j \in [m \cdot B]$.
		\label{alg:ois_naive_cpsi:L6} 
		\State $P_0$ (as OT sender):
		run $\mathcal{F}_{\mathsf{OT}}(\widetilde{\mathbf{r}})$.
		\label{alg:ois_naive_cpsi:L7}
		\State $P_1$ (as OT receiver):
		$\mathbf{r}'[i] \gets \mathcal{F}_{\mathsf{OT}}(z_* B + u_*)$ \label{alg:ois_naive_cpsi:L8}
		\State $P_1$:
		$\langle \mathbf{b}_*[i] \rangle^\mathsf{B}_1 = \mathbf{r}'[i] \oplus \langle \widetilde{\mathbf{M}}_1 \rangle^\mathsf{B}_1[z_* B + u_*, i]$.
		\label{alg:ois_naive_cpsi:L9}
		\EndFor 
	\EndIf
	\State $P_l$ computes $\langle \mathbf{b}^L \rangle^\mathsf{B}_l = \mathcal{F}_{\mathsf{and}}(\langle \mathbf{b}_* \rangle^\mathsf{B}_l, \langle \mathbf{b} \rangle^\mathsf{B}_l)$ and \newline
	$\langle \mathbf{b}^R \rangle^\mathsf{B}_l = \langle \mathbf{b}^L \rangle^\mathsf{B}_l \oplus \langle \mathbf{b} \rangle^\mathsf{B}_l$.
	\label{alg:ois_naive_cpsi:L10}
	\Comment{Secure evaluation of Eq.~\eqref{sample_indicator_update}}
\end{algorithmic}
\end{tcolorbox}
\vspace{-5pt}
\caption{Ideal functionality and protocol of $\mathsf{OIS}$}
\label{alg:ois_naive_cpsi}
\vspace{-3pt}
\end{figure}

\begin{figure*}[t!]
\begin{tcolorbox}[colback = white!5!white, left = 0mm, right = 0mm, top = 0mm, bottom = 0mm, title = Base AnonGBDT $\mathcal{\varPi}^{\mathsf{Base}}_{\mathsf{GBDT}}(\mathbf{Input}_0^{\varPi}\mbox{, }\mathbf{Input}_1^{\varPi}\mbox{, }\mathbf{pp})$]
\begin{algorithmic}[1]
	\PaInput Identifier lists $\mathbf{ID}_0$, raw feature matrix $\mathbf{X}_0$, and labels $\mathbf{y}$ (Initially, $\tilde{\mathbf{y}} = \mathbf{0}$).
	\PbInput Identifier lists $\mathbf{ID}_1$ and raw feature matrix $\mathbf{X}_1$, and $\underline{\sk}_1$ for $\mathcal{F}_{\mathsf{A2H}}$, $\sk_1$ for $\mathcal{F}_{\mathsf{H2A}}$.
	\PP $\mathbf{pp} = \{T, D, B\}$, and lifting key $\mathsf{LK}_{\underline{\sk}_1\rightarrow\sk_1}$ for $\mathsf{KeySwitch}$.
	\Ensure $\mathsf{Output}^\Pi_0$ for $P_0$ and $\mathsf{Output}^\Pi_1$ for $P_1$.
	\State $P_l$ discretizes $\mathbf{X}_l$ to $\mathbf{M}_l \in \{0, 1\}^{m \cdot B \times n}$, which is aggregated to $\widetilde{\mathbf{M}}_l$.
	\Comment{Data pre-processing as in {Squirrel}}
	\State {\color{blue}$P_l$:
	$\langle \mathbf{b}^{(1)} \rangle_l^\mathsf{B}, \langle \mathbf{y} \rangle_l^\mathsf{A}, \langle \mathbf{M}_1 \rangle_l^\mathsf{B}, \langle \widetilde{\mathbf{M}}_1 \rangle_l^\mathsf{B} \leftarrow 
	\mathcal{F}_{\mathsf{CPSI}}(\mathbf{ID}_0; \mathbf{ID}_1)$, with $P_0$'s payload $\mathbf{y}$ and $P_1$'s payloads $(\mathbf{M}_1, \widetilde{\mathbf{M}}_1)$.}
	\For{ tree $t \in \mathcal{T}_t$}
	\State $P_l$ computes gradients $\langle \mathbf{g} \rangle^\mathsf{A}_l, \langle \mathbf{h} \rangle^\mathsf{A}_l$ using {\color{blue}{$\mathcal{F}_{\mathsf{sigmoid}}(\langle \tilde{\mathbf{y}} \rangle^\mathsf{A}_l)$}} and $\mathcal{F}_{\mathsf{mul}}$ based on Eq.~\eqref{eq:gradients_plaintext}.
	\Comment{Per-sample gradient}
	\label{alg:base_sys:gradient_comp}
	\State $P_l$ filters gradients $\langle \mathbf{g}^{(1)} \rangle^\mathsf{A}_l, \langle \mathbf{h}^{(1)} \rangle^\mathsf{A}_l$ at the root using $\mathcal{F}_{\mathsf{mux}}$ given $\langle \mathbf{b}^{(1)} \rangle_l^\mathsf{B}$ based on Eq.~\eqref{eq:filter_gradient}.
	\label{alg:base_sys:filter}
	\For{internal nodes $k \in \{1, 2, \ldots, 2^{D - 1} - 1\}$} 
		\If{$k$ is $1$ or even} \Comment{Per-bin gradient}
			\State \parbox[t]{\dimexpr\linewidth - \leftmargin - \labelsep - \labelwidth}{$P_l$ aggregates gradients over $\mathbf{M}_0$ (as in Squirrel):\\ 
			$\langle \widetilde{\mathbf{g}}^{(k, 0)} \rangle^\mathsf{A}_l || \langle \widetilde{\mathbf{h}}^{(k, 0)} \rangle^\mathsf{A}_l = \mathcal{F}_{\mathsf{BinMatVec}}( 
			\{ \langle \mathbf{g}^{(k)} \rangle^\mathsf{A}_0 || \langle \mathbf{h}^{(k)} \rangle^\mathsf{A}_0, \mathbf{M}_0\},
			\{\langle \mathbf{g}^{(k)} \rangle^\mathsf{A}_1 || \langle \mathbf{h}^{(k)} \rangle^\mathsf{A}_1, \sk_1 \})$.\strut}
		\State \parbox[t]{\dimexpr\linewidth - \leftmargin - \labelsep - \labelwidth}{\color{blue}{$P_l$ aggregates gradients over $\langle \mathbf{M}_1 \rangle^\mathsf{B}_l$ using our $\mathcal{F}_{\mathsf{mux}}$-based approach (see Eq.~\eqref{our_gradient_agg}):} \\
			\color{blue}{$\forall j \in [mB] \colon \langle \widetilde{\mathbf{g}}^{(k, 1)}[j] \rangle^\mathsf{A}_l (/\langle \widetilde{\mathbf{h}}^{(k, 1)}[j] \rangle^\mathsf{A}_l) = \sum_i \mathcal{F}_{\mathsf{mux}}(\langle \mathbf{M}_1 [j, i] \rangle^\mathsf{B}_l, \langle \mathbf{g}^{(k)} [i] \rangle^\mathsf{A}_l (/\langle \mathbf{h}^{(k)} [i] \rangle^\mathsf{A}_l))$
			.}\strut}
			\State $ P_l$ locally concatenates the shares as $\langle \widetilde{\mathbf{g}}^{(k)} \rangle^\mathsf{A}_l = \langle \widetilde{\mathbf{g}}^{(k, 0)} \rangle^\mathsf{A}_l \concat \langle \widetilde{\mathbf{g}}^{(k, 1)} \rangle^\mathsf{A}_l$, $\langle \widetilde{\mathbf{h}}^{(k)} \rangle^\mathsf{A}_l = \langle \widetilde{\mathbf{h}}^{(k, 0)} \rangle^\mathsf{A}_l \concat \langle \widetilde{\mathbf{h}}^{(k, 1)} \rangle^\mathsf{A}_l$.
			\Else~$P_l$ gets $\langle \widetilde{\mathbf{g}}^{(k)} \rangle^\mathsf{A}_l = \langle \widetilde{\mathbf{g}}^{(k / 2)} \rangle^\mathsf{A}_l - \langle \widetilde{\mathbf{g}}^{(k - 1)} \rangle^\mathsf{A}_l, \langle \widetilde{\mathbf{h}}^{(k)} \rangle^\mathsf{A}_l = \langle \widetilde{\mathbf{h}}^{(k / 2)} \rangle^\mathsf{A}_l - \langle \widetilde{\mathbf{h}}^{(k - 1)} \rangle^\mathsf{A}_l$ (as in {Squirrel}).
		\EndIf 
		\State $P_l$ computes the splitting scores $\langle \mathbf{L}_{\mathsf{sp}}^{(k)} \rangle^\mathsf{A}_l$ using $\mathcal{F}_{\mathsf{mul}}$ and $\mathcal{F}_{\mathsf{div}}$ for all $(z, u)$ based on Eq.~\eqref{eq:split_gain_plaintext}.
		\State $P_l$ gets $\langle z_*^{(k)} \rangle^\mathsf{A}_l, \langle u_*^{(k)} \rangle^\mathsf{A}_l \gets \mathcal{F}_{\mathsf{argmax}}(\langle \mathbf{L}_{\mathsf{sp}}^{(k)} \rangle^\mathsf{A}_l)$.
		\Comment{Best split selection}
		\State Open a bit $c \gets \mathcal{F}_{\mathsf{greater}}(\langle z_*^{(k)} \rangle^\mathsf{A}_l, m - 1)$ and reveal $(z_*^{(k)}, u_*^{(k)})$ to $P_c$, who writes it to $\mathsf{Output}_c^{\Pi}$.
	\State {\color{blue}$P_l$ updates $\langle \mathbf{b}^{(2k)} \rangle^\mathsf{B}_l, \langle \mathbf{b}^{(2k + 1)} \rangle^\mathsf{B}_l \gets \mathcal{F}_{\mathsf{OIS}}(\{ \langle \widetilde{\mathbf{M}}_1 \rangle^\mathsf{B}_{0}, \langle \mathbf{b}^{(k)} \rangle^\mathsf{B}_{0}, \widetilde{\mathbf{M}}_0\};
	\{\langle \widetilde{\mathbf{M}}_1 \rangle^\mathsf{B}_1, \langle \mathbf{b}^{(k)} \rangle^\mathsf{B}_1\})$ with $P_c$'s $(z_*^{(k)}, u_*^{(k)})$}
		\State $P_l$ updates left-child gradients:
		$\langle \mathbf{g}^{(2k)} \rangle^\mathsf{A}_l = \mathcal{F}_{\mathsf{mux}}( \langle \mathbf{b}^{(2k)} \rangle^\mathsf{B}_l, \langle \mathbf{g}^{(k)} \rangle^\mathsf{A}_l), \langle \mathbf{h}^{(2k)} \rangle^\mathsf{A}_l = \mathcal{F}_{\mathsf{mux}}( \langle \mathbf{b}^{(2k)} \rangle^\mathsf{B}_l, \langle \mathbf{h}^{(k)} \rangle^\mathsf{A}_l)$.
		\State $P_l$ updates right-child gradients:
		$\langle \mathbf{g}^{(2k + 1)} \rangle^\mathsf{A}_l = \langle \mathbf{g}^{(k)} \rangle^\mathsf{A}_l - \langle \mathbf{g}^{(2k)} \rangle^\mathsf{A}_l, \langle \mathbf{h}^{(2k + 1)} \rangle^\mathsf{A}_l = \langle \mathbf{h}^{(k)} \rangle^\mathsf{A}_l - \langle \mathbf{h}^{(2k)} \rangle^\mathsf{A}_l$.
	\EndFor
		\State $P_l$ computes leaf weights $\langle \mathbf{w} \rangle^\mathsf{A}_l$ using $\mathcal{F}_{\mathsf{div}}$ based on Eq.~\eqref{eq:weight_plaintext} and writes them in $\mathsf{Output}_l^{\Pi}$.
		\State $P_l$ updates $\langle \tilde{\mathbf{y}}[i] \rangle^\mathsf{A}_l$ by accumulating the result of $\mathcal{F}_{\mathsf{mux}}(\langle \mathbf{b}^{(k)}[i] \rangle^\mathsf{B}_l, \langle \mathbf{w}[k] \rangle^\mathsf{A}_l)$ ($\forall i$) over all leaves $k$.
	\EndFor
\end{algorithmic}
\end{tcolorbox}
\vspace{-5pt}
\caption{Base \sys (with {\color{blue}$\mathcal{F}_{\mathsf{mux}}$-based gradient histogram aggregation} and {\color{blue}new $\mathcal{F}_{\mathsf{OIS}}$ for sample indicator updates})}
\label{alg:base_sys}
\vspace{-10pt}
\end{figure*}

\subsubsection{Oblivious Indicator Synchronization}
\label{sect:base_OIS}
Sample indicators are updated from the chosen split via the exclusive-prefix matrix $\widetilde{\mathbf{M}}$.
As a circuit-PSI receiver, $P_{0}$ knows which column of $\widetilde{\mathbf{M}}_0$\footnote{%
It is resized to $mB \times (1 + \varepsilon)n$ after circuit-PSI cuckoo hashing.}
matches the correct indicator index.
The sender $P_{1}$ does not;
it needs to carry $\widetilde{\mathbf{M}}_{1}$ as a payload.

%

Because both $\widetilde{\mathbf{M}}_{1}$ and indicators
$\mathbf{b}^{(k)}$ are XOR-shared, the plaintext (local) update Eq.~\eqref{sample_indicator_update} is inapplicable.
To this end, we introduce an OT-based \emph{oblivious indicator synchronization} (OIS)
protocol, whose functionality
$\mathcal{F}_{\mathsf{OIS}}$ appears in
Figure~\ref{alg:ois_naive_cpsi}:
If $P_0$ owns $(z_*, u_*)$, it secret-shares $\mathbf{b}_* = \widetilde{\mathbf{M}}_0[z_* B + u_*]$ with $P_1$ (Lines~\ref{alg:ois_naive_cpsi:L1}--\ref{alg:ois_naive_cpsi:L4}).
Otherwise, $P_0$ (as OT sender) masks its share of $\widetilde{\mathbf{M}}_1$ and, for \emph{each} row, the two parties execute a $1$-out-of-$mB$ OT such that $P_{1}$ (as OT receiver) obtains correct $\langle \mathbf{b}_{*} \rangle_1^{\mathsf{B}}$ given its private choice $(z_*, u_*)$.
In either case, the left-child indicator is then computed with secure AND $\mathcal{F}_{\mathsf{and}}$, and the right-child indicator by bitwise XOR~(Section~\ref{sect:base_design_revisit}).

\smallskip
\noindent
\textbf{Complexity}:
If $P_0$ holds the best split, it sends $\mathcal{O}(n\ell)$ bits and the round is $1$;
if $P_1$ holds the best split, they invoke $n$ $1$-out-of-$mB$ $\mathcal{F}_{\mathsf{OT}}$ (each sends $\mathcal{O}(mB\ell)$ bits) with total communication round of $\mathcal{O}(\log{n})$), and $n$ $\mathcal{F}_{\mathsf{and}}$ (each sends $O(\ell)$ bits) with total communication round of $\mathcal{O}(\log{n} + 2)$.

\subsubsection{Improved Sigmoid Approximation}
\label{sect:our_sigmoid}
We replace the integral-derived Fourier coefficients with least-squares fitting to estimate $a_{j \in [0, J]}$, allowing us to reduce the order to $J = 3$ while achieving better accuracy and efficiency.
We retain $\mu = 5.6$ and evaluate only with $\mathcal{F}_{\mathsf{mul}}$, using an implementation optimized for trigonometric inputs within $[-1, 1]$;
see Appendix~\ref{apdx:sigmoid_mul}.
The resulting average/maximum errors are $0.002$/$0.015$ (vs.\ $0.022$ in Squirrel);
see Appendix~\ref{apdx:fourier}.

We also streamline secure comparisons with $\pm \mu$.
Under fixed-point representation with precision $2^f$, the two thresholds share their $f - \delta$ least-significant bits;
hence, within $\mathcal{F}_{\mathsf{greater}}$, we reuse $\mathcal{F}_{\mathsf{OT}}$ and $\mathcal{F}_{\mathsf{and}}$ on the shared bits.
Overall, our new sigmoid costs \textbf{33\%} of runtime and \textbf{16\%} of communication of the {Squirrel} variant.

Figure~\ref{alg:base_sys} presents the base \sys workflow, with differences from {Squirrel} highlighted in blue.
Two substitutions enable anonymous training under circuit-PSI alignment: 
i) gradient-histogram aggregation at the sender via $\mathcal{F}_{\mathsf{mux}}$ over payload-shared $\mathbf{M}_{1}$ (Section~\ref{sect:hist_agg});
ii) indicator synchronization via $\mathcal{F}_{\mathsf{OIS}}$ (Section~\ref{sect:base_OIS}), ensuring XOR-shared indicators remain aligned and usable by both parties.


\section{\texorpdfstring{$\sys$}{AnonGBDT}: OT-light, Symmetric Alignment}
The base design conceals intersection membership, cardinality, and alignment information throughout training,~offering \emph{stronger} privacy than prior work~\cite{pvldb/WuCXCO20,cikm/FangZT0YWWZZ21,uss/LuHZWH23}.
Its communication remains substantial: 
i) histogram construction via $\mathcal{F}_{\mathsf{mux}}$
requires~$O(mnB)$ many $1$-out-of-$2$ OTs, and
ii) $\mathcal{F}_{\mathsf{OIS}}$ runs one $1$-out-of-$mB$ per row over $(1 + \varepsilon) n$ rows (with $\varepsilon$ the cuckoo-hash parameter).
The histogram term dominates.
We eliminate it via a dual-circuit-PSI framework.

\subsection{Our \emph{Dual}-circuit-PSI Framework}
\label{sec:dual_cpsi}
Circuit-PSI is \emph{asymmetric}: only the receiver learns the (bucketed) input-output correspondence.
We ``symmetrize'' this by running two circuit-PSI instances \emph{in parallel} with swapped roles, so each party is a receiver once and learns the correspondence for exactly one instance.

This allows both parties to keep $\mathbf{M}_0$ and $\mathbf{M}_1$ locally (no payload for $\mathbf{M}_1$), and to aggregate gradients on~non-shared matrices solely with $\mathcal{F}_{\mathsf{BinMatVec}}$ (as in {Squirrel}), removing the many calls to $\mathcal{F}_{\mathsf{mux}}$.
We further accelerate $\mathcal{F}_{\mathsf{BinMatVec}}$ by replacing its subroutines, $\mathsf{PackLWEs}$-then-$\mathsf{LWEDimLift}$, with our $\mathsf{FastPackLWEs}$ (see Section~\ref{sec:fast_pack_rlwe}).

Figure~\ref{fig:advacned_sys_oveivew} overviews the dual-circuit-PSI.
Superscripts $(k, 0)$ and $(k, 1)$ respectively denote the first and second instances at node~$k$.
In instance~$0$, $P_{0}$ (receiver) inputs $\mathbf{ID}_{0}$ with payload $\mathbf{y}$, yielding $\langle \mathbf{b}^{(1, 0)} \rangle^\mathsf{B}$ and
$\langle \mathbf{y}^{0} \rangle^{\mathsf{A}}$.
Roles are reversed in instance~$1$, producing $\langle \mathbf{b}^{(1, 1)} \rangle^\mathsf{B}$ and
$\langle \mathbf{y}^{1} \rangle^{\mathsf{A}}$.

Each instance applies its own cuckoo permutation, so \emph{order-dependent} steps must be done per instance, \eg, gradient computation uses both $\langle \mathbf{y}^{0} \rangle^{\mathsf{A}}$ and $\langle \mathbf{y}^{1} \rangle^{\mathsf{A}}$, whereas (order-independent) aggregation runs once locally over $\mathbf{M}_0$ or $\mathbf{M}_1$.

\begin{figure*}[!t]
	\centering
	\includegraphics[width = 0.82\linewidth]{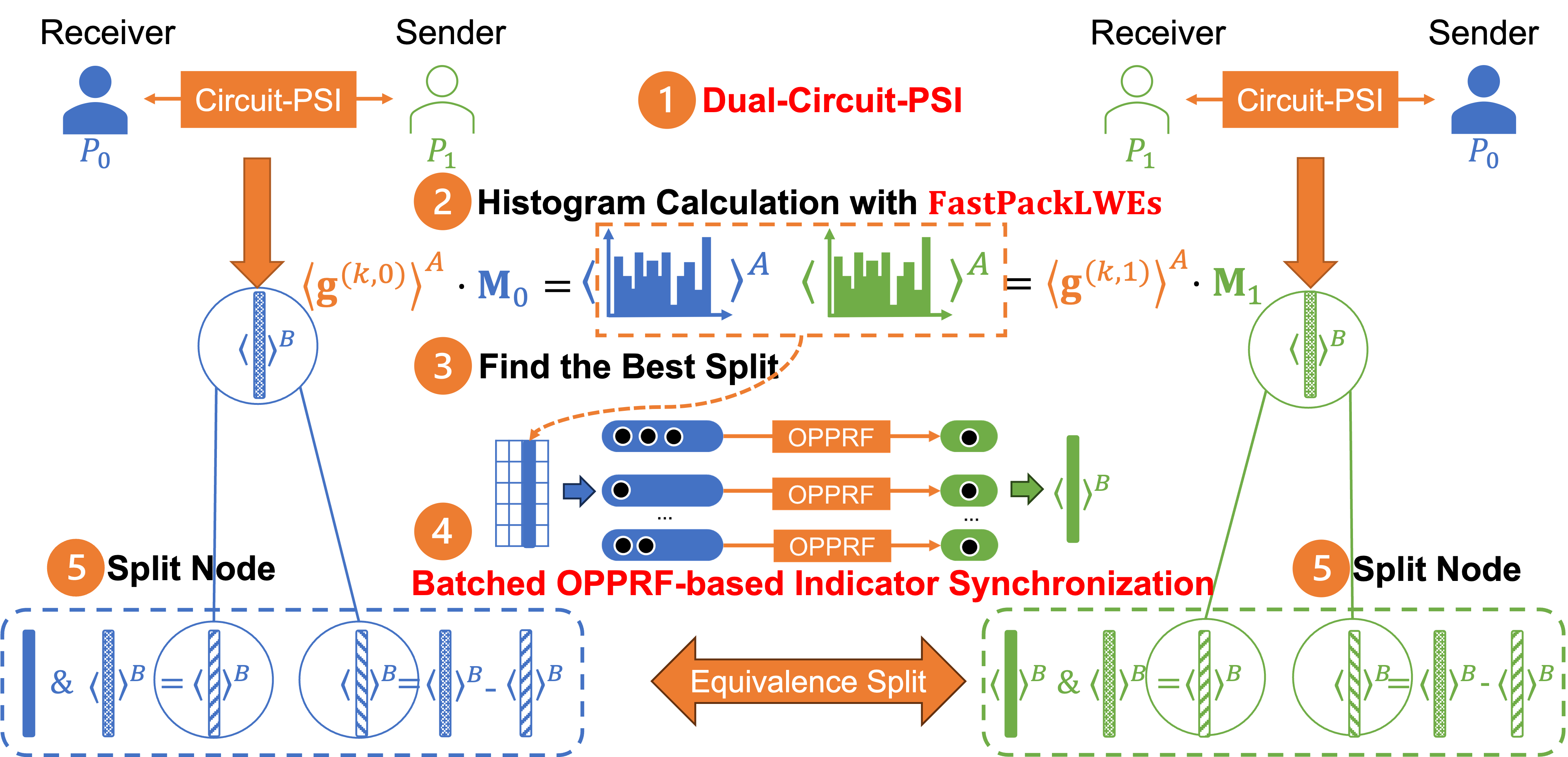}
	\vspace{-5pt}
	\caption{Overview of $\sysfinal$
	(with our key contributions highlighted in {\color{red}red})}
	\label{fig:advacned_sys_oveivew}
\vspace{-10pt}
\end{figure*}




\smallskip
\noindent \textbf{Base OIS limitation}:
With dual-circuit-PSI, updating the two XOR-shared indicators $\mathbf{b}^{(k, 0)}$ and $\mathbf{b}^{(k, 1)}$ is non-trivial.
If the split owner $P_c$ is the \emph{receiver} for an instance, it can update as in base OIS (Lines~\ref{alg:ois_naive_cpsi:L1}--\ref{alg:ois_naive_cpsi:L4} and~\ref{alg:ois_naive_cpsi:L10}, Figure~\ref{alg:ois_naive_cpsi}).
If it is the sender for the other instance, it lacks the share of $\widetilde{\mathbf{M}}_{1 - c}$ (no longer a payload of $P_{1 - c}$), so base OIS (Lines~\ref{alg:ois_naive_cpsi:L5}--\ref{alg:ois_naive_cpsi:L10}) fails.
We resolve this with a new OPPRF-based OIS (Section~\ref{sec:opprf_indicator_sync}), which supports level-wise \emph{batching} for efficiency.

With these changes and engineering optimizations (Section~\ref{sect:adv_optimizations}), 
$\sysfinal$ significantly reduces runtime and communication versus the base one (Section~\ref{sect:exp}).

\begin{figure}[!t]
\label{fig:pack_rlwe}
	\centering
	\includegraphics[width = 0.9 \linewidth]{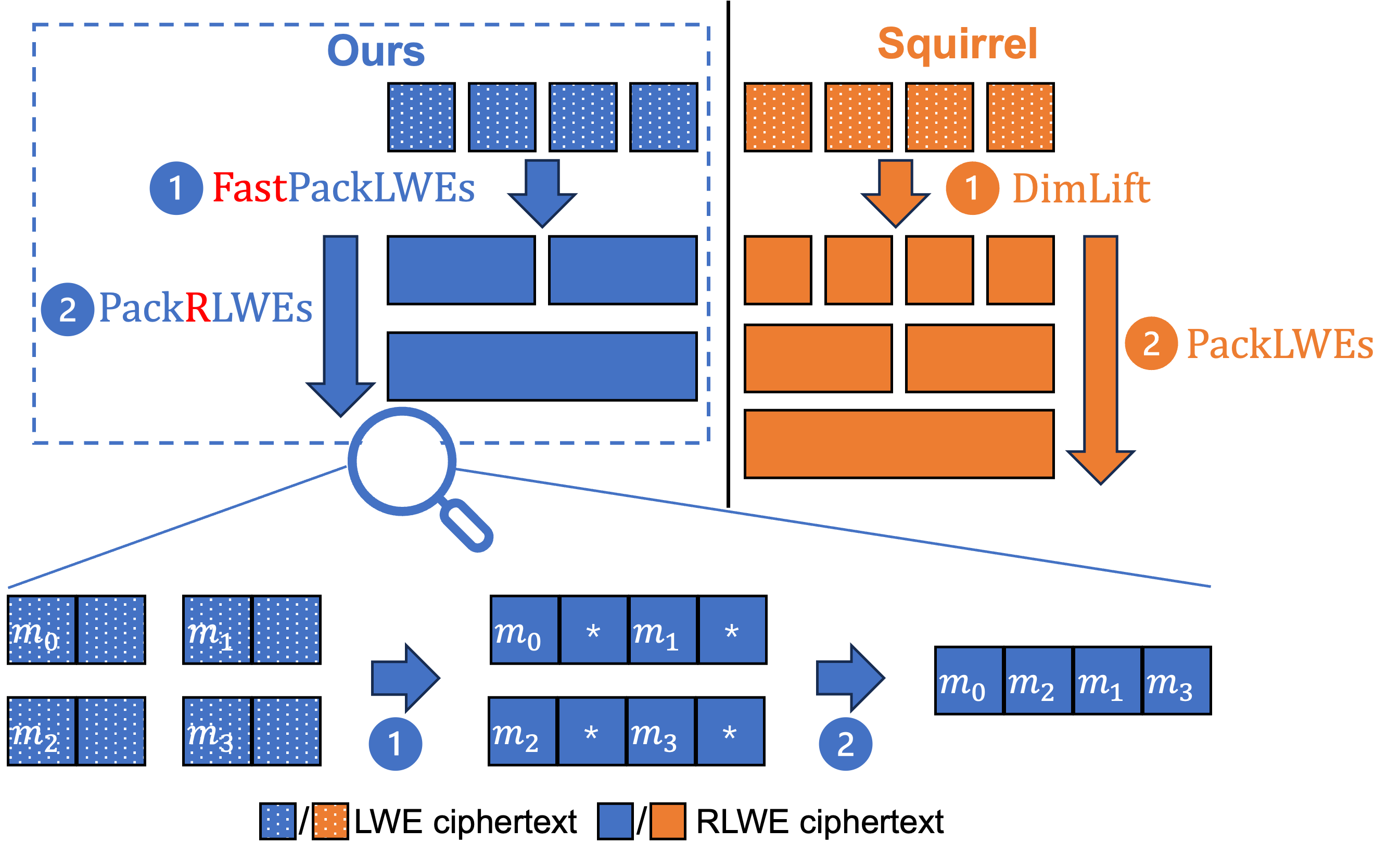}
	\caption{Example of $\mathsf{FastPackLWEs}$ ($\underline{n} = 4$, $N = 4$, and $\underline{N} = 2$):
	The adjacent LWE ciphertexts are first packed into an RLWE ciphertext via Algorithm~\ref{alg:fast_pack}.
	The RLWE ciphertexts are then recursively packed into a single RLWE ciphertext via Algorithm~\ref{alg:pack_rlwe}.}
 \label{fig:fast_pack}
\vspace{-5pt}
\end{figure}

\subsection{Faster LWE Ciphertext Packing}
\label{sec:fast_pack_rlwe}

The secure matrix-vector multiplication~$\mathcal{F}_{\mathsf{BinMatVec}}$ uses RLWE encryption for computational efficiency and outputs \underline{$n$} LWE ciphertexts, which incur large communication because $\mathcal{F}_{\mathsf{H2A}}$ must convert these LWE ciphertextsto shares.
Communication can be cut by packing several LWE ciphertexts with a smaller dimension ($N$) into a single RLWE ciphertext with a larger dimension ($\underline{N}$).
{Squirrel} does so in two stages:
$\mathsf{LWEDimLift}$ (lift dimension $\underline{N}\! \to \!N$) followed by $\mathsf{PackLWEs}$ (recursive merging).


We replace this flow with
$\mathsf{FastPackLWEs}$ followed by
$\mathsf{Pack\textcolor{blue}{R}LWEs}$
(Algorithms~\ref{alg:fast_pack} and~\ref{alg:pack_rlwe}).
Figure~\ref{alg:binmatvec} shows the adapted version.
We also modify $\mathcal{F}_{\mathsf{BinMatVec}}$ from Squirrel with changes boxed (Line~\ref{alg:binmatvec:L7}).

Let $\{\mathsf{LWE}_{\mathbf{s}}(m_{i})\}_{0 \le i < \underline{n}}$,
$\underline{n} = 2^{\tau}$, be encrypted under secret
$\mathbf{s}$ (the vector representation of 
$\mathcal{F}_{\mathsf{A2H}}$'s $\underline{\sk}$).
For each adjacent pair
\(
\mathsf{LWE}_{\mathbf{s}}(m_{2j}) = (\mathbf{a}_{1}, b_{1}),\;
\mathsf{LWE}_{\mathbf{s}}(m_{2j + 1}) = (\mathbf{a}_{2}, b_{2})
\),
we can manipulate them as a new RLWE ciphertext pair:
\begin{align}
	\widehat{a} = \widehat{a}_1 + \widehat{a}_2 \cdot X^{N / 2}, 
	\;\;
	\widehat{b} = b_1 + b_2 \cdot X^{N / 2},
\label{eq:fast_pack}
\end{align}
under a new secret
$\widehat{s} = \underline{\sk}[0] - \sum_{k = 1}^{\underline{N} - 1}\underline{\sk}[k]\cdot X^{N - k}$, where $\widehat{a}_i = \sum_{k = 0}^{\underline{N} - 1}\mathbf{a}_i[k] \cdot X^k$ ($i \in \{1, 2\}$).

The $0$- and $(N / 2)$-th coefficients, when decrypting $(\widehat{a}, \widehat{b})$ by $\widehat{s}$, are exactly $m_{2j}$ and $m_{2j + 1}$, \ie, the two messages are packed within one RLWE ciphertext without lifting their dimensions individually.
As $\widehat{s}$ is insecure (half its coefficients are zero), we immediately switch it to a~valid~RLWE secret key $\sk$ with one
$\mathsf{KeySwitch}$ call (in Line~\ref{alg:fast_pack:L3}).
By merging all adjacent ones, we thus obtain $\underline{n} / 2$ RLWE ciphertexts.

To integrate this idea, we remove the first recursion of $\mathsf{PackLWEs}$, leading to $\mathsf{Pack\textcolor{blue}{R}LWEs}$
(Algorithm~\ref{alg:pack_rlwe}).
Also, at Line~\ref{alg:pack_rlwe:L5}, 
we change $t$ from $N / 2^{\tau}$ to $N / 2^{\tau + 1}$.

\smallskip
\noindent \textbf{Complexity}:
$\mathsf{FastPackLWEs}$ processes $\underline{n} / 2$ pairs, cutting the cost of
$\mathsf{LWEDimLift}$~\cite{uss/LuHZWH23} in half.
$\mathsf{Pack\textcolor{blue}{R}LWEs}$ now requires only $\tau - 1$ recursions, again halving the workload compared with $\mathsf{PackLWEs}$~\cite{uss/LuHZWH23}.
Consequently, the overall ciphertext-packing time of
$\mathcal{F}_{\mathsf{BinMatVec}}$ is reduced by \textbf{50\%}.

\begin{figure}[htb]
\begin{tcolorbox}[colback = white!5!white, left = 0mm, right = 0mm, top = 0mm, bottom = 0mm, title = Functionality $\mathcal{F}_{\mathsf{BinMatVec}}$]
\begin{algorithmic}[1]
	\PcInput gradients share $\langle \mathbf{g} \rangle^\mathsf{B}_{c}$, binary feature representation $\mathbf{M}_c \in \{0, 1\}^{B}$
	\PdInput gradients share $\langle \mathbf{g} \rangle^\mathsf{A}_{1 - c}$, RLWE $\underline{\sk}_{1 - c}$ for $\mathcal{F}_{\mathsf{A2H}}$, and $\sk_{1 - c}$ for $\mathcal{F}_{\mathsf{H2A}}$
	\PP $\mathsf{KeySwitch}$ lifting key $\mathsf{LK}$
	\Ensure $\langle \mathbf{M}_c \cdot \mathbf{g} \rangle^\mathsf{A}_l$ to $P_l$.
\end{algorithmic}
\end{tcolorbox}

\begin{tcolorbox}[colback = white!5!white, left = 0mm, right = 0mm, top = 0mm, bottom = 0mm, title = Protocol $\varPi_{\mathsf{BinMatVec}}$]
	\begin{algorithmic}[1]
	\State Jointly run $\mathcal{F}_{\mathsf{A2H}}$, $P_{1 - c}$ inputs $\langle \mathbf{g} \rangle^\mathsf{A}_{1 - c}$, $\underline{sk}_{1 - c}$, and $P_{c}$ inputs $\langle \mathbf{g} \rangle^\mathsf{B}_{c}$;
	$P_c$ gets $\llbracket \mathbf{g} \rrbracket$, and $P_{1 - c}$ gets nothing.
		\State $P_c$ initializes a size-$(B \cdot m)$ array of LWE encryption of $0$, denoted by $\mathbf{ct}[j]$ for $j \in [B \cdot m]$.
		\For{all position $(j, i)$ such that $\mathbf{M}[j, i] = 1$}
				\State \parbox[t]{\dimexpr\linewidth - \leftmargin}{
			$P_c$ extracts LWE $\bar{\mathbf{ct}}[j, i]$ from $\llbracket \mathbf{g} \rrbracket$ that decrypts to $\mathbf{g}[i]$ under $P_{1 - c}$'s secret key.
				\strut}
				\State $P_c$:
				$\mathbf{ct}[j] = \mathbf{ct}[j] + \bar{\mathbf{ct}}[j, i]$ \Comment{pick-then-sum}
		\EndFor 
		\State \fbox{\!$P_c$: $\llbracket \mathbf{\mathbf{M}_c \cdot \mathbf{g}} \rrbracket \!\leftarrow \mathsf{FastPackLWEs}(\{\mathbf{ct}[j]\}_{j \in [m \cdot B)}, \mathsf{LK})$\!\label{alg:binmatvec:L7}} 
	\State Jointly invoke $\mathcal{F}_{\mathsf{H2A}}$, $P_c$ inputs $\llbracket \mathbf{\mathbf{M}_c \cdot \mathbf{g}} \rrbracket$ while
	$P_{1 - c}$ inputs $\sk_{1 - c}$.
		Both parties return the resulting arithmetic share $\langle \mathbf{M}_c \cdot \mathbf{g} \rangle^\mathsf{A}$.
	\end{algorithmic}
\end{tcolorbox}
\vspace{-5pt}
\caption{Functionality and protocol of $\mathsf{BinMatVec}$}
\vspace{-3pt}
\label{alg:binmatvec}
\end{figure}

\begin{algorithm}[t!]
\caption{$\mathsf{FastPackLWEs}$ (Fast LWE Packing)}
\label{alg:fast_pack}
\begin{algorithmic}[1]
	\Require LWE ciphertexts $\{\mathsf{LWE}_{\mathbf{s}}(m_i)\}_{0 \leq i< \underline{n}}$ (for~$\underline{n} = 2^{\tau} < N$) and lifting key $\mathsf{LK}_{\underline{\sk} \rightarrow \sk}$.
	\Ensure An RLWE ciphertext $\widehat{ct}$.
	\For{$i \in [0, \underline{n})$ with step-size as $2$}
			\State \parbox[t]{\dimexpr\linewidth - \leftmargin}{Build an RLWE ciphertext $\widehat{ct^*}_j = (\widehat{a}_j, \widehat{b}_j)$, where $j = i / 2$ and $\widehat{a}_j, \widehat{b}_j$ are given as Eq.~\eqref{eq:fast_pack}.
			\strut} 
			\State \parbox[t]{\dimexpr\linewidth - \leftmargin}{$\widehat{ct}_j = \mathsf{KeySwitch}_{\underline{\sk} \rightarrow \sk}(\widehat{ct^*}_j)$ based on $\mathsf{LK}_{\underline{\sk} \rightarrow \sk}$.
			\label{alg:fast_pack:L3}
			\strut} 
	\EndFor
	\State \textbf{return} $\widehat{ct} \gets \mathsf{Pack\textcolor{blue}{R}LWEs}(\{\widehat{ct}_{j}\})$.
\end{algorithmic}
\end{algorithm}

\begin{algorithm}[t!]
\caption{$\mathsf{Pack\textcolor{blue}{R}LWEs}$ (RLWE Ciphertext Packing)}
\label{alg:pack_rlwe}
\begin{algorithmic}[1]
	\Require RLWE ciphertexts $\{\widehat{ct}_j\}$ for some $j \in [2^\tau]$.
	\Ensure An RLWE ciphertext $\widehat{ct}$.
	\If {$\tau = 0$}
	\State \textbf{return} $\widehat{ct} \gets \widehat{ct}_0$.
	\EndIf
	\State $\widehat{ct}_{\even} \leftarrow \mathsf{Pack\textcolor{blue}{R}LWEs}(\{\widehat{ct}_{2j}|j \in [0, 2^{\tau - 1} - 1]\})$.
	\State $\widehat{ct}_{\odd} \leftarrow \mathsf{Pack\textcolor{blue}{R}LWEs}(\{\widehat{ct}_{2j + 1}|j \in[0, 2^{\tau - 1} - 1]\})$.
	\State \parbox[t]{\dimexpr\linewidth - \leftmargin}{\textbf{return} $\widehat{ct} \gets (\widehat{ct}_{\even} + X^{t} \cdot \widehat{ct}_{\odd}) + \mathsf{EvalAuto}(\widehat{ct}_{\even} - X^{t} \cdot \widehat{ct}_{\odd}, N + t)$, where $t = N / 2^{\tau + 1}$.
	\strut} \label{alg:pack_rlwe:L5}
\end{algorithmic}
\end{algorithm}

\subsection{OPPRF-based OIS}
\label{sec:opprf_indicator_sync}

\begin{figure}[t!]
\begin{tcolorbox}[colback = white!5!white, left = 0mm, right = 0mm, top = 0mm, bottom = 0mm, title = Functionality $\mathcal{F}_{\mathsf{OOIS}}$]
\begin{algorithmic}[1]
	\PcInput $\widetilde{\mathbf{M}}_c$, $\langle \mathbf{b}^{(0)} \rangle^\mathsf{B}_c$, $\langle \mathbf{b}^{(1)} \rangle^\mathsf{B}_c$, $(z_*, u_*)$, and
	\newline
	simple hash table $\mathbf{T}_{\mathsf{sh}}^{1 - c}$.
	\PdInput $\langle \mathbf{b}^{(0)} \rangle^\mathsf{B}_{1 - c}$, $\langle \mathbf{b}^{(1)} \rangle^\mathsf{B}_{1 - c}$, and \newline
	cuckoo hash table $\mathbf{T}_{\mathsf{ch}}^{1 - c}$.
	\Ensure $\langle \mathbf{b}^{(L, 0)}\rangle^\mathsf{B}, \langle \mathbf{b}^{(R, 0)}\rangle^\mathsf{B}, \langle \mathbf{b}^{(L, 1)}\rangle^\mathsf{B}, \langle \mathbf{b}^{(R, 1)}\rangle^\mathsf{B}$ for the `L'eft and `R'ight children of the two instances
\end{algorithmic}
\end{tcolorbox}

\begin{tcolorbox}[colback = white!5!white, left = 0mm, right = 0mm, top = 0mm, bottom = 0mm, title = Protocol $\varPi_{\mathsf{OOIS}}$]
\begin{algorithmic}[1]
	\State $P_c$ gets a random mask $ \langle \mathbf{b}_*\rangle^\mathsf{B}_c \coloneqq \mathbf{r}^c \sample \{0, 1\}^{(1 + \epsilon)n}$.\label{alg:oois:L1}
	\State $P_c$ sets $\mathbf{b}_*^{(c)} = \widetilde{\mathbf{M}}_c[z_*B + u_*]$.
	\State $P_c$ shares $\langle \mathbf{b}_*^{(c)}\rangle^\mathsf{B}_{1 - c} \gets \mathbf{b}_*^{(c)} \oplus \mathbf{r}^c$ with $P_{1 - c}$.
	\State $P_l$ derives $\langle \mathbf{b}^{(L, c)} \rangle^\mathsf{B}_l = \mathcal{F}_{\mathsf{and}}(\langle \mathbf{b}^{(c)}_*\rangle^\mathsf{B}_l, \langle \mathbf{b}^{(c)} \rangle^\mathsf{B}_l)$ \newline and
	$\langle \mathbf{b}^{(R, c)} \rangle^\mathsf{B}_l = \langle \mathbf{b}^{(L, c)}\rangle^\mathsf{B}_l \oplus \langle \mathbf{b}^{(c)} \rangle^\mathsf{B}_l$.
		\label{alg:oois:L4}
	\State $P_c$ ``reorders'' $\mathbf{b}_*^{(c)}$ to $\mathbf{b}'_*$ to align with $\mathbf{ID}_c$.\label{alg:oois:L5}
	\State $P_c$ samples $\langle \mathbf{b}^{(1 - c)}_* \rangle^\mathsf{B}_c = \mathbf{r}^{1 - c} \sample \{0, 1\}^{(1 + \epsilon)n}$.
	\State $\forall j \in [0, e), \forall i \in [0, (1 + \epsilon)n)$:
	\newline $P_c$ creates a table
	$\mathbf{T} = \{(\mathbf{T}_{\mathsf{sh}}^{1 - c}[i, j], \langle \mathbf{b}^{1 - c}_*[i] \rangle^B_{1 - c}\}$, \newline 
	where $\langle \mathbf{b}^{(1 - c)}_*[i] \rangle^B_{1 - c} = \mathbf{b}'_*[i']\oplus \mathbf{r}^{1 - c}[i]$ and \newline $\mathbf{ID}_{c}[i']$ is placed in $\mathbf{T}_{\mathsf{sh}}^{1 - c}[i, j]$ via simple hashing.
	\State $\langle \mathbf{b}^{(1 - c)}_* \rangle^\mathsf{B}_{1 - c}\leftarrow \mathcal{F}_{\mathsf{OPPRF}}(\{\mathbf{T}\}, \{\mathbf{T}_{\mathsf{ch}}^{1 - c}\})$.
	\State $P_l$ runs $\langle \mathbf{b}^{(L, 1 - c)} \rangle^\mathsf{B}_l = \mathcal{F}_{\mathsf{and}}(\langle \mathbf{b}^{(1 - c)}_*\rangle^\mathsf{B}_l, \langle \mathbf{b}^{(1 - c)} \rangle^\mathsf{B}_l)$.
	\State $P_l$ runs~$\langle \mathbf{b}^{(R, 1 - c)} \rangle^\mathsf{B}_l = \langle \mathbf{b}^{(L, 1 - c)}\rangle^\mathsf{B}_l \oplus \langle \mathbf{b}^{(1 - c)} \rangle^\mathsf{B}_l$.\label{alg:oois:L10}
\end{algorithmic}
\end{tcolorbox}
\vspace{-5pt}
\caption{Ideal functionality and protocol of $\mathsf{OOIS}$}
\label{alg:oois}
\vspace{-3pt}
\end{figure}

The dual framework is \emph{symmetric} across parties, but each instance remains asymmetric.
Assume $P_{c}$ owns the best split $(z_*, u_*)$.
Figure~\ref{alg:oois} shows our new OPPRF-based OIS (OOIS).
If $P_c$ is the \emph{receiver}, it sets
$\mathbf{b}^{(c)}_* = \widetilde{\mathbf{M}}_c[z_*B + u_*]$ secret-shares it with $P_{1 - c}$, from which the children's
$\langle \mathbf{b}^{(L, c)} \rangle^\mathsf{B}_l, 
\langle \mathbf{b}^{(R, c)} \rangle^\mathsf{B}_l$ are derived (Lines~\ref{alg:oois:L1}--\ref{alg:oois:L4}) as in base OIS.

If $P_c$ is the \emph{sender}, synchronizing the children's indicator shares given $\mathbf{b}^{(c)}_*$ is tricky.
Recall that circuit-PSI facilitates the inclusion of payloads through two batch OPPRFs.
The first identifies intersecting items via PSM testing.
The second distinguishes sender-side payloads of items within each bucket and selects the ``correct'' one for computation:
If $\mathbf{T}_{\mathsf{ch}}[i]$ is $j$-th item in $\mathbf{T}_{\mathsf{sh}}[i]$, then the programmed OPPRF output at the receiver is a share of $\mathbf{T}_{\mathsf{sh}}[i, j]$'s payload.

The remaining OOIS steps (Lines~\ref{alg:oois:L5}--\ref{alg:oois:L10}) build atop this concept:
Firstly, $P_{c}$ reorders $\mathbf{b}^{(c)}_*$ from its cuckoo-hash order to the raw order in $\mathbf{ID}_{c}$, producing a temporary indicator $\mathbf{b}'_*$ for the other instance.
It then samples a random vector $\mathbf{r}^{1 - c}$ and forms XOR share of $\mathbf{b}'_*$.
Each sample $\mathbf{ID}_c[i']$ resides in some bucket $\mathbf{T}^{1 - c}_{\mathsf{sh}}[i, j]$ under the other instance's $e$ simple hashes.
This lets $P_c$ build a key-value table $\mathbf{T}$ 
with those items as keys and the corresponding shares of $\mathbf{b}^{(1 - c)}_*$ as values
(derived from $\mathbf{b}'_*$ and $\mathbf{r}^{1 - c}$).
A single call to $\mathcal{F}_{\mathsf{OPPRF}}$ supplies $P_{1 - c}$ with the complementary shares of $\mathbf{b}^{(1 - c)}_*$.
Finally, both parties obtain the left- and right-child indicators via $\mathcal{F}_{\mathsf{and}}$ and string XOR as in Line~\ref{alg:oois:L4}.

\smallskip
\noindent 
\textbf{Abstraction}:
Conceptually, our OPPRF call expects the same ideal functionality as oblivious transfer for a sparse array~\cite{popets/ZhaoC17} (with the same index domain and range).
Sender-chosen values are returned on a sparse set of indices and pseudorandom fillers appear elsewhere.
We keep the OPPRF interface because our pipeline already uses batch OPPRF at two call sites.
This preserves batching in $\mathcal{F}_{\mathsf{BOIS}}$.

\smallskip
\noindent 
\textbf{Level batching}:
All internal nodes share $\mathbf{ID}_{c}$ (as the hash table's key set), so we pack all node indicators per level into one OPPRF value.
The \emph{batched} variant (BOIS, Figure~\ref{alg:bois} in Appendix~\ref{apdx:bois}) reduces the $O(2^{D})$ cost to $O(D)$ for depth $D$.

\smallskip
\noindent 
\textbf{Complexity}:
For a batch of nodes, it first calls one $\mathcal{F}_{\mathsf{OPPRF}}$, costs $O(n \ell)$ bits with $\mathcal{O}(1)$ rounds.
Then, it calls $2n$ $\mathcal{F}_{\mathsf{and}}$ (each sends $O(\ell)$ bits) with $\mathcal{O}(1)$ rounds.

\begin{theorem}
\label{theorem:bois}
$\varPi_{\mathsf{BOIS}}$ realizing $\mathcal{F}_{\mathsf{BOIS}}$ in Figure~\ref{alg:bois} is a secure protocol under $\mathcal{F}_{\mathsf{OPPRF}}$ and $\mathcal{F}_{\mathsf{and}}$ hybrids.
\end{theorem}
We defer the proof to Appendix~\ref{sec:appendix_security_proof} due to space limits.

\subsection{Further Optimizations}
\label{sect:adv_optimizations}
\noindent $\bullet$~\textbf{High-precision score computation}:
With fixed-point precision $f = 20$, intermediate shares (\eg, $\mathcal{F}_{\mathsf{mul}}$ in Eq.~\eqref{eq:split_gain_plaintext}) can exceed $\mathbb{Z}_{2^{\ell}}$ when $n$ is large.
We prescale inputs to $[1, 2]$ and rescale the result, preventing overflow at negligible cost.

\smallskip
\noindent $\bullet$~\textbf{Gradient packing in $\mathcal{F}_{\mathsf{mux}}$}:
Filtering the (secret-shared) gradients $\mathbf{g}$ and~$\mathbf{h}$ at $k$-th node requires the same sample indicator $\mathbf{b}^{(k)}$ (Eq.~\eqref{eq:filter_gradient}).
Real values (\eg, every entry of $\mathbf{g}$) are encoded as $\ell$-bit integers by fixed-point representation.
As the security parameter is $\lambda = 2 \ell = 128$, we can halve the number of OT used in $\mathcal{F}_{\mathsf{mux}}$ by entry-wise packing of $\mathbf{g}$ and $\mathbf{h}$, compared to transferring them individually.

\smallskip
\noindent
$\bullet$~\textbf{Gradient packing in $\mathcal{F}_{\mathsf{A2H}}$}:
In $\mathcal{F}_{\mathsf{A2H}}$, we can encode gradients $\langle \mathbf{g} \rangle^\mathsf{A}$ and $\langle \mathbf{h} \rangle^\mathsf{A}$ within one RLWE ciphertext, \ie, 
$\widehat{a}_l = \sum_{i = 0}^{\underline{N}/2 - 1} \langle \mathbf{g}[i] \rangle^\mathsf{A}_l X^i + \sum_{i = \underline{N}/2}^{\underline{N} - 1} \langle \mathbf{h}[i] \rangle^\mathsf{A}_l X^i$ (see Section~\ref{sect:crypto_tools}).
After pick-then-sum, two LWE ciphertexts (of dimension $\underline{N}$) can be directly packed in one RLWE ciphertext (of dimension $\underline{N}$) without $\mathsf{KeySwitch}$, thereby reducing one recursion of $\mathsf{FastPackLWEs}$.

\smallskip
\noindent $\bullet$~\textbf{Skipping $\mathcal{F}_{\mathsf{A2H}}$ for the split owner}:
If party $P_c$ holds the best split at node $k$, it locally allocates the RLWE-encrypted $\mathbf{g}^{(k, c)}$ and $\mathbf{h}^{(k, c)}$ to bins at node $2k$ and aggregates them, eliminating the costly share-to-ciphertext conversion.

\smallskip
\noindent $\bullet$~\textbf{Dummy-free aggregation}:
In $\mathcal{F}_{\mathsf{BinMatVec}}$, ciphertexts follow the cuckoo-hash order, including the $\varepsilon n$ dummies.
By realigning $\mathbf{M}_l$ and the ciphertext list with the original row order of $\mathbf{X}_l$, each party sums only genuine samples, avoiding dummy-related overhead.

\smallskip
\noindent
$\bullet$~\textbf{Lightweight oblivious $\mathsf{argmax}$}:
Pairwise comparisons use OTs, outperforming GC variants~\cite{ccs/RatheeR0CGR020}.
Because the split-index space is public, the final selection is obtained directly with $\mathcal{F}_{\mathsf{OT}}$, removing the extra $\mathcal{F}_{\mathsf{mux}}$ step of Squirrel~\cite{uss/LuHZWH23}.


\smallskip
\noindent
$\bullet$~\textbf{Erasing superfluous coefficients in $\mathsf{FastPackLWEs}$}:
To prevent leakage, the prior work~\cite{acns/ChenDKS21} zeroizes useless coefficients between $X^{i * N / 2^{\tau}}$ and $X^{(i + 1) * N / 2^{\tau}}$ for $i \in [2^{\tau} - 1]$ by the trace function,~which recursively evaluates $\mathsf{EvalAuto}$ on the ciphertexts after $\mathsf{PackLWEs}$.
Instead, we just add randomness to nullify these coefficients (without $\mathsf{EvalAuto}$ but sharing a similar purpose) to improve efficiency.


\begin{figure*}[t!]
\begin{tcolorbox}[colback = white!5!white, left = 0mm, right = 0mm, top = 0mm, bottom = 0mm, 
title = {$\sys$ 
$\mathcal{\varPi}^{\mathsf{OTSA}}_{\mathsf{GBDT}}(\mathbf{Input}_0^{\varPi_0}$,
$\mathbf{Input}_1^{\varPi_0}$,
$\mathbf{pp})$}]
\begin{algorithmic}[1]
	\PaInput $\mathbf{ID}_0$, feature $\mathbf{X}_0$ and label $\mathbf{y}$.
	\PbInput $\mathbf{ID}_1$, feature $\mathbf{X}_1$.
	\PP $\mathbf{pp} = \{D, B, T, (\underline{sk}_l, \underline{pk}_l)\textrm{~for }\mathcal{F}_{\mathsf{A2H}}, (sk_l, pk_l)\textrm{~for }\mathcal{F}_{\mathsf{H2A}}, \mathsf{KeySwitch}$ Lifting Key $\mathsf{LK}_{\underline{sk}_l \rightarrow \sk_l}\}$.
	\Ensure $\mathsf{Output}^\Pi_0$ for $P_0$ and $\mathsf{Output}^\Pi_1$ for $P_1$
	\State $P_l$ discretizes its raw data $\mathbf{X}_l$ to $\mathbf{M}_l \in \{0, 1\}^{m \cdot B \times n}$, which is aggregated to $\widetilde{\mathbf{M}}_l$.
	\label{alg:final_sys:L1}
	\State {\color{blue}For instance $i = 0$:
	$\langle \mathbf{b}^{(1, 0)} \rangle^\mathsf{B}_l, \langle \mathbf{y}^0 \rangle^\mathsf{A}_l \gets \mathcal{F}^0_{\mathsf{CPSI}}(\{\mathbf{ID}_0\, \mathbf{y}\};
	\{\mathbf{ID}_1\})$, with $P_0$'s $\mathbf{T}^0_{\mathsf{ch}}$ and $P_1$'s $\mathbf{T}^0_{\mathsf{sh}}$.}
	\label{alg:final_sys:L2}
	\State {\color{blue}For instance $i = 1$:
	$\langle \mathbf{b}^{(1, 1)} \rangle^\mathsf{B}_l, \langle \mathbf{y}^1 \rangle^\mathsf{A}_l \gets \mathcal{F}^1_{\mathsf{CPSI}}(\{\mathbf{ID}_1\};
	\{\mathbf{ID}_0, \mathbf{y}\})$, with $P_1$'s $\mathbf{T}^1_{\mathsf{ch}}$ and $P_0$'s $\mathbf{T}^1_{\mathsf{sh}}$.}
	\label{alg:final_sys:L3}
	\For{tree $t \in \mathcal{T}_t$ } \label{alg:final_sys:L4}
	\State {\color{blue}For $i \in \{0, 1\}$}:
	$P_l$ derives gradients $\langle \mathbf{g}^{(1, i)} \rangle^\mathsf{A}_l, \langle \mathbf{h}^{(1, i)}\rangle^\mathsf{A}_l$ at root using $\mathcal{F}_{\mathsf{sigmoid}} (\langle\tilde{\mathbf{y}}^i\rangle^\mathsf{A}_l)$ and $\mathcal{F}_{\mathsf{mul}}$;
	see Eq.~\eqref{eq:gradients_plaintext}.
	\label{alg:final_sys:L5}
	\State {\color{blue}For $i \in \{0, 1\}$}:
	$P_l$ filters $\langle \mathbf{g}^{(1, i)}\rangle^\mathsf{A}_l,
	\langle \mathbf{h}^{(1, i)}\rangle^\mathsf{A}_l$
	at root using $\mathcal{F}_{\mathsf{mux}}$ given $\langle \mathbf{b}^{(1, i)} \rangle^\mathsf{B}_l$;
	see Eq.~\eqref{eq:filter_gradient}.
	\label{alg:final_sys:L6}
	\For{level $d \in \{1, 2, \ldots, D - 1\}$ of the tree} \label{alg:final_sys:L7}
		\State {\color{blue}Both parties maintain two public sets $\mathbf{k}_0 \gets \emptyset, \mathbf{k}_1 \gets \emptyset$.}\label{alg:final_sys:L8}
		\Comment{Record node $k$'s best split is at $P_0$ or $P_1$}
		\For{the internal nodes $k \in \{2^{d - 1}, \ldots, 2^d - 1\}$ at level $d$}
		\Comment{Batch updates on level $d$ for BOIS}
		\label{alg:final_sys:L9}
		\If{$k$ is $1$ or even} \label{alg:final_sys:L10}
		\State \parbox[t]{\dimexpr\linewidth - \leftmargin - \labelsep - \labelwidth}{$P_l$ aggregates gradients over $\mathbf{M}_0$ (as in Squirrel) {\color{blue} with our new subroutine $\mathsf{FastPackLWEs}$}: \\ 
			$\langle \widetilde{\mathbf{g}}^{(k, 0)} \rangle^\mathsf{A}_l || \langle \widetilde{\mathbf{h}}^{(k, 0)} \rangle^\mathsf{A}_l = \mathcal{F}_{\mathsf{BinMatVec}}(
			\{ \langle \mathbf{g}^{(k,0)} \rangle^\mathsf{A}_0 || \langle \mathbf{h}^{(k,0)} \rangle^\mathsf{A}_0, \mathbf{M}_0\},
			\{\langle \mathbf{g}^{(k,0)} \rangle^\mathsf{A}_1 || \langle \mathbf{h}^{(k,0)} \rangle^\mathsf{A}_1, \sk_1 \})$.\strut} \label{alg:final_sys:L11}
		\State {\color{blue}Similarly, $P_l$ runs $\mathcal{F}_{\mathsf{BinMatVec}}$ with $\mathsf{FastPackLWEs}$ over $\mathbf{M}_1$ to get $\langle \tilde{\mathbf{g}}^{(k, 1)} \rangle^\mathsf{A}_l$ and
		$\langle \tilde{\mathbf{h}}^{(k, 1)} \rangle^\mathsf{A}_l$.}\label{alg:final_sys:L12}
		\State $ P_l$ locally concatenates the shares as $\langle \widetilde{\mathbf{g}}^{(k)} \rangle^\mathsf{A}_l = \langle \widetilde{\mathbf{g}}^{(k, 0)} \rangle^\mathsf{A}_l \concat \langle \widetilde{\mathbf{g}}^{(k, 1)} \rangle^\mathsf{A}_l$, $\langle \widetilde{\mathbf{h}}^{(k)} \rangle^\mathsf{A}_l = \langle \widetilde{\mathbf{h}}^{(k, 0)} \rangle^\mathsf{A}_l\concat \langle \widetilde{\mathbf{h}}^{(k, 1)} \rangle^\mathsf{A}_l$.\label{alg:final_sys:L13}
		\Else~$P_l$ calculates $\langle \tilde{\mathbf{g}}^{(k)} \rangle^\mathsf{A}_l = \langle \tilde{\mathbf{g}}^{(k / 2)} \rangle^\mathsf{A}_l - \langle \tilde{\mathbf{g}}^{(k - 1)} \rangle^\mathsf{A}_l, \langle \tilde{\mathbf{h}}^{(k)} \rangle^\mathsf{A}_l = \langle \tilde{\mathbf{h}}^{(k / 2)} \rangle^\mathsf{A}_l - \langle \tilde{\mathbf{h}}^{(k - 1)} \rangle^\mathsf{A}_l$.
			\label{alg:final_sys:L14}
		\EndIf 
			\State $P_l$ computes the splitting scores $\langle \mathbf{L}_{\mathsf{sp}}^{(k)} \rangle^\mathsf{A}_l$ using $\mathcal{F}_{\mathsf{mul}}$ and $\mathcal{F}_{\mathsf{div}}$ for all $(z, u)$ based on Eq.~\eqref{eq:split_gain_plaintext}.\label{alg:final_sys:L15}
		\State $P_l$ gets $\langle z_*^{(k)}\rangle^\mathsf{A}_l, \langle u_*^{(k)}\rangle^\mathsf{A}_l \gets \mathcal{F}_{\mathsf{argmax}}(\langle \mathbf{L}_{\mathsf{sp}} \rangle^\mathsf{A}_l)$.
		\label{alg:final_sys:L16}
		\State Open $c \gets \mathcal{F}_{\mathsf{greater}}(\langle z_*^{(k)} \rangle, m - 1)$ and reveal $(z_*^{(k)}, u_*^{(k)})$ to $P_c$, who writes it to $\mathsf{Output}_c^{\Pi}$ {\color{blue}and $k$ to $\mathbf{k}_c$.}
		\label{alg:final_sys:L17}
		\EndFor
		\State {\color{blue}\parbox[t]{\dimexpr\linewidth - \leftmargin - \labelsep - \labelwidth}{For $k \in \mathbf{k}_0$:
		$P_l$ runs 
		$\mathcal{F}_{\mathsf{BOIS}}(\{
		\langle \{ \mathbf{b}^{(k, 0)} \rangle^\mathsf{B}_0,
		\langle \mathbf{b}^{(k, 1)} \rangle^\mathsf{B}_0,
		(z_*^{(k)},
		u_*^{(k)})\},
		\widetilde{\mathbf{M}}_0, \mathbf{T}_{\mathsf{sh}}^1\};
		\{ 
		\{\langle \mathbf{b}^{(k, 0)} \rangle^\mathsf{B}_1,
		\langle \mathbf{b}^{(k, 1)} \rangle^\mathsf{B}_1\},
		\mathbf{T}_{\mathsf{ch}}^1\})$
		\newline
		to get
		$\{\langle \mathbf{b}^{(2k, i)} \rangle^\mathsf{B}_l, 
		\langle \mathbf{b}^{(2k + 1, i)} \rangle^\mathsf{B}_l \}_{i \in \{0, 1\}}$.\label{alg:final_sys:L18}} 
		\strut}
		\State {\color{blue}For $k \in \mathbf{k}_1$:
		$P_l$ runs $\mathcal{F}_{\mathsf{BOIS}}$ (now flip the role with $\widetilde{\mathbf{M}}_1, \mathbf{T}_{\mathsf{sh}}^0, \mathbf{T}_{\mathsf{ch}}^0$) to get $\{\langle \mathbf{b}^{(2k, i)} \rangle^\mathsf{B}_l, 
		\langle \mathbf{b}^{(2k + 1, i)} \rangle^\mathsf{B}_l \}_{i \in \{0, 1\}}$}.
		\label{alg:final_sys:L19}
		\State {\color{blue}For $i \in \{0, 1\}$}:
		update left-child:
		$\langle \mathbf{g}^{(2k, i)} \rangle^\mathsf{A}_l (/ \langle \mathbf{h}^{(2k, i)} \rangle^\mathsf{A}_l) = \mathcal{F}_{\mathsf{mux}}( \langle \mathbf{b}^{(2k, i)} \rangle^\mathsf{B}_l, \langle \mathbf{g}^{(k, i)} \rangle^\mathsf{A}_l (/ \langle \mathbf{h}^{(k, i)} \rangle^\mathsf{A}_l))$.\label{alg:final_sys:L20}
		\State {\color{blue}For $i \in \{0, 1\}$}:
		update right-child:
		$\langle \mathbf{g}^{(2k + 1, i)} \rangle^\mathsf{A}_l = \langle \mathbf{g}^{(k, i)} \rangle^\mathsf{A}_l - \langle \mathbf{g}^{(2k, i)} \rangle^\mathsf{A}_l, 
	\langle \mathbf{h}^{(2k + 1, i)} \rangle^\mathsf{A}_l = \langle \mathbf{h}^{(k, i)} \rangle^\mathsf{A}_l - \langle \mathbf{h}^{(2k, i)} \rangle^\mathsf{A}_l$.\label{alg:final_sys:L21}
	\EndFor
	\State $P_l$ computes leaf weights $\langle \mathbf{w} \rangle^\mathsf{A}_l$ using $\mathcal{F}_{\mathsf{div}}$ based on Eq.~\eqref{eq:weight_plaintext} and writes them in $\mathsf{Output}_l^{\Pi}$.\Comment{Order-independent}\label{alg:final_sys:L22}
		\State {\color{blue}For $i \in \{0, 1\}$}:
		$P_l$ updates $\langle \tilde{\mathbf{y}}^i[j] \rangle^\mathsf{A}_l$ by accumulating $\mathcal{F}_{\mathsf{mux}}(\langle \mathbf{b}^{(k, i)}[j] \rangle^\mathsf{B}_l, \langle \mathbf{w}[k] \rangle^\mathsf{A}_l)$ ($\forall j$) over all leaves $k$.\label{alg:final_sys:L23}
	\EndFor
\end{algorithmic}
\end{tcolorbox}
\vspace{-5pt}
\caption{$\sysfinal$ (with key differences to our base solution highlighted in {\color{blue}blue})}
\label{alg:final_sys}
\vspace{-5pt}
\end{figure*}

\subsection{Bringing Everything Together}
Figure~\ref{alg:final_sys} presents $\sysfinal$.
Two circuit‑PSI instances run in parallel with swapped roles.
For instance $i \in \{0, 1\}$, inputs are $\mathbf{ID}_0$ and $\mathbf{ID}_1$, and the only payload is $\mathbf{y}$, producing the root indicator shares $\langle \mathbf{b}^{(1, i)} \rangle^\mathsf{B}_l$ and label shares $ \langle \mathbf{y}^i \rangle^\mathsf{A}_l$.
Unlike our base scheme, $\mathbf{M}_i, \widetilde{\mathbf{M}}_i$ stay local.

For each tree, both parties derive and filter $\langle \mathbf{g}^{(1, i)} \rangle^\mathsf{A}_l$ and $\langle \mathbf{h}^{(1, i)}
\rangle^\mathsf{A}_l$ per instance 
 (Lines~\ref{alg:final_sys:L6}--\ref{alg:final_sys:L7}).
Nodes are processed level-wise using BOIS (Figure~\ref{alg:bois}).
At depth $d$, two public sets~$\mathbf{k}_0$, $ \mathbf{k}_1$ record the owner of each node's best split.

\smallskip\noindent
\textbf{Histogram construction}:
For roots and left children, both parties aggregate gradients with $\mathcal{F}_{\mathsf{BinMatVec}}$ over $\mathbf{M}_0$ and $\mathbf{M}_1$ (accelerated by $\mathsf{FastPackLWEs}$) and then concatenate locally.
For right children, histograms are obtained by subtracting the parent and left-child values.
Aggregation is order-independent, so it runs only once (not per instance).

\smallskip
\noindent
\textbf{Split selection}:
Split gains for all candidates are evaluated with $\mathcal{F}_{\mathsf{mul}}$ and $\mathcal{F}_{\mathsf{div}}$ (Line~\ref{alg:final_sys:L16}.
$\mathcal{F}_{\mathsf{argmax}}$ reveals $(z^{(k)}_*, u^{(k)}_*)$ to its owner $P_c$, which adds $k$ to $\mathbf{k}_c$.
Given $\mathbf{k}_0, \mathbf{k}_1$, $\mathcal{F}_{\mathsf{BOIS}}$ updates level-wise indicators for the next level's filtering.

\smallskip
\noindent
\textbf{Leaf weights}:
For each leaf~$k$, the weight is derived from whichever gradient share $\langle\mathbf{g}^{(k, i)} \rangle_l^\mathsf{A}$ has fewer samples (for order independence), and the predictions are accumulated.

Confining heavy computation to local data removes the OT-intensive histogram phase of the base design.
The trade-off is maintaining one ``logical'' tree per instance, increasing total communication to at most $2(1 + \varepsilon){\times}$ that of {Squirrel}.

\subsubsection{Complexity}
Table~\ref{tab:complexity_comparison_base_final} reports the per-component complexity for $\sys$.
$\sysfinal$ achieves its performance gains primarily due to:
i) In base gradient histogram aggregation (GHA), the dominant cost is $\mathcal{F}_{\mathsf{mux}}$ with complexity $\mathcal{O}(Bmnl)$, which is eliminated in $\sysfinal$, and
ii) Unlike the base $\mathcal{F}_{\mathsf{OIS}}$, which is invoked per node and scales linearly with $Bm$, its batch extension $\mathcal{F}_{\mathsf{BOIS}}$ is invoked per level.

\begin{table}[t]
	\small
	\centering
	\caption{Communication Complexity and rounds of each component of $\sys$ (Base versus Final)}
\setlength{\tabcolsep}{4pt}
\renewcommand{\arraystretch}{1.12}
\begin{tabularx}{\columnwidth}{@{}l >{\centering\arraybackslash}X c@{}}
\toprule
Component & Complexity & Round \\
\midrule
Dual-circuit-PSI & $\mathcal{O}(n \ell)$ & $\mathcal{O}(\log \ell) + 1$ \\
GHA (Base) & \makecell{$\mathcal{O}(2Bmn \ell + n\log Q$ \\ $+ 2Bm\log Q)$} & $\mathcal{O}(1)$ \\
GHA (Final) & $\mathcal{O}(2n\log Q + 4Bm\log Q)$ & $2$ \\
$\mathcal{F}_{\mathsf{OIS}}$ (Base) & $\mathcal{O}(Bmn \ell)$ per node & $\mathcal{O}(1)$ \\
$\mathcal{F}_{\mathsf{BOIS}}$ (Final) & $\mathcal{O}(3n \ell)$ per level & $\mathcal{O}(1)$ \\
\bottomrule
\end{tabularx}
\label{tab:complexity_comparison_base_final}
	\vspace{-3pt}
\end{table}

\begin{theorem}
\label{theorem:advaned_sys}
$\mathcal{\varPi}^{\mathsf{OTSA}}_{\mathsf{GBDT}}$ (Figure~\ref{alg:final_sys})
is anonymous GBDT training under 
$\mathcal{F}_{\mathsf{argmax}}$, 
$\mathcal{F}_{\mathsf{BOIS}}$, 
$\mathcal{F}_{\mathsf{BinMatVec}}$ (with $\mathcal{F}_{\mathsf{A2H}}$, $\mathcal{F}_{\mathsf{H2A}}$),
$\mathcal{F}_{\mathsf{CPSI}}$, 
$\mathcal{F}_{\mathsf{div}}$, 
$\mathcal{F}_{\mathsf{greater}}$,
$\mathcal{F}_{\mathsf{mul}}$, 
$\mathcal{F}_{\mathsf{mux}}$, 
$\mathcal{F}_{\mathsf{sigmoid}}$-hybrid and semantically-secure RLWE encryption.
\end{theorem}
We defer the proof to Appendix~\ref{sec:appendix_security_proof} due to space limits.

\subsubsection{Anonymous Batch Inference}
In inference~\cite{ndss/MaTZC21},
each party $P_l$ holds a vertically partitioned dataset (no labels), (plaintext) best splits $\mathbf{C}_l$, and weight shares $\langle \mathbf{w} \rangle^\mathsf{A}_l$.
One party finally learns the prediction probabilities $\mathbf{p}$ on the joint data.
Prior protocols often assume PSI-aligned inputs.

Evaluation using a secure multiplexer can proceed as follows:
i) local traversal: each $P_l$ computes a binary vector $\mathbf{b}^{(k, l)}$ (AND-based share $\langle \mathbf{b}^{(k)} \rangle^\mathsf{B}_l$ in Squirrel) by comparing features with local split criteria;
$\mathbf{b}^{(k, l)}[i] = 1$ indicates that sample~$i$ \emph{may} reach node $k$;
ii) conjunction: the true indicator is $\mathbf{b}^{(k)}= \mathbf{b}^{(k, 0)} \wedge \mathbf{b}^{(k, 1)}$, and the tree output is
\begin{align*}
\sum\nolimits_{k} \mathcal{F}_{\mathsf{mux}}(\langle \mathbf{b}^{(k)} \rangle^\mathsf{B}_l, \langle \mathbf{w}[k] \rangle^\mathsf{A}_l).
\end{align*}

To keep inference-time identifiers hidden, we can run circuit-PSI preceding the above procedure, which induces misaligned local indicators.
An OPPRF-based synchronization step then aligns indicators across parties:
$P_1$ transfers a share of $\mathbf{b}^{(k,1)}$ to
$P_{0}$, yielding aligned shares 
$\langle \mathbf{b}^{(k,1)} \rangle^\mathsf{B}_l$ and
$\langle \mathbf{b}^{(k,0)} \rangle^\mathsf{B}_l$ on which
$\mathcal{F}_{\mathsf{and}}$ applies.
The full algorithm is in Appendix~\ref{apdx:sec_inference}.

\subsection{Generalizability}
Our dual-circuit-PSI framework generalizes to two-party \emph{vertical} FL (VFL) models that exchange intermediate computations over aligned samples.
These models require a data alignment step that, if using standard PSI, leaks sensitive intersection membership.
Our framework provides a blueprint for \emph{anonymous} VFL by replacing this leaky pre-alignment.

\smallskip
\noindent 
\textbf{1) Anonymous Handshake}:
It establishes a ``symmetric private'' alignment in which each party, and only that party, learns its own cuckoo-hash permutation of the intersecting samples, so subsequent computations are performed on the same set of entities without revealing their identities.


\smallskip
\noindent 
\textbf{2) Synchronization and computation}:
$P_l$, with its own secret permutation, re-aligns its local information and continues the computation with the secret shared payload (\eg, completing a forward pass).
For data flowing in the reverse direction (\eg, encrypted/noised gradients), our $\mathcal{F}_{\mathsf{BOIS}}$ can be adapted to securely transfer the aligned information.

This workflow can be instantiated in various VFL architectures, \eg, federated transformers~\cite{nips/WuHDH24} and split-NN~\cite{jnca/GuptaR18}.


\section{Evaluation}
\label{sect:exp}
\noindent
\textbf{Parameters}:
The security parameters are $\lambda = 128, \kappa = 40$.
Our designs are implemented in C++ based on Cheetah~\cite{githubcheetah}.
We use the stash-less circuit-PSI~\cite{ccs/RaghuramanR22} with 3 hash functions and~$\epsilon = 1.3$,
and $\ell = 64$-bit arithmetic sharing with precision $f = 20$.
For GBDT, $\alpha = 0.001$ and $\gamma = 0$.
RLWE is realized using Cheon--Kim--Kim--Song~\cite{asiacrypt/CheonKKS17} of SEAL~4.1~\cite{sealcrypto} with $p = 2^\ell, N = 8192, \underline{N} = 4096, Q \approx 2^{109}$~\cite{uss/LuHZWH23}.

\smallskip
\noindent
\textbf{Testbed}:
All experiments run on two identical hosts, each with an $8$-core $2.5$GHz Intel Xeon Platinum 8255C CPU and $64$GB RAM under:
i) \textbf{LAN}: $1$Gbps bandwidth with a latency (ping time) ${\sim}2$ms, and
ii) \textbf{WAN}: we emulate a $200$Mbps link with $20$ms latency using the traffic control utility ($\texttt{tc}$).

\smallskip
\noindent
\textbf{Metrics}:
We measure the end-to-end (offline and online) runtime and the total communication costs.

\smallskip
\noindent
\textbf{Microbenchmarks}:
Table~\ref{tab:microbenchmark} shows the performance breakdown.
Interactive primitives, requiring message exchange between parties, include
the protocols for
$\mathsf{argmax}$
(taking $10$ inputs), 
$\mathsf{greater}$, 
$\mathsf{mux}$, 
$\mathsf{BOIS}$, 
$\mathsf{BinMatVec}$ 
(taking two $10^5 \time 128$ matrices), and our optimized 
$\mathsf{sigmoid}$ (Section~\ref{sect:our_sigmoid}).

For local primitives, we compare our $\mathsf{FastPackLWEs}$ with Squirrel's subroutine.
Combined with the optimizations (Section~\ref{sect:adv_optimizations}), 
we achieve a ${\sim}5 \times$ runtime speedup.

\begin{table}[!t]
	\small
	\centering
	\caption{Microbenchmarks of $\sys$, single-thread}

	\begin{tabular}{l|r|r|r|r}
		\hline 
		Protocol 
		& 
		\!\!instance\!\! &
		\begin{tabular}{@{}c@{}}Time (s)\\
		WAN
		\end{tabular} 
		& \begin{tabular}{@{}c@{}}Time (s)
		\\
		LAN \end{tabular}
		& \begin{tabular}{@{}c@{}}Comm. 
		\\(MB) \end{tabular}\\
		\hline 
		\hline 
		$\mathcal{F}_{\mathsf{argmax}}$
		\tablefootnote{Each run with $10$ inputs.} 
		& $10^5$ 
		& $4.9$ & $3.0$ & $13.4$\\
		$\mathcal{F}_{\mathsf{greater}}$
		& $10^6$
		& $22.2$ & $21.2$ & $81.9$\\
		$\mathcal{F}_{\mathsf{mux}}$
		& $10^6$
		& $12.4$ & $1.7$ & $17.5$\\
		$\mathcal{F}_{\mathsf{BOIS}}$
		& $10^6$
		& $12.3$ & $8.4$ & $68.6$\\
		$\mathcal{F}_{\mathsf{BinMatVec}}$
        \tablefootnote{Input vector size of $128$ vector size and matrix size of $128\times10^5$.}
		& $1$ 
		& $3.7$ & $1.8$ & $29.6$\\
		\hline
		$\mathcal{F}_{\mathsf{sigmoid}}$ (ours)
		& $10^5$\!\!
		& $\mathbf{13.4}$ & $\mathbf{4.8}$ & $\mathbf{38.4}$\\
		$\mathcal{F}_{\mathsf{sigmoid}}$~\cite{uss/LuHZWH23}
		& $10^6$
		& $31.5$ & $15.5$ & $235.8$\\
		\hline
	\end{tabular}

	\vspace{3pt}

	\begin{tabular}{l|c|c}
		\hline
		Local &
		\# of inputs & Time (s)\\
		\hline
		$\mathsf{FastPackLWEs}$ (ours) & 
		$128$ & $\mathbf{0.22}$\\
		$\mathsf{FastPackLWEs}$ (ours) & 
		$512$ & $\mathbf{0.87}$\\
		$\mathsf{LWEDimLift} + \mathsf{PackLWEs}$~\cite{uss/LuHZWH23} &
	$128$ & $1.02$\\
		$\mathsf{LWEDimLift} + \mathsf{PackLWEs}$~\cite{uss/LuHZWH23} &
		$512$ & $3.80$\\
		\hline
	\end{tabular}
	\label{tab:microbenchmark}
	\vspace{-3pt}
\end{table}

\subsection{Overhead of \texorpdfstring{$\sysfinal$}{AnonGBDT-OTSA}}
We ran benchmarks on synthetic datasets with varying $n$, $m$, and $D$, using a single thread.
Table~\ref{tab:eva_comp_base_final_performance} reports efficiency results.
For runtime, $\sysfinal$ is consistently faster by ${\sim}3 \times$ to $20\times$, and the gap widens as $n$, $m$, or $D$ grows.
For communication, our base \sys scales linearly in $n$ and $m$, eventually out-of-memory (OOM), due to the OT‑heavy $\mathcal{F}_{\mathsf{mux}}$-based gradient aggregation non‑batched $\mathcal{F}_{\mathsf{OIS}}$.
$\sysfinal$ removes these bottlenecks, achieving up to
$44\times$ lower traffic and avoiding OOM even at the largest dataset sizes tested.

\sys is also evaluated under unbalanced settings ($n_0\neq n_1$), as shown in Table~\ref{tab:eva_comp_base_final_unbalanced_performance}.
The base design benefits more than $\sysfinal$, with LAN speedups of $1.5{-}5 \times$ as the imbalance ratio ($\frac{n_0}{n_1}$) decreases.
This is because $\sysfinal$ must run $\mathcal{F}_{\mathsf{BOIS}}$, $\mathcal{F}_{\mathsf{sigmoid}}$ twice, processing to align with each party's inputs separately.


\begin{table}[t]
	\small
	\centering
	\caption{Efficiency of $\sys$
	($n = n_0 = n_1 = 10^5$, $m = 50$, $B = 16$, $D = 5$, single thread)}
	\begin{tabular}{l|rr|rr|rr}
		\hline 
		Parameters\!\!\! & \multicolumn{2}{c|}{WAN Time (s)} & \multicolumn{2}{c|}{LAN Time (s)} & \multicolumn{2}{c}{Comm. (GB)\!} \\
				& Base & Final & Base & Final & Base & \!\!\!Final\! \\
		\hline 
		$n = 100$K\!\!\! & \!\!$3873.5$ & $195.2$ & $567.0$ & $112.5$ & $26.1$ & \!\!\!$0.6$\! \\
		$n = 250$K\!\!\! & OOM & $346.5$ & OOM & $227.0$ & OOM & \!\!\!$1.1$\! \\
		$n = 500$K\!\!\! & OOM & $602.6$ & OOM & $422.2$ & OOM & \!\!\!$2.1$\! \\
		$n = 1$M	& OOM & \!\!\!\!$1120.9$ & OOM & $817.1$ & OOM & \!\!\!$4.0$\! \\
		\hline 
		$D = 3$	& \!\!$1077.0$ & $72.3$ & $142.8$ & $37.3$ & $6.6$ & \!\!\!$\!0.2$\! \\
		$D = 4$	& \!\!$1999.2$ & $117.1$ & $271.2$ & $64.5$ & $13.1$ & \!\!\!$0.4$\! \\
		$D = 6$	& \!\!$7462.7$ & $357.6$ & \!\!\!\!$1024.6$ & $203.6$ & $51.8$ & \!\!\!$0.9$\! \\
		\hline 
		$m = 10$	& $969.2$ & $91.2$ & $144.0$ & $36.9$ & $5.5$ & \!\!\!$0.4$\! \\
		$m = 25$	& \!\!$2061.5$ & $131.1$ & $303.1$ & $67.1$ & $13.2$ & \!\!\!$0.5$\! \\
		$m = 100$	& OOM & $334.0$ & OOM & $214.1$ & OOM & \!\!\!$0.7$\! \\
		\hline
	\end{tabular}
	\label{tab:eva_comp_base_final_performance}
    \vspace{-5pt}
\end{table}

\begin{table}[t]
	\small
	\centering
	\caption{Efficiency of $\sys$ for $n_0 \neq n_1$
	($n_1 = 10^5$, $m = 50$, $B = 16$, $D = 5$, single thread)}
	\begin{tabular}{l|rr|rr|rr}
		\hline 
		Size $n_0$ & \multicolumn{2}{c|}{WAN Time (s)} & \multicolumn{2}{c|}{LAN Time (s)} & \multicolumn{2}{c}{Comm. (GB)} \\
				& Base & Final & Base & Final & Base & Final \\
		\hline 
		$0.10n_1$ & $439.3$ & $98.5$ & $72.4$ & $49.1$ & $3.4$ & $0.34$ \\
		$0.20n_1$ & $780.4$ & $105.6$ & $121.3$ & $54.5$  & $5.6$ & $0.37$ \\
		$0.50n_1$ & $2156.3$ & $138.8$ & $267.4$ & $69.7$  & $15.2$ & $0.45$ \\
		$0.75n_1$	& $3047.2$ & $167.4$ & $413.8$ & $87.3$  & $21.3$ & $0.54$  \\
		$1.00n_1$	& $3873.5$ & $195.2$ & $567.0$ & $112.5$ & $26.1$ & $0.6$ \\
		\hline
	\end{tabular}
	\label{tab:eva_comp_base_final_unbalanced_performance}
    \vspace{-5pt}
\end{table}
\begin{table}[t]
	\small
	\centering
	\caption{Running time of $\sysfinal$ 
    and prior private (but ``leaky'') solutions ($T = 10$)}
	\begin{tabular}{c|c|c|r}
		\hline 
		Approach & Parameters & Settings & Time (s)\\ 
		\hline 
		\hline 
		Ours & \multirow{3}{*}{
		\begin{tabular}{@{}c@{}}
		$n = 5 \times 10^4$\\
		$D = 4$, $B = 8$\\
		$m_0 = 8$, $m_1 = 7$ \end{tabular}} & \multirow{3}{*}{
		\begin{tabular}{@{}c@{}} 
		LAN\\
		$6$ threads
		\end{tabular}} & $\mathbf{39.0}$\\
		\cite{uss/LuHZWH23} & & & $60.0$\\
		\cite{pvldb/WuCXCO20} & & & $1680.0$\\
		\hline
		Ours & \multirow{3}{*}{
		\begin{tabular}{@{}c@{}}
		$n = 2 \times 10^5$\\
		$D = 4$, $B = 8$\\
		$m_0 = 8$, $m_1 = 7$ \end{tabular}} & \multirow{3}{*}{
		\begin{tabular}{@{}c@{}}
		LAN\\ $6$ threads
		\end{tabular}} & $\mathbf{96.5}$\\
		\cite{uss/LuHZWH23} & & & $111.0$\\
		\cite{pvldb/WuCXCO20} & & & $448.0$\\
		\hline
		\hline 
		Ours & \multirow{3}{*}{
		\begin{tabular}{@{}c@{}}
		$n = 1.4 \times 10^5$\\
		$D = 5, B = 10$\\
		$m_0 = 7, m_1 = 16$ \end{tabular}} & \multirow{3}{*}{
		\begin{tabular}{@{}c@{}}
		LAN\\
		$32$ threads%
		\tablefootnote{%
		\label{footnote:tab:eva_comp_advance_exisiting_works}
		Pivot~\cite{pvldb/WuCXCO20} uses $32$ cores.
		Numbers of threads follow prior works~\cite{uss/LuHZWH23}.
		}
		\end{tabular}} & $139.1$\\
		\cite{uss/LuHZWH23} & & & $\mathbf{114.0}$\\
		\cite{cikm/FangZT0YWWZZ21} & & & $476.0$\\
		\hline
		Ours & \multirow{3}{*}{
		\begin{tabular}{@{}c@{}}
		$n = 1.4 \times 10^5$\\
		$D = 5, B = 10$\\
		$m_0 = 7, m_1 = 16$ \end{tabular}} & \multirow{3}{*}{\begin{tabular}{@{}c@{}}
		$100$Mbps\\
		$32$ threads%
        \textsuperscript{\ref{footnote:tab:eva_comp_advance_exisiting_works}}
		\end{tabular}} & $668.7$\\
		\cite{uss/LuHZWH23} & & & $\mathbf{400.0}$\\
		\cite{cikm/FangZT0YWWZZ21} & & & $1510.0$\\
		\hline
	\end{tabular}
	\label{tab:eva_comp_advance_exisiting_works}
    \vspace{-3pt}
\end{table}

\begin{figure*}[!t]
	\begin{minipage}[t]{.32\linewidth}
		\centering
		\includegraphics[width = \linewidth]{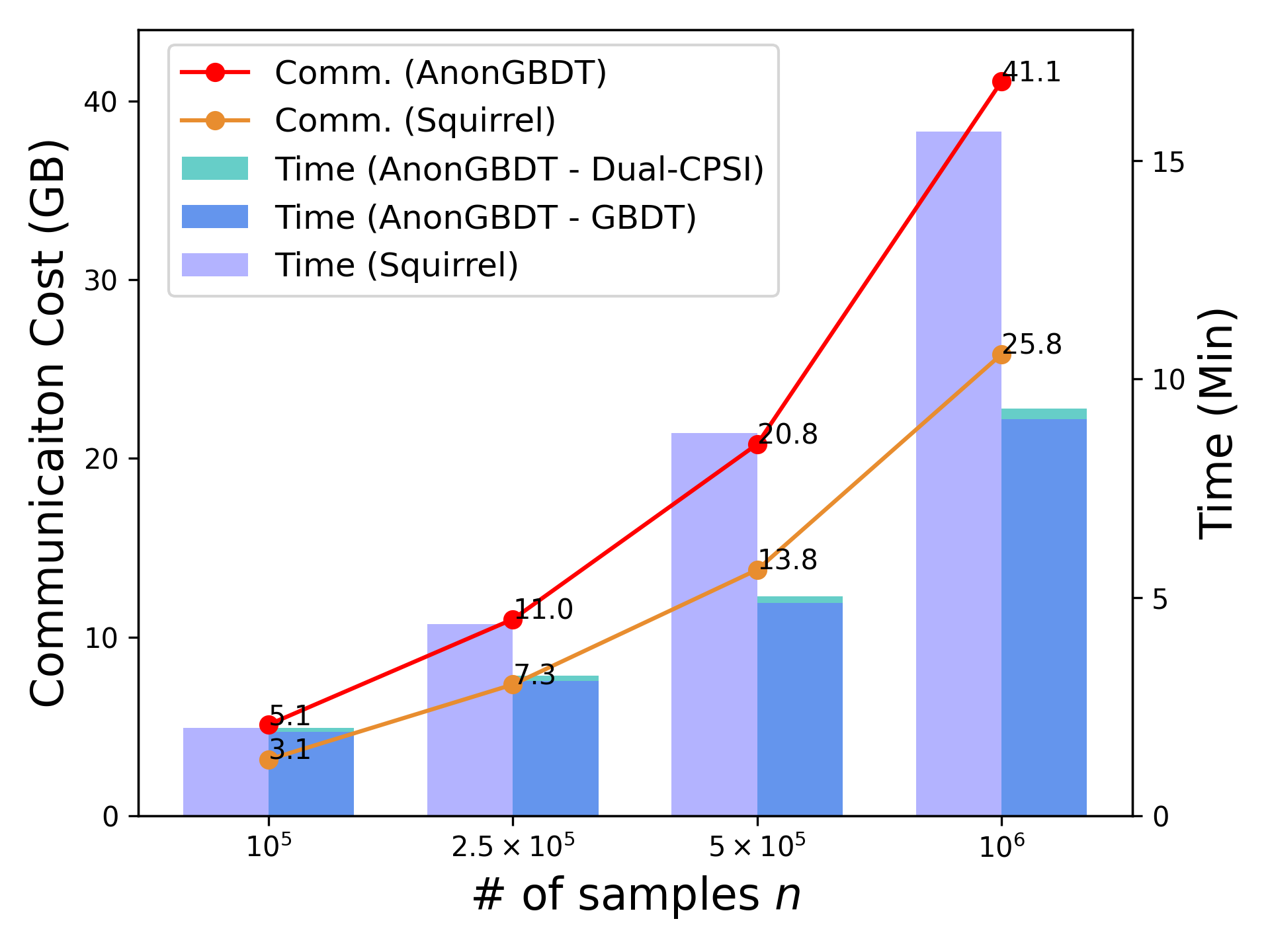}
		\subfloat{\footnotesize Time (LAN) vs. Sample size $n$}
		\label{fig:eva_scalability_sample_size_lan}
	\end{minipage}
	\hfill
	\begin{minipage}[t]{.32\linewidth}
		\centering
		\includegraphics[width = \linewidth]{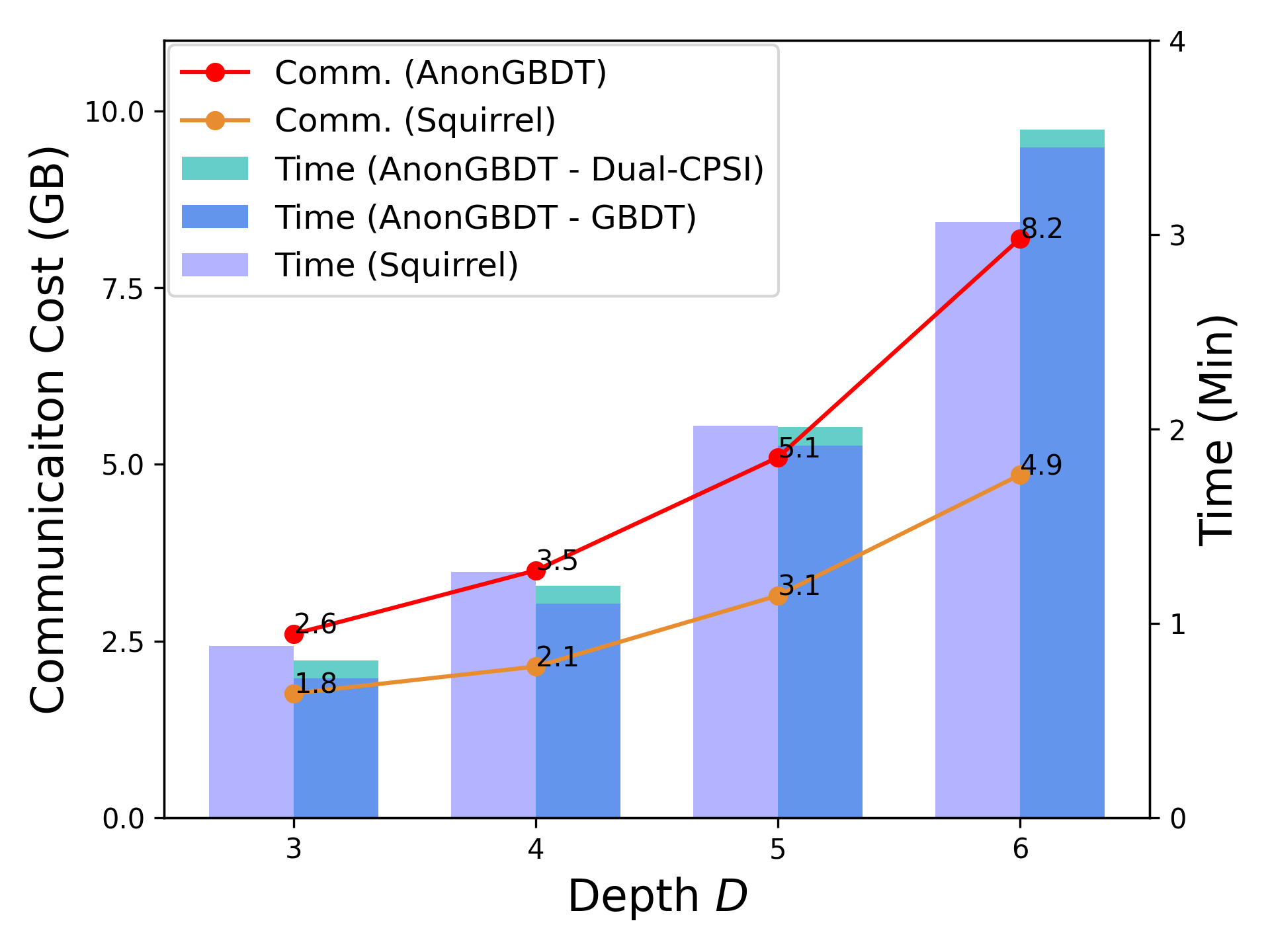}
		\subfloat{\footnotesize Time (LAN) vs. Depth $D$}
		\label{fig:eva_scalability_depth_lan}
	\end{minipage}
	\hfill
	\begin{minipage}[t]{.32\linewidth}
		\centering
		\includegraphics[width = \linewidth]{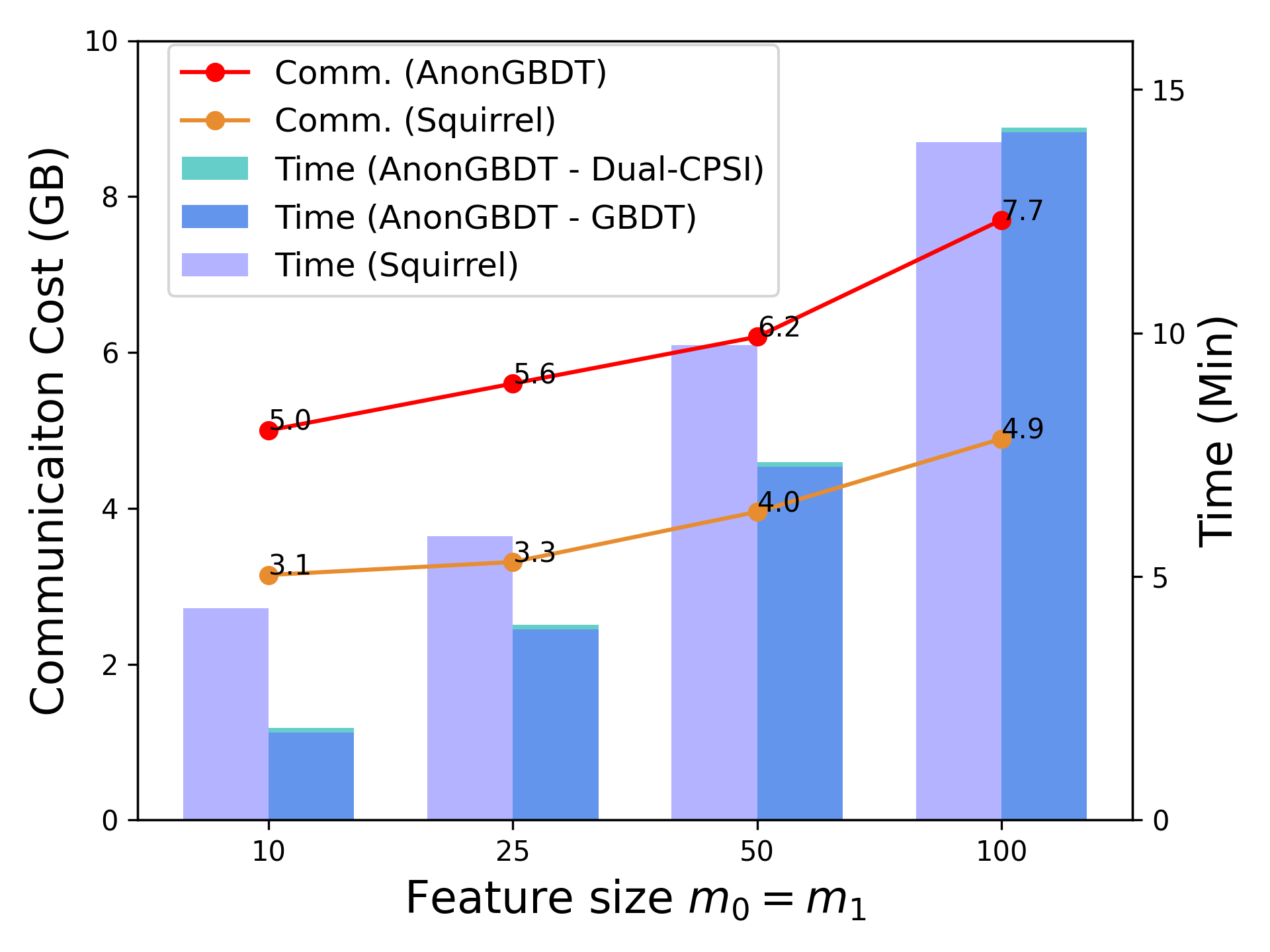}
		\subfloat{\footnotesize Time (LAN) vs. Feature size $m_0 = m_1$}
		\label{fig:eva_scalability_feature_size_lan}
	\end{minipage}
	\begin{minipage}[t]{.32\linewidth}
		\centering
		\includegraphics[width = \linewidth]{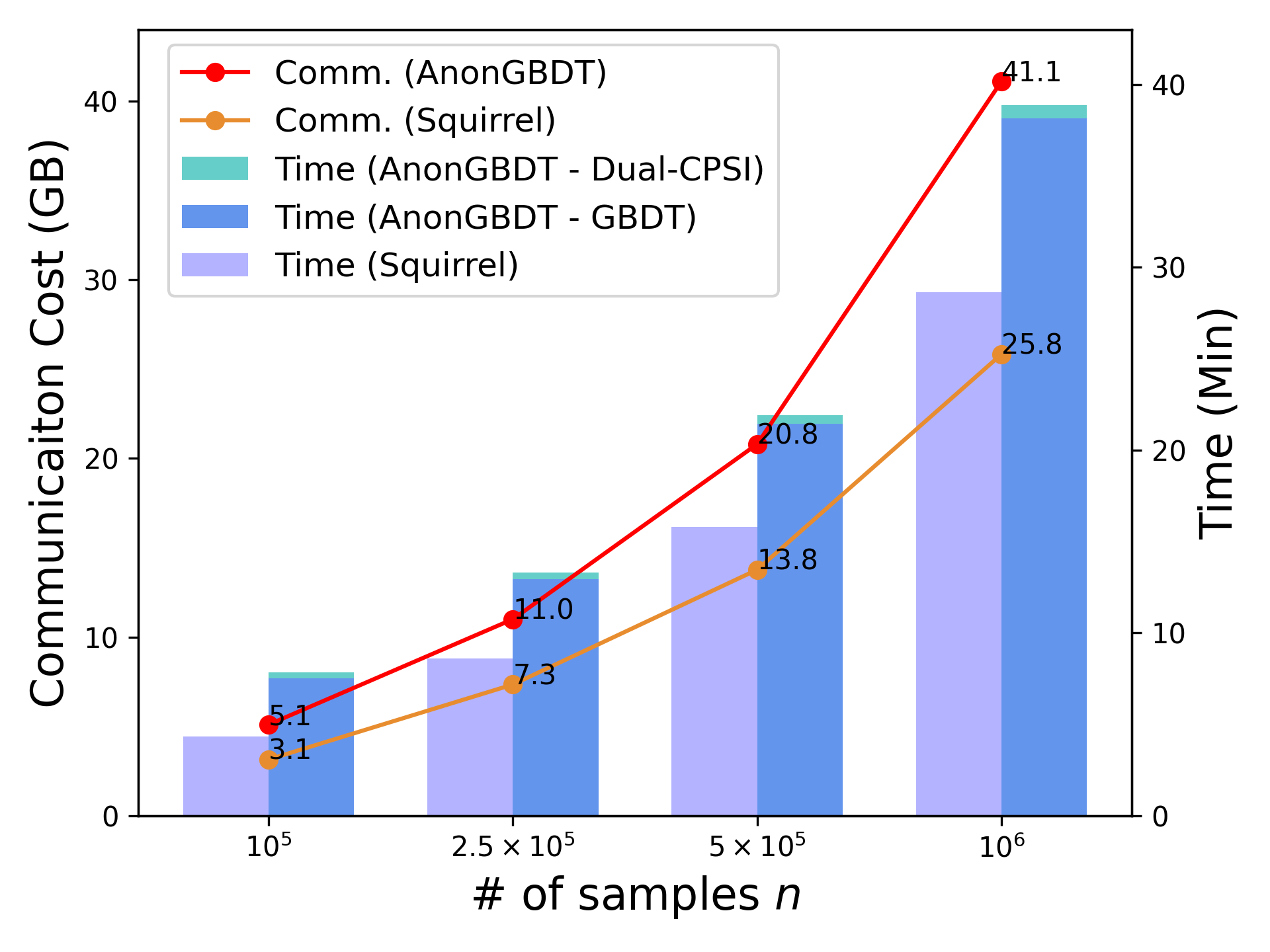}
		\subfloat{\footnotesize Time (WAN) vs. Sample size $n$}
\label{fig:eva_scalability_sample_size_wan}
	\end{minipage}
	\hfill
	\begin{minipage}[t]{.32\linewidth}
		\centering
		\includegraphics[width = \linewidth]{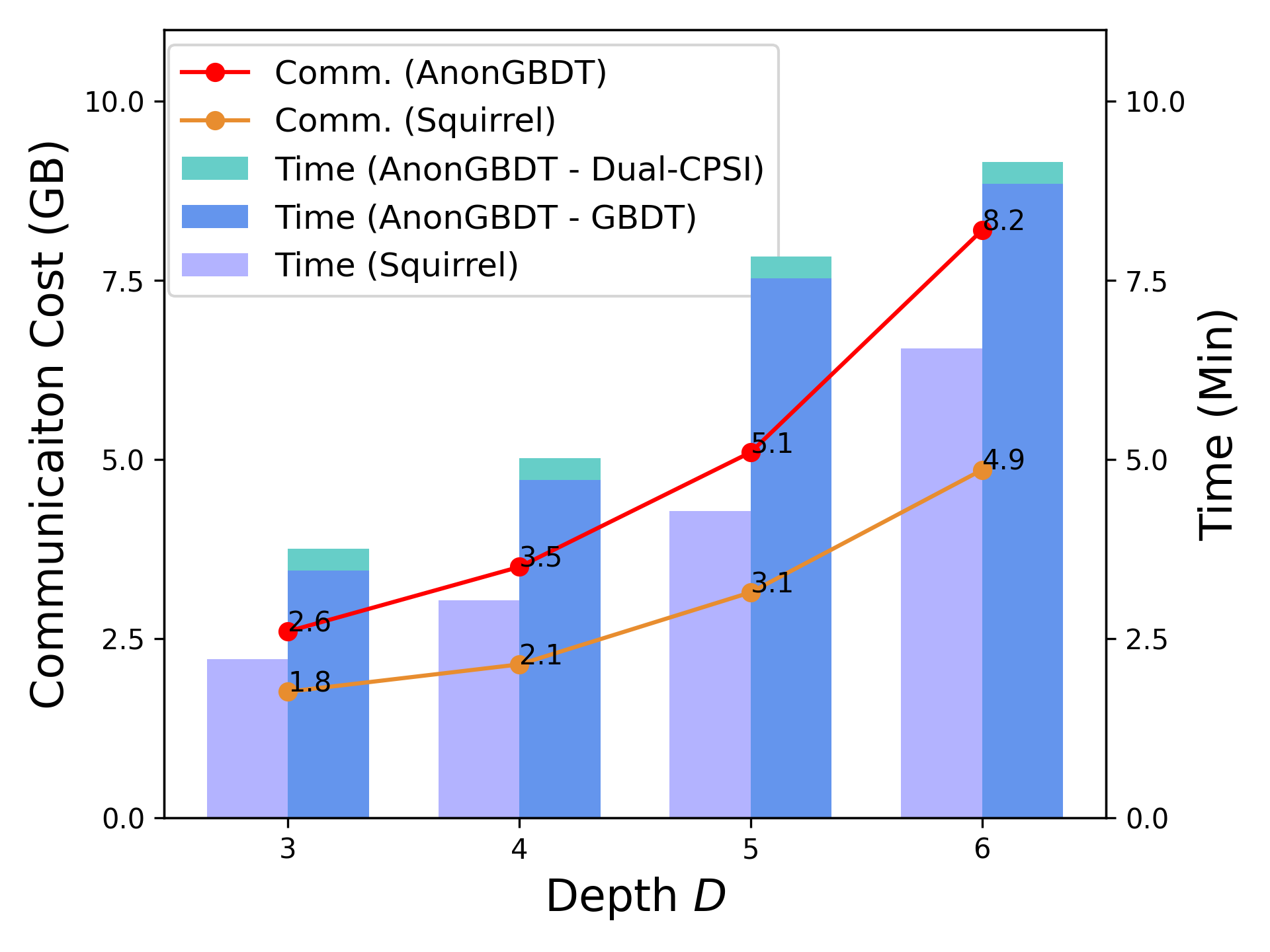}
		\subfloat{\footnotesize Time (WAN) vs. Depth $D$}
		\label{fig:eva_scalability_depth_wan}
	\end{minipage}
	\hfill
	\begin{minipage}[t]{.32\linewidth}
		\centering
		\includegraphics[width = \linewidth]{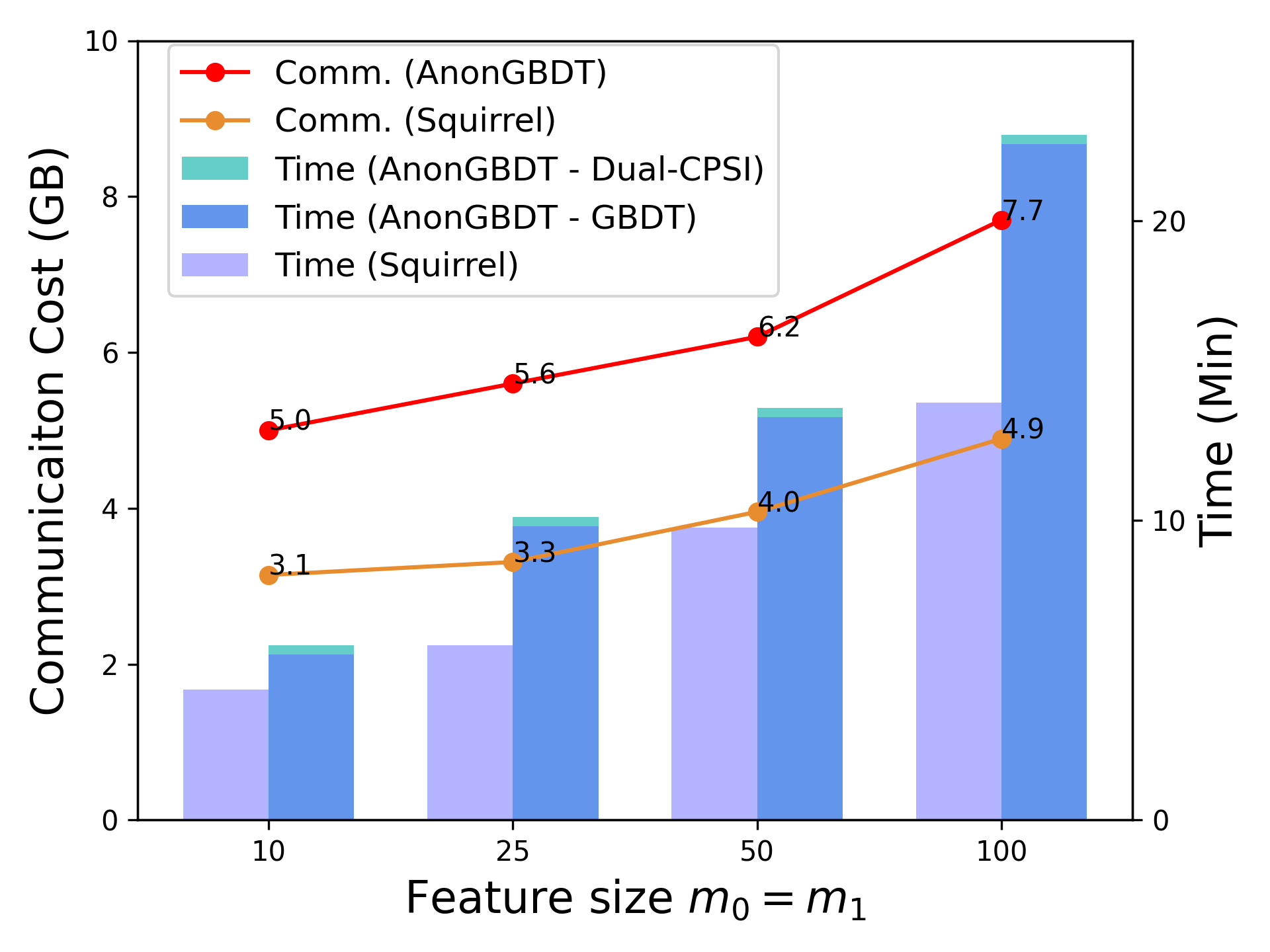}
		\subfloat{\footnotesize Time (WAN) vs. Feature size $m_0 = m_1$}
		\label{fig:eva_scalability_feature_size_wan}
	\end{minipage}
	\caption{Scalability of $\sysfinal$ vs. Squirrel ($T = 10$, $n = 10^5$, $m= 10$, $B = 16$, $D = 5$, and $8$ threads)}
	\vspace{-5pt}
	\label{fig:eva_scalability}
\end{figure*}

\subsection{Comparisons with Prior Arts}
\label{sect:exp:comparison}
\noindent
\textbf{Efficiency}:
$\sysfinal$ is evaluated on the same hardware as Squirrel~\cite{uss/LuHZWH23}, across two hosts with $8$ CPU cores each.
Pivot~\cite{pvldb/WuCXCO20} and HEP-XGB~\cite{cikm/FangZT0YWWZZ21} were evaluated under more favorable conditions (\eg, $32$ cores for HEP-XGB) and heterogeneous network settings;
we quoted their reported results in Table~\ref{tab:eva_comp_advance_exisiting_works}.
Despite the stronger privacy guarantees (no identifier leakage), our runtime is only ${\sim}0.7\times$--$1.7\times$ that of Squirrel and remains~competitive with the other less private schemes (which would require prohibitively costly generic tools for similar privacy).

\smallskip
\noindent
\textbf{Effectiveness}:
We train models via $\sysfinal$ on six public datasets~\cite{libsvm} and evaluate with our inference protocol $\varPi_{\mathsf{Infer}}$ in Figure~\ref{alg:anon_infer}.
XGBoost~\cite{xgblib} serves as a plaintext baseline.
Using $5$-fold cross-validation, $T = 10$ trees, and depth $D = 5$, we attain the highest F1 score across all tasks (Table~\ref{tab:eva_comp_literatures_effectiveness}), primarily due to our improved sigmoid approximation (Section~\ref{sect:our_sigmoid}).




\begin{table}[t]
	\small
	\centering
	\caption{F1 score comparison}
	\begin{tabular}{l|c|c|c|c||c}
		\hline 
		Dataset & Ours & Squirrel & Pivot & \begin{tabular}{@{}c@{}}HEP- \\ XGB \end{tabular} & \begin{tabular}{@{}c@{}}XG- \\ Boost \end{tabular} \\ 
		\hline 
		\hline 
		breast-cancer	& $\mathbf{0.920}$ & $0.917$ & $0.919$ & $0.889$ & $0.947$ \\
		phishing	& $\mathbf{0.957}$ & $0.957$ & $0.957$ & $0.951$ & $0.961$ \\
		a9a			& $\mathbf{0.654}$ & $0.651$ & $0.653$ & $0.643$ & $0.654$ \\
		cod-rna		& $\mathbf{0.882}$ & $0.402$ & $0.408$ & $0.403$ & $0.882$ \\
		skin\_nonskin & $\mathbf{0.988}$ & $0.742$ & $0.743$ & $0.741$ & $0.989$ \\
		covtype.binary\!\!		& $\mathbf{0.782}$ & $0.556$ & $0.572$ & $0.552$ & $0.796$ \\
		\hline
	\end{tabular}
	\label{tab:eva_comp_literatures_effectiveness}
	\vspace{-3pt}
\end{table}

\subsection{Scalability}
We assess scalability by running $\sysfinal$ with eight threads (unlike the settings for comparison in Section~\ref{sect:exp:comparison})
on synthetic datasets that vary in sample count~$n$, feature count $m$, and tree depth $D$.
Figure~\ref{fig:eva_scalability} reports runtime and traffic;
Squirrel's numbers are directly quoted.

Thanks to $\mathsf{FastPackLWEs}$ and the other engineering optimizations (Section~\ref{sect:adv_optimizations}), our LAN runtime is on par with Squirrel.
Because the improved sigmoid scales sublinearly with $n$, our scheme grows more gracefully as the sample size increases.
Communication rises linearly with~$n$, but its growth with $m$ and $D$ is markedly flatter:
i) dual-circuit-PSI allows us to use the communication-friendly $\mathcal{F}_{\mathsf{BinMatVec}}$, keeping the cost in $m$ low, and 
ii) our $\mathcal{F}_{\mathsf{BOIS}}$ reduces the indicator-sync complexity from $\mathcal{O}(2^D)$ to $\mathcal{O}(D)$.
Under WAN conditions, runtime is $2$--$4\times$ higher than in LAN, consistent with the added latency and reduced bandwidth.


\section{Discussion}
\subsection{Gradient-based One-Side Sampling (GOSS)}
LightGBM~\cite{nips/KeMFWCMYL17} accelerates histogram construction by prioritizing ``informative'' samples with large (first-order) gradient magnitudes.
Given the per-sample gradients $\{g_{i}\}_{i = 1}^{n}$ at a node, GOSS first sorts the samples by $|g_{i}|$ in descending order.
It keeps the top fraction $a$ (high-gradient set $\mathbb{H}$) \emph{intact}.
From the remaining set $\mathbb{L}$, it randomly samples a fraction $b$.
To correct the sampling bias, it multiplies the gradients and Hessians of the sampled low-gradient points by $(1 - a) / (1 - b)$ before building the histogram.
This reduces the working-set size from the sample size $n$ to $|\mathbb{H}| + b|\mathbb{L}|$, speeding up training in the plaintext setting.

\smallskip
\noindent
\textbf{Limitations in Secure Two-party Computation}:
Identifying $H$ requires an \emph{oblivious sort} of the secret-shared gradients, which costs $O(n\log n)$ time and substantial bandwidth (\cf, ${\sim}4$GB for one million records with $\ell = 32$~\cite{pkc/PecenyRRS25}).
The subsequent random sampling must also be oblivious, adding further overhead.
Consequently, the gains of GOSS in the clear are outweighed by the expense of these secure primitives, making it unsuitable for our secure GBDT pipeline.


\subsection{Private-ID}
Private-ID protocols~\cite{iacr/PrasadAPYSEV20,pkc/GarimellaMRSS21} return \emph{pseudo-random} (or secret-shared) identifiers for the intersection, rather than a Boolean mask over \emph{all} records (even those not in the intersection) as in circuit-PSI.
They have two limitations.

\smallskip
\noindent
\textbf{Cardinality leakage}:
They explicitly reveal $|\mathbf{ID}_{0} \cap \mathbf{ID}_{1}|$, whereas circuit-PSI does not.
Hiding $|\mathbf{ID}_{0} \cap \mathbf{ID}_{1}|$ motivates threshold PSI designs~\cite{wpes/ZhaoC18} that return intersection only if its size exceeds a threshold, whereas prior threshold PSI protocols may leak this size.

\smallskip
\noindent
\textbf{Loss of ordering}:
Neither party learns how the output IDs align with its local feature matrix~$\mathbf{M}$;
otherwise, the true IDs could be inferred from the feature order.
As a result, any subsequent operations must operate on \emph{secret-shared} versions $\langle \mathbf{M} \rangle_l^{\mathsf{B}}$.
For secure GBDT, this means gradient histograms must be built with the OT-heavy multiplexer, replicating the inefficiency of our base design.

In contrast, circuit-PSI reveals the ID-feature correspondence to the \emph{receiver}.
With our elaborated dual-circuit-PSI design, both parties possess a plaintext copy of their own~$\mathbf{M}$ aligned to the intersection, enabling both sides to use the more communication-efficient $\mathcal{F}_{\mathsf{BinMatVec}}$.


\subsection{Hiding the Splits}
All known secure two-party decision-tree protocols, including ours (with data alignment hidden), reveal the best splits at internal nodes~\cite{uss/LuHZWH23,crypto/LindellP00,pvldb/WuCXCO20,cikm/FangZT0YWWZZ21}, a leakage previously noted~\cite{dbsec/ZhuD10}.
A recent work~\cite{sigmod/HanCZFHS25} attempts to mitigate this by obliviously permuting the raw matrix $\mathbf{M}$ to $\langle\mathbf{M}'\rangle^{\mathsf{A}}$, secret-sharing a one-hot vector $\langle \mathbf{s} \rangle^{\mathsf{A}}$ for the chosen bin, and revealing, for each sample~$i$, the indicator $\mathbf{I}_{s}[i] = \mathbf{M}'[i] \cdot \mathbf{s}$ that determines the branch.
Since the permutation is hidden, the split index is assumed private.
Yet, the protocol leaks $|\mathbf{ID}_{0} \cap \mathbf{ID}_{1}|$ and enables each party to correlate revealed left/right counts with its own prefix matrix~$\widetilde{\mathbf{M}}$, allowing recovery of the feature/bin and contradicting privacy claims.

A direct remedy is to keep split indices secret-shared.
Three-party designs~\cite{popets/HamadaIKC23,ccs/BhardwajSCG24,ccs/LinHRZSL24} realize this via group-based data structures and per-layer oblivious sorting, but such sorting is prohibitively expensive with two parties.
Devising an efficient $2$-party variant remains an open problem.

\section{Related Work}
Some federated‑learning (FL) frameworks target private GBDT training.
SecureBoost(+)~\cite{expert/ChengFJLCPY21} employs Paillier homomorphic encryption, and VF$^{2}$Boost~\cite{sigmod/FuSYJXT021} refines the pipeline for large datasets.
They only protect the \emph{gradient aggregation},
leaving internal artifacts, including split points, leaf updates, \etc, in plaintext.
Such leakage enables inference attacks~\cite{uss/FuZJCWGZLW22,corr/TakahashiLL23},
for instance, if an income feature determines a split, a party lacking that feature can deduce a sample's income range from its branch assignment.
Takahashi~\etal~\cite{corr/TakahashiLL23} use local differential privacy to curb label leakage.
Although computationally light, differential privacy adds noise that trades accuracy for privacy.

Secure multiparty computation has been widely adopted.
Lindell and Pinkas~\cite{crypto/LindellP00} propose secure multiparty computation for decision-tree training using garbled circuits and OT.
Hoogh~\etal~\cite{fc/HooghSCA14} use secret sharing for multiparty GBDT training.
Abspoel~\etal~\cite{popets/AbspoelEV21} propose an oblivious sorting network for GBDT training with continuous feature values.
Due to high communication costs, these designs are inefficient for large datasets.
Solutions using function secret sharing~\cite{tifs/ChenLWHXZ23,asiaccs/JiangMDL24} achieve efficient online decision tree training, but they require offline 
correlated random number generation to support secure multiplication.

The most relevant work~\cite{pvldb/WuCXCO20, cikm/FangZT0YWWZZ21, uss/LuHZWH23} to ours is the hybrid approaches.
Pivot~\cite{pvldb/WuCXCO20} incorporates threshold Paillier HE and secret sharing for tree training.
Fang~\etal~\cite{cikm/FangZT0YWWZZ21} propose private large-scale XGBoost training.
However, their scheme heavily relies on a semi-honest third party, \eg, trusted execution environments;
otherwise, the performance will degrade significantly.
Squirrel~\cite{uss/LuHZWH23} improves the histogram-calculation efficiency via RLWE-based (instead of Paillier) HE and proposes a new method to approximate the sigmoid function using a Fourier series.
Instead, its follow-up~\cite{dlsp/DaiJLM24} realizes pick-then-sum with pure secret sharing, shifting most costs to an offline phase.
See the systematization of Chatel~\etal~\cite{popets/ChatelPTH21} for related results, including earlier work.

Some selected schemes develop anonymous vertical FL.
Sun~\etal~\cite{corr/SunYYZGXW21} suggest using private set union instead of PSI to obscure the common IDs and generating synthetic features for samples outside the intersection, which can degrade model accuracy.
OpenVFL~\cite{tifs/YangCPSCDLSYD24} proposes a labeled PSI that returns only homomorphically encrypted features within the intersection.
This conceals the common elements and confines the subsequent training to the encrypted domain under homomorphic encryption, which is costly when GBDT pays heavily for secure comparisons.
OpenVFL~\cite{tifs/YangCPSCDLSYD24} implemented logistic regression only
and is not open source.

 
\section{Conclusion}
Existing private GBDT protocols run standard PSI once to align the two datasets on their intersection, then train on the aligned records, which inevitably discloses the intersection.
Circuit-PSI can avoid this leakage in principle, but its asymmetric interface clashes with order-dependent training flows, and embedding the training loop as a generic circuit makes it impractically slow.

We introduce \sys, the first anonymous GBDT framework.
The base variant adds two new primitives: $\mathcal{F}_{\mathsf{OIS}}$ for indicator synchronization, and $\mathcal{F}_{\mathsf{mux}}$-based gradient aggregation.
Although more private, it incurs considerable communication and memory overheads.

Our final scheme $\sysfinal$ resolves these bottlenecks through a dual-circuit-PSI architecture that symmetrizes ownership of the ID-feature mapping.
It enables efficient~$\mathcal{F}_{\mathsf{BinMatVec}}$ aggregation, further accelerated by our $\mathsf{FastPackLWEs}$ subroutine.
Moreover, we propose a batched OPPRF-based OIS protocol.
Empirically, $\sysfinal$ is up to an order of magnitude faster than the base design.

Beyond GBDT, our dual-circuit-PSI can integrate with other vertical-learning frameworks, \eg, Split-NN~\cite{jnca/GuptaR18}, to add anonymity at modest overheads.

\bibliographystyle{IEEEtran}
\bibliography{reference}

\appendices


\section{Optimizations in Sigmoid}
\subsection{Secure Decimal Multiplication for 
Trigonometric Functions}
\label{apdx:sigmoid_mul}
$\mathcal{F}_{\mathsf{mul}}$ relies on Beaver triples~\cite{crypto/Beaver91a} for $a \cdot b = c$ in secret shares, which can be efficiently implemented in RLWE-based HE~\cite{cans/RatheeSS19}, \eg, Brakerski--Gentry--Vaikuntanathan~\cite{toct/BrakerskiGV14}.
Specifically, $P_0$ generates random $a \sample \mathbb{Z}_{2^{64}}$ and sends ciphertext $\llbracket a \rrbracket$ to $P_1$.
$P_1$ generates another random $b \sample \mathbb{Z}_{2^{64}}$ and gets $\llbracket c \rrbracket = \llbracket a \rrbracket \cdot b$.
However, the ciphertext length $N$ is usually $8192$, and the plaintext modulus is no more than a $32$-bit prime number (for efficient homomorphic operations), which cannot encode $c$ with a bit size of $128$ to fit $\mathbb{Z}_{2^{64}}$.
The literature~\cite{cans/RatheeSS19} solves this by splitting the large plaintext modulus $p$ into four smaller $p_i$ where $p = \prod{p_i}$ via the Chinese remainder theorem (CRT).
Four instances (with modulus $p_i$ respectively) of the above protocol are executed to get outputs $c_0, \ldots, c_3$,
which can be combined into the output $c$ (with modulus $p$) via CRT and then truncated into $\mathbb{Z}_{2^{64}}$.
Note that each $p_i$ corresponds to $Q$ of bit size.

In our $\mathcal{F}_{\mathsf{sigmoid}}$, the range of floating numbers $\sin(x)$, $\cos(x)$, and $\sin(x) \cdot \cos(x)$ are all within $[-1, 1]$ (encoded as $[-2^f, 2^f]$).
Recalling the average error of our sigmoid approximation is $0.002$, we can set $f = 15$ in multiplication.
Therefore, the plaintext modulus when encrypting $\sin(x)$ and $\cos(x)$ can be less than $2^{32}$, using only one $p$ instead of four $p_i$.
The communication/computation cost of one $\mathcal{F}_{\mathsf{mul}}$ in $\mathcal{F}_{\mathsf{sigmoid}}$ is now only $1 / 4$ of the original $\mathcal{F}_{\mathsf{mul}}$.

\subsection{Configuration of Fourier Series}
\label{apdx:fourier}
\begin{figure}[htb]
	\centering
	\includegraphics[width = 0.8 \linewidth]{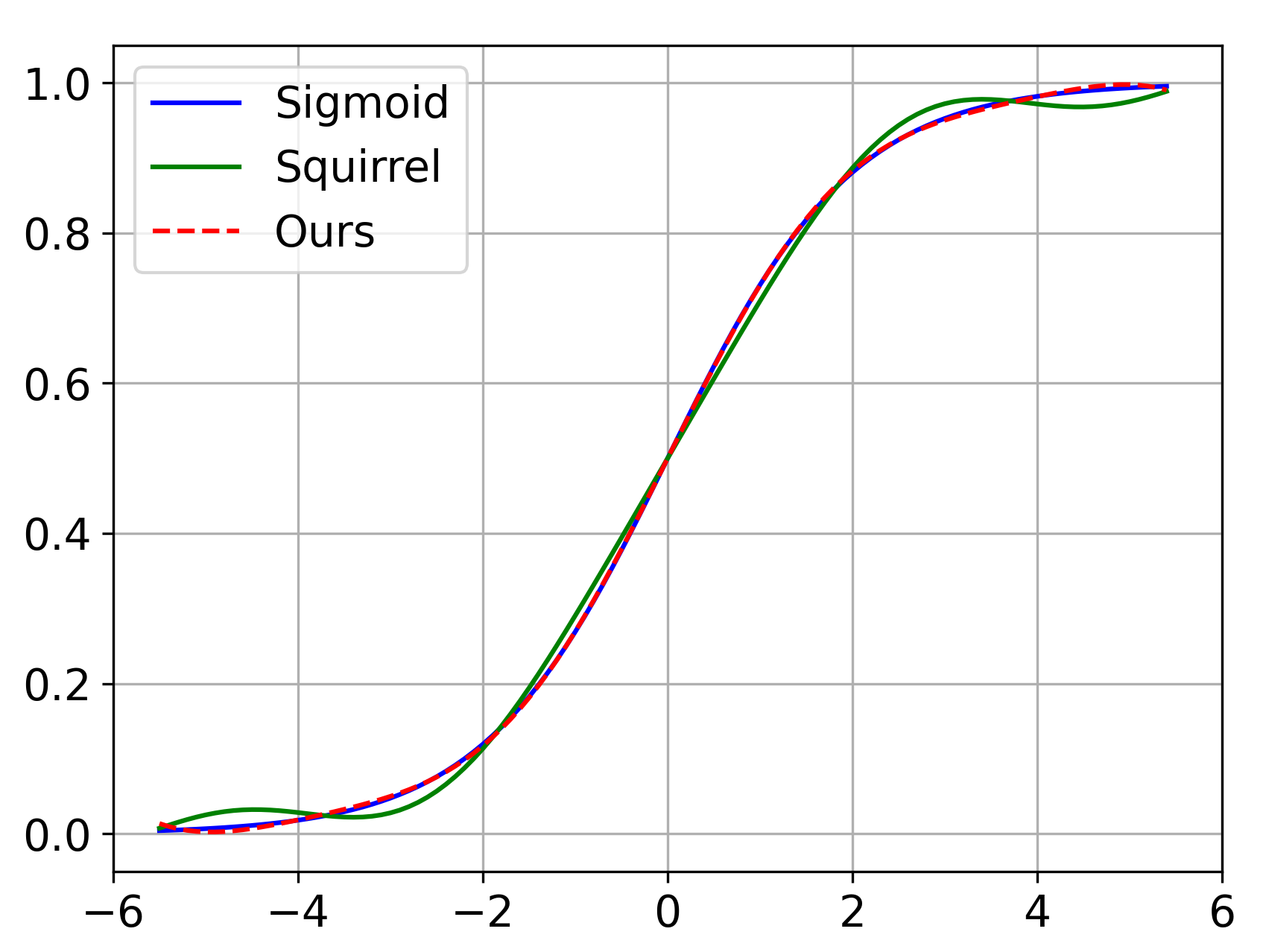}
	\vspace{-10pt}
	\caption{Comparison of sigmoid approximations}
\label{fig:sigmoid_approximation}
	\vspace{-3pt}
\end{figure}

We use the following coefficients and parameters
\begin{align*}
	a_0 & = 0.5, & a_1 & = 1.642327,
	\\
	a_2 & = -1.070336, & a_3 & = 0.5510985,
\end{align*}
and $L = 3$ such that the Fourier series output is bounded by $1$ within $[-5.6, 5.6]$.
Figure~\ref{fig:sigmoid_approximation} shows the comparison between our sigmoid with Squirrel and the ground truth.

\subsection{Complexity Comparison}
Table~\ref{tab:complexity_comparison_sigmoid} compares our $\mathcal{F}_{\mathsf{sigmoid}}$ with the corresponding protocol of Squirrel.
With $\lambda = 128$, $f = 20$, and $\delta = 4$, our total communication is $2,976$ bits (vs.\ $19,816$ for Squirrel).

\begin{table}[htb]
	\small
	\centering
	\caption{Amortized bit communication (comm.) complexity comparison of $\mathcal{F}_{\mathsf{sigmoid}}$ between ours and Squirrel's}
	\vspace{-5pt}
	\begin{tabular}{|c|c|c|}
		\hline
		$\mathcal{F}_{\mathsf{sigmoid}}$ & Ours & Squirrel \\
		\hline
		\begin{tabular}{@{}c@{}} Comm. complexity \\ of $\mathcal{F}_{\mathsf{mul}}$ \end{tabular} 
		& $288$ & $288 \cdot 4$ \\
		\hline
		\begin{tabular}{@{}c@{}} Comm. complexity \\ of two $\mathcal{F}_{\mathsf{greater}}$\end{tabular}
		& 
		\begin{tabular}{@{}c@{}}${\approx}1384 - 4 \log_2 \lfloor\frac{f - \delta}{4}\rfloor$
		\\
		$- 32\lfloor\frac{f - \delta}{4}\rfloor$
		\end{tabular}
		& ${\approx}1384$ \\
		\hline
		\begin{tabular}{@{}c@{}}\# of terms in\\ Fourier series \end{tabular} & $4$ & $9$ \\
		\hline
		\begin{tabular}{@{}c@{}} Comm. complexity \\ of $\mathcal{F}_{\mathsf{sigmoid}}$\end{tabular} & 
		\begin{tabular}{@{}c@{}}${\approx}1384 - 4 \log_2 \lfloor\frac{f - \delta}{4}\rfloor$
		\\
		$ - 32\lfloor\frac{f - \delta}{4}\rfloor+ 288 \cdot 6$
		\end{tabular}&
		\begin{tabular}{@{}c@{}} ${\approx}1384~+$ \\
		~~$288 \cdot 4 \cdot 16$\end{tabular} \\
		\hline
	\end{tabular}
	\label{tab:complexity_comparison_sigmoid}
	\vspace{-5pt}
\end{table}

\section{Batched OPPRF-based OIS (BOIS)}
\label{apdx:bois}

The functionality and protocol of batched OPPRF-based BOIS are shown in Figure~\ref{alg:bois}.

For completeness, we note an alternative realization of $\mathcal{F}_{\mathsf{BOIS}}$ via oblivious transfer for a sparse array (OTSA)~\cite{popets/ZhaoC17},
which provides an alternative view of the OPPRF step as a sparse-array transfer.
Namely, it stores the programmed shares at the same bucket-index domain and uses fixed-length fillers sampled from the output range.

For comparison of array-transfer abstractions, Zhang~\etal~\cite{uss/ZhangCLZL23} compare OTSA with oblivious key-value stores~\cite{ccs/RaghuramanR22} and observe that OTSA additionally hides the query-to-output correspondence.

\begin{figure}[t!]
\begin{tcolorbox}[colback = white!5!white, left = 0mm, right = 0mm, top = 0mm, bottom = 0mm, title = Functionality $\mathcal{F}_{\mathsf{BOIS}}$]
\begin{algorithmic}[1]
	\PcInput
	$\{(z_*^{(k)}, u_*^{(k)})\}_{k \in \mathbf{k}}$,
	$\{\langle \mathbf{b}^{(k, 0)} \rangle^\mathsf{B}_c, \langle \mathbf{b}^{(k, 1)} \rangle^\mathsf{B}_c\}_{k \in \mathbf{k}}$,
	$\widetilde{\mathbf{M}}_c$, 
	and the simple hash table $\mathbf{T}_{\mathsf{sh}}^{1 - c}$.
	\PdInput
	$\{\langle \mathbf{b}^{(k, 0)} \rangle^\mathsf{B}_{1 - c}, \langle \mathbf{b}^{(k, 1)} \rangle^\mathsf{B}_{1 - c}\}_{k \in \mathbf{k}}$, 
	and
	\newline
	the cuckoo hash table $\mathbf{T}_{\mathsf{ch}}^{1 - c}$.
	\PP A set of nodes $\mathbf{k}$ to be synchronized
	\Ensure
	$\{\langle \mathbf{b}^{(2k, i)} \rangle^\mathsf{B}$,
	$\langle \mathbf{b}^{(2k + 1, i)} \rangle^\mathsf{B}\}_{i \in \{0, 1\}, k \in \mathbf{k}}$.
\end{algorithmic}
\end{tcolorbox}

\begin{tcolorbox}[colback = white!5!white, left = 0mm, right = 0mm, top = 0mm, bottom = 0mm, title = Protocol $\varPi_{\mathsf{BOIS}}$]
\begin{algorithmic}[1]
	\For{each node $k \in \mathbf{k}$} \label{alg:bois:L1}
	\State $P_c$ samples $\langle \mathbf{b}_*^{(k, c)} \rangle^\mathsf{B}_c = \mathbf{r}^c \sample \{0, 1\}^{(1 + \epsilon)n}$.
	\label{alg:bois:L2}
	\State $P_c$ sets $\mathbf{b}_*^{(k, c)} = \widetilde{\mathbf{M}}_c[z_*^{(k)}B + u^{(k)}_*]$.
	\label{alg:bois:L3}
	\State $P_c$ shares $\langle \mathbf{b}_*^{(k, c)} \rangle^\mathsf{B}_{1 - c} \gets \mathbf{b}_*^{(k, c)} \oplus \mathbf{r}^c$ with $P_{1 - c}$.
	\label{alg:bois:L4}
	\State \parbox[t]{\dimexpr\linewidth - \leftmargin - \labelsep - \labelwidth}{$P_l$: $\langle \mathbf{b}^{(2k, c)} \rangle^\mathsf{B}_l = \mathcal{F}_{\mathsf{and}}(\langle \mathbf{b}^{(k, c)}_* \rangle^\mathsf{B}_l, \langle \mathbf{b}^{(k, c)} \rangle^\mathsf{B}_l)$ 
	\newline and
	$\langle \mathbf{b}^{(2k + 1, c)} \rangle^\mathsf{B}_l = \langle \mathbf{b}^{(2k, c)} \rangle^\mathsf{B}_l \oplus \langle \mathbf{b}^{(k, c)} \rangle^\mathsf{B}_l$.\strut}
	\label{alg:bois:L5}
	\EndFor
	\State $P_c$ samples $\mathbf{r}^{1 - c} \sample \{\mathbb{Z}_{2^{|\mathbf{k}|}}\}^{(1 + \epsilon)n}$, and initializes packed assignment $\widetilde{\mathbf{b}}'_* = \mathbf{0}^{(1 + \epsilon)n}$.
	\label{alg:bois:L6}
	\For{each node $k \in \mathbf{k}$}
	\label{alg:bois:L7}
	\State $P_c$ ``reorders'' $\mathbf{b}_*^{(k, c)}$ to $\mathbf{b}'_*$ to align with $\mathbf{ID}_c$.
	\label{alg:bois:L8}
	\State $P_c$ packs $\widetilde{\mathbf{b}}'_* = \widetilde{\mathbf{b}}'_* \oplus (\mathbf{b}'_* \ll k)$.
	\label{alg:bois:L9}
	\State $P_c$ sets $\langle \mathbf{b}^{(k, 1 - c)}_* \rangle^\mathsf{B}_{c} = (\mathbf{r}^{1 - c}\gg k)\& \mathbf{1}^{(1 + \epsilon)n}$.
	\EndFor \label{alg:bois:L10}
	\State $\forall j \in [0, e), \forall i \in [0, (1 + \epsilon)n)$:
	\newline
	$P_c$ sets a table $\mathbf{T} = \{(\mathbf{T}_{\mathsf{sh}}^{1 - c}[i, j], \widetilde{\mathbf{b}}'_*[i'] \oplus \mathbf{r}^{1 - c}[i] \}$, and
	$\mathbf{ID}_c[i']$ is placed in $\mathbf{T}_{\mathsf{sh}}^{1 - c}[i, j]$ via simple hashing.\label{alg:bois:L11} 
	\State $P_l$ runs $\mathcal{F}_{\mathsf{OPPRF}}(\{\mathbf{T}\}, \{\mathbf{T}_{\mathsf{ch}}^{1 - c}\})$;
	$P_{1 - c}$ gets $\widetilde{\mathbf{b}}^{(1 - c)}_*$ \label{alg:bois:L12}
	\For{each node $k \in \mathbf{k}$}
	\label{alg:bois:L13}
	\State $P_{1 - c}$: $\langle \mathbf{b}^{(k, 1 - c)}_* \rangle^\mathsf{B}_{1 - c} = (\widetilde{\mathbf{b}}^{(1 - c)}_* \gg k) \& \mathbf{1}^{(1 + \epsilon)n}$.
	\label{alg:bois:L14}
	\State $P_l$:$\langle \mathbf{b}^{(2k, 1 - c)} \rangle^\mathsf{B}_l\!=\!\mathcal{F}_{\mathsf{and}}(\langle \mathbf{b}^{(k, 1 - c)}_* \rangle^\mathsf{B}_l, \langle \mathbf{b}^{(k, 1 - c)} \rangle^\mathsf{B}_l)$.\!\label{alg:bois:L15}
	\State $P_l$:$\langle \mathbf{b}^{(2k + 1, 1 - c)} \rangle^\mathsf{B}_l = \langle \mathbf{b}^{(2k, 1 - c)} \rangle^\mathsf{B}_l \oplus \langle \mathbf{b}^{(k, 1 - c)} \rangle^\mathsf{B}_l$.
	\label{alg:bois:L16}
	\EndFor
\end{algorithmic}
\end{tcolorbox}
	\vspace{-5pt}
\caption{Ideal functionality and protocol of $\mathsf{BOIS}$}
	\vspace{-3pt}
\label{alg:bois}
\end{figure}

\section{Anonymous Batch Inference}
\label{apdx:sec_inference}
The functionality and the protocol of our anonymous batch inference are shown in Figure~\ref{alg:anon_infer}.

\begin{figure}
\begin{tcolorbox}[colback = white!5!white, left = 0mm, right = 0mm, top = 0mm, bottom = 0mm, title = Functionality $\mathcal{F}_{\mathsf{Infer}}$]
\begin{algorithmic}[1]
	\PaInput $\mathbf{ID}_0$, $\mathbf{X}_0$, best split $\mathbf{C}_0$, weight share $\langle \mathbf{w} \rangle^\mathsf{A}_0$.
	\PbInput $\mathbf{ID}_1$, $\mathbf{X}_1$, best split $\mathbf{C}_1$, weight share $\langle \mathbf{w} \rangle^\mathsf{A}_1$.
	\Ensure Prediction probability $\mathbf{p}$ to $P_0$.
\end{algorithmic}
\end{tcolorbox}

\begin{tcolorbox}[colback = white!5!white, left = -1mm, right = -1mm, top = 0mm, bottom = 0mm, title = Protocol $\varPi_{\mathsf{Infer}}$]
	\begin{algorithmic}[1]
		\State Jointly invoke
		$\langle \mathbf{q} \rangle^\mathsf{B}_l \ \gets
		\mathcal{F}_{\mathsf{CPSI}}(\{\mathbf{ID}_0\}, \{\mathbf{ID}_1\})$,
		$P_0$ runs as receiver with cuckoo hash table $\mathbf{T}_{\mathsf{ch}}$,
		$P_1$ runs as sender with simple hash table $\mathbf{T}_{\mathsf{sh}}$.
	\label{alg:anon_infer:L1}
			\For {tree $t \in \mathcal{T}_t$}
			\State $P_l$ sets root node's indicator $\mathbf{b}^{(0,~\!l)} = \mathbf{1}^{(1 + \epsilon)n}$
		\label{alg:anon_infer:L3}
			\For{internal node $k \in [1, 2^{D - 1})$, each sample $i$}
			\If{$\mathbf{C}_l[k]$ is not $\perp$}
			\State $P_l$ gets $(z^{(k)}_*, u^{(k)}_*)\gets\mathbf{C}_l[k]$
			\If{$\mathbf{X}_l[i, z^{(k)}_*] \leq$ split value of ${z_*^{(k)}, u_*^{(k)}}$}
			\State $P_l$:
			$\mathbf{b}^{(2k, l)}[i] = \mathbf{b}^{(k, l)}[i]$.
			\State $P_l$:
			$\mathbf{b}^{(2k + 1, l)}[i] = 0$.
			\Else
			\State $P_l$:
			$\mathbf{b}^{(2k, l)}[i] = 0$.
			\State $P_l$:
			$\mathbf{b}^{(2k + 1, l)}[i] = \mathbf{b}^{(k, l)}_l[i]$.
			\EndIf
			\Else~
			$P_l$: $\mathbf{b}^{(2k, l)}[i] = \mathbf{b}^{(2k + 1, l)}[i] = \mathbf{b}^{(k, l)}[i]$ 
			\EndIf
			\EndFor 
			\State $P_0$ shares $\mathbf{b}^{(0)}$, $P_l$ gets $\langle \mathbf{b}^{(0)} \rangle^\mathsf{B}_l$.
		\label{alg:anon_infer:L14}
			\State \parbox[t]{\dimexpr\linewidth - \leftmargin - \labelsep - \labelwidth+2mm}{$P_1$: sets $\mathbf{r} \sample \{\mathbb{Z}_{2^{\text{\#leaves}}}\}^{(1 + \epsilon)n}$, and extracts $\langle \mathbf{b}^{(k, 1)} \rangle^\mathsf{B}_1$ which is the $k - 2^{D - 1}$-th bit of $\mathbf{r}$.\strut}
			\State $P_1$ packs all the leaves' indicators $\mathbf{b}^{(k, 1)}$ as $\mathbf{\widetilde{b}}'$
			\State \parbox[t]{\dimexpr\linewidth - \leftmargin - \labelsep - \labelwidth+2mm}{$\forall i \in [0, (1 + \epsilon)n)$, $\forall j \in [0, e)$,
					$P_1$ creates a table
					$\mathbf{T} = \{(\mathbf{T}_{\mathsf{sh}}[i, j], \mathbf{\widetilde{b}}'[i'] \oplus \mathbf{r}[i])\}$
					where $\mathbf{ID}[i']$ is placed in $\mathbf{T}_{\mathsf{sh}}$ via simple hashing.
			\strut}
			\State $P_l$ runs $\mathcal{F}_{\mathsf{OPPRF}}(\{\mathbf{T}\}, \{\mathbf{T}_{\mathsf{ch}}\})$;
			$P_0$ gets $\widetilde{\mathbf{b}}''$.
			\State $P_l$ initializes $\langle \mathbf{S} \rangle^\mathsf{B}_l = \mathbf{0}^{(1 + \epsilon)n}$
			\For{node $k \in \{2^{D - 1}, \ldots, 2^D - 1\}$}
			\State \parbox[t]{\dimexpr\linewidth - \leftmargin - \labelsep - \labelwidth}{$P_l$ extract indicator $\langle \mathbf{b}^{(k, 1)} \rangle^\mathsf{B}_l$ which is the $(k - 2^{D - 1})$-th bit of $\widetilde{\mathbf{b}}''$.
			\strut}
			\State $P_l$:
			$\langle \mathbf{b}^{(k)} \rangle^\mathsf{B}_l = 
			\mathcal{F}_{\mathsf{and}}(\langle \mathbf{b}^{(k, 0)} \rangle^\mathsf{B}_l, \langle \mathbf{b}^{(k, 1)} \rangle^\mathsf{B}_l)$
			\State \parbox[t]{\dimexpr\linewidth - \leftmargin - \labelsep - \labelwidth}{$P_l$:
			$\forall i \in [0, (1 + \epsilon)n), \langle \mathbf{S}[i] \rangle^\mathsf{A}_l = \langle \mathbf{S}[i] \rangle^\mathsf{A}_l + \mathcal{F}_{\mathsf{mux}}(\langle \mathbf{b}^{(k)}[i] \rangle^\mathsf{B}_l, \langle \mathbf{w}[k] \rangle^\mathsf{A}_l)$ \strut} \label{alg:anon_infer:L23}
			\EndFor 
			\EndFor 
			\State $P_l$:
			$\langle \tilde{\mathbf{y}} \rangle^\mathsf{A}_l = \mathcal{F}_{\mathsf{mux}}(\langle \mathbf{q} \rangle^\mathsf{B}_l, \langle \mathbf{S} \rangle^\mathsf{A}_l)$ \label{alg:anon_infer:L24}
			\State $P_0$ recovers $\tilde{\mathbf{y}}$ and gets $\mathbf{p} = \mathsf{sigmoid}(\tilde{\mathbf{y}})$.
		\label{alg:anon_infer:L25}
	\end{algorithmic}
\end{tcolorbox}
\vspace{-5pt}
\caption{Functionality and protocol of $\mathsf{Inference}$}
\label{alg:anon_infer}
\vspace{-3pt}
\end{figure}

\section{Security Proof}
\label{sec:appendix_security_proof}

\begin{theorem}
\label{theorem:ois}
$\varPi_{\mathsf{OIS}}$ (Figure~\ref{alg:ois_naive_cpsi}) is a secure protocol, which follows Definition~\ref{def:private_function} under ($\mathcal{F}_{\mathsf{OT}}$, $\mathcal{F}_{\mathsf{and}}$)-hybrid.
\end{theorem}

\begin{proof}[Proof for Theorem~\ref{theorem:ois} (Sketch)]
\label{proof:ois}

\noindent\textbf{Case 1: $P_0$ is corrupted}:
The simulator $\mathcal{S}$ interacts with $P_0$ as follows.
\begin{enumerate}[leftmargin = *, wide=0pt]

\item If $(z_*, u_*)$ is owned by $P_0$, $\mathcal{S}$ receives $\langle \mathbf{b}_* \rangle_1^\mathsf{B}$ from $P_0$ (Line~\ref{alg:ois_naive_cpsi:L3}).
Then $\mathcal{S}$ computes $\langle\mathbf{b}_*\rangle_0^\mathsf{B} = \langle\mathbf{b}_*\rangle_1^\mathsf{B}\oplus \widetilde{\mathbf{M}}_0[z_*B + u_*]$;
else $\mathcal{S}$ plays the sender of $\mathcal{F}_{\mathsf{OT}}$ with input of the XOR of randomly chosen $\mathbf{r}$($=\langle\mathbf{b}_*\rangle_0^\mathsf{B}$) and $\langle\widetilde{\mathbf{M}}_1\rangle^\mathsf{B}_0$ as in Line~\ref{alg:ois_naive_cpsi:L7}.
\item $\mathcal{S}$ plays the role of $\mathcal{F}_{\mathsf{and}}$ with inputs of $\langle \mathbf{b}_* \rangle_0^\mathsf{B}$, 
$\langle \mathbf{b} \rangle_0^\mathsf{B}$ to obtain 
$\langle \mathbf{b}^L \rangle_0^\mathsf{B}$, then sets $\langle \mathbf{b}^R \rangle_0^\mathsf{B} = \langle \mathbf{b}^L \rangle_0^\mathsf{B} \oplus \langle \mathbf{b} \rangle_0^\mathsf{B}$ as Line~\ref{alg:ois_naive_cpsi:L10}.
\end{enumerate}
To prove the simulation is indistinguishable from the real protocol, we consider the following hybrids.
\\
$\mathbf{Hybrid}_0$.
The same as the real protocol.
\\
$\mathbf{Hybrid}_1$.
Lines~\ref{alg:ois_naive_cpsi:L7}--\ref{alg:ois_naive_cpsi:L8} in Protocol $\Pi_{\mathsf{OIS}}$ use replaced $\mathcal{F}_{\mathsf{OT}}$.
\\
$\mathbf{Hybrid}_2$.
Line~\ref{alg:ois_naive_cpsi:L10} in Protocol $\Pi_{\mathsf{OIS}}$ uses replaced $\mathcal{F}_{\mathsf{and}}$.
\\
$\mathbf{Hybrid}_3$.
The same as the execution of $\mathcal{S}$ above.
\\
Note that $\mathbf{Hybrid}_1 \approx_c \mathbf{Hybrid}_0$ as $\mathcal{F}_{\mathsf{OT}}$ is secure, and $\mathbf{Hybrid}_2 \approx_c \mathbf{Hybrid}_1$ as $\mathcal{F}_{\mathsf{and}}$ is secure.
$\mathbf{Hybrid}_3$ is identical to $\mathbf{Hybrid}_2$.
Therefore, $\mathcal{S} \approx_c \mathcal{V}^{\mathbf{Hybrid}_0}_0$.

\noindent \textbf{Case 2: $P_1$ is corrupted}:
The simulator $\mathcal{S}$ plays the role $P_0$, and the interaction with $P_1$ is similar to the above.

In short, in both cases, the $\mathcal{F}_{\mathsf{OT}}$, $\mathcal{F}_{\mathsf{and}}$-hybrid $\Pi_{\mathsf{OIS}}$ protocol is secure in the semi-honest model.
\end{proof}

\begin{proof}[Proof for Theorem~\ref{theorem:bois} (Sketch)]


\textbf{Case 1: The sender $P_c$ is corrupted}:
The simulator $\mathcal{S}$ interacts with $P_c$ as follows.

\begin{enumerate}[leftmargin = *, wide=0pt]
	\item $\forall k \in \mathbf{k}$, $\mathcal{S}$ receives $\langle \mathbf{b}_*^{(k, c)} \rangle_{1 - c}^\mathsf{B}$ from $P_c$ and computes $\mathbf{r}^c$ with $(z_*^{(k)}, u_*^{(k)})$ and $\widetilde{\mathbf{M}}_c$.
	Then $\mathcal{S}$ sets $\langle \mathbf{b}_*^{(k, c)} \rangle_c^\mathsf{B} = \mathbf{r}^c$ and 
	plays the role of $\mathcal{F}_{\mathsf{and}}$ with $\langle \mathbf{b}_*^{(k, c)} \rangle_{c}^\mathsf{B}$ and $\langle \mathbf{b}^{(k, c)} \rangle_{c}^\mathsf{B}$.
	After obtaining $\langle \mathbf{b}^{(2k, c)} \rangle_{c}^\mathsf{B}$, $\mathcal{S}$ computes $\langle \mathbf{b}^{(2k + 1, c)} \rangle_{c}^\mathsf{B}$ as Lines~\ref{alg:bois:L2}--\ref{alg:bois:L5}.
	\item $\mathcal{S}$ constructs $\{\mathbf{T}\}$ and $\langle \mathbf{b}^{k, 1 - c}_* \rangle_{c}^\mathsf{B}$ ($\forall k\in\mathbf{k}$) as Lines~\ref{alg:bois:L6}--\ref{alg:bois:L11}.
	Then $\mathcal{S}$ plays the role of $\mathcal{F}_{\mathsf{OPPRF}}$ with $\{\mathbf{T}\}$.
	\item $\forall k \in \mathbf{k}$, $\mathcal{S}$ plays the role of $\mathcal{F}_{\mathsf{and}}$ with inputs $\langle \mathbf{b}^{(k, 1 - c)}_* \rangle^\mathsf{B}_l$ and $\langle \mathbf{b}^{(k, 1 - c)} \rangle^\mathsf{B}_l$, and obtain $\langle \mathbf{b}^{2k + 1, 1 - c} \rangle_{c}^\mathsf{B}$ as Line~\ref{alg:bois:L15}.
	Finally, $\mathcal{S}$ computes $\langle \mathbf{b}^{2k + 1, 1 - c} \rangle_{c}^\mathsf{B}$ as Line~\ref{alg:bois:L16}.
\end{enumerate}

For indistinguishable simulation, we consider:
\\
$\mathbf{Hybrid}_0$.
The same as the real protocol.
\\
$\mathbf{Hybrid}_1$.
Line~\ref{alg:bois:L5} in $\varPi_{\mathsf{BOIS}}$ uses replaced $\mathcal{F}_{\mathsf{and}}$.
\\
$\mathbf{Hybrid}_2$.
Line~\ref{alg:bois:L12} uses replaced $\mathcal{F}_{\mathsf{OPPRF}}$.
\\
$\mathbf{Hybrid}_3$.
Lines~\ref{alg:bois:L15}--\ref{alg:bois:L16} use replaced $\mathcal{F}_{\mathsf{and}}$.
\\
$\mathbf{Hybrid}_4$.
The same as the execution of $\mathcal{S}$ above.
\\
Note that $\mathbf{Hybrid}_1 \approx_c \mathbf{Hybrid}_0$ as $\mathcal{F}_{\mathsf{and}}$ is secure;
$\mathbf{Hybrid}_2 \approx_c \mathbf{Hybrid}_1$ as $\mathcal{F}_{\mathsf{OPPRF}}$ is secure; $\mathbf{Hybrid}_3 \approx_c \mathbf{Hybrid}_2$ as $\mathcal{F}_{\mathsf{and}}$ is secure; 
$\mathbf{Hybrid}_4$ is identical to $\mathbf{Hybrid}_3$.
Therefore, $\mathcal{S} \approx_c \mathcal{V}_c^{\mathbf{Hybrid}_0}$.

\textbf{Case 2: The receiver $P_{1 - c}$ is corrupted}:
$\mathcal{S}$ plays the role of $P_c$ and interacts with $P_{1 - c}$ like Case~1 except:
\begin{enumerate}[leftmargin = *, wide=0pt]
	\item $\mathcal{S}$ ignores the process of Lines~\ref{alg:bois:L2}--\ref{alg:bois:L4}, and chooses randomly $\langle \mathbf{b}_*^{(k, c)} \rangle_{1 - c}^\mathsf{B}$ ($\forall k\in\mathbf{k}$) and sends to $P_{1 - c}$.
	\item $\mathcal{S}$ ignores the process of Lines~\ref{alg:bois:L6}--\ref{alg:bois:L11}, and plays the role as $\mathcal{F}_{\mathsf{OPPRF}}$ with input $\mathbf{T}_{\mathsf{ch}}^{1 - c}$ to obtain $\widetilde{\mathbf{b}}_*^{1 - c}$.
	Then $\mathcal{S}$ extracts $\langle \mathbf{b}_*^{(k, 1 - c)} \rangle_{1 - c}^\mathsf{B}$ $\forall k\in\mathbf{k}$ as Line~\ref{alg:bois:L14}.
\end{enumerate}

In short, in both cases, the $\mathcal{F}_{\mathsf{and}}$, $\mathcal{F}_{\mathsf{OPPRF}}$-hybrid $\varPi_{\mathsf{BOIS}}$ protocol is secure in the semi-honest model.
\end{proof}

\begin{proof}[Proof for Theorem~\ref{theorem:advaned_sys} (Sketch)]

\textbf{Case 1: $P_0$ is corrupted}:
The simulator $\mathcal{S}$ interacts with $P_0$ as follows.
\begin{enumerate}[leftmargin = *, wide=0pt]
	\item $\mathcal{S}$ plays the sender and receiver of two $\mathcal{F}_{\mathsf{CPSI}}$ with $\{\mathbf{ID}_0,\mathbf{y}\}$, respectively, as Lines~\ref{alg:final_sys:L2}--\ref{alg:final_sys:L3}.
	Then $\mathcal{S}$ obtains $\langle \mathbf{b}^{(1, 0)} \rangle_0^\mathsf{B}$, $\langle \mathbf{y}^0 \rangle_0^\mathsf{A}$, $\langle \mathbf{b}^{(1, 1)} \rangle_0^\mathsf{B}$ and $\langle \mathbf{y}^1 \rangle_0^\mathsf{A}$.
	\item For tree $t \in \mathcal{T}_t$, $\mathcal{S}$ plays the role of $\mathcal{F}_{\mathsf{sigmoid}}$, $\mathcal{F}_{\mathsf{mul}}$, and $\mathcal{F}_{\mathsf{mux}}$ with input $\langle \mathbf{b}^{(1,i)} \rangle_0^\mathsf{B}$ and $\langle \mathbf{y}^i \rangle_0^\mathsf{A}$ to derive the gradients $\langle \mathbf{g}^{(1,i)} \rangle_0^\mathsf{A}$, $\langle \mathbf{h}^{(1,i)} \rangle_0^\mathsf{A}$, for $i \in \{0, 1\}$ as Lines~\ref{alg:final_sys:L5}--\ref{alg:final_sys:L6}.
	\item $\forall$ internal node $k$ of tree $t$, if $k$ is a root/left node, $\mathcal{S}$ plays the role of $\mathcal{F}_{\mathsf{BinMatVec}}$ alternately to obtain $\langle \tilde{\mathbf{g}}^{(k,i)} \rangle_0^\mathsf{A} || \langle \tilde{\mathbf{h}}^{(k,i)} \rangle_0^\mathsf{A}$ for $i \in \{0, 1\}$, which are concatenated as Line~\ref{alg:final_sys:L13};
	else $\mathcal{S}$ calculates $\langle \widetilde{\mathbf{g}}^{(k)} \rangle_0^\mathsf{A}$ and $\langle \widetilde{\mathbf{h}}^{(k)} \rangle_0^\mathsf{A}$ as Line~\ref{alg:final_sys:L14}.
	\item $\forall$ internal node $k$ of tree $t$, $\mathcal{S}$ plays the role of $\mathcal{F}_{\mathsf{mul}}$ and $\mathcal{F}_{\mathsf{div}}$ with $\langle \tilde{\mathbf{g}}^{(k)} \rangle_0^\mathsf{A}$, $\langle \tilde{\mathbf{h}}^{(k)} \rangle_0^\mathsf{A}$ to obtain $\langle \mathbf{L}_{\mathsf{sp}} \rangle_0^\mathsf{A}$ as Line~\ref{alg:final_sys:L15}.
	\item $\forall$ internal node $k$ of tree $t$, $\mathcal{S}$ plays the role of $\mathcal{F}_{\mathsf{argmax}}$ with input $\langle \mathbf{L}_{\mathsf{sp}} \rangle_0^\mathsf{A}$ to obtain $\langle z_*^{(k)} \rangle^\mathsf{A}_0$ and $\langle u_*^{(k)} \rangle^\mathsf{A}_0$ as Line~\ref{alg:final_sys:L16}.
	\item $\forall$ internal node $k$ of tree $t$, $\mathcal{S}$ plays the role of $\mathcal{F}_{\mathsf{greater}}$ with inputs $\langle z_*^{(k)} \rangle_0^\mathsf{A}$ to obtain $\langle c \rangle_0^\mathsf{B}$, and then open $c$ to recovery $z_*^{(k)}$ and $u_*^{(k)}$ to $P_c$, which is exactly equal to the programmed outputs in real execution when $c = 0$.
	Finally, $\mathcal{S}$ appends $k$ into $\mathbf{k}_c$ as Line~\ref{alg:final_sys:L17}.
	\item $\forall$ level $d$ in tree $t$, $\mathcal{S}$ plays the role of $\mathcal{F}_{\mathsf{BOIS}}$ to obtain $\{\langle \mathbf{b}^{(2k, i)} \rangle_0^\mathsf{B}, \langle \mathbf{b}^{(2k + 1,i)} \rangle_0^\mathsf{B}\}_{k \in \mathbf{k}_i, i \in \{0, 1\}}$ as Lines~\ref{alg:final_sys:L18}--\ref{alg:final_sys:L19}.
	\item $\mathcal{S}$ plays the role of $\mathcal{F}_{\mathsf{mux}}$ with the inputs $\langle \mathbf{b}^{(2k, i)} \rangle_0^\mathsf{B}$, $\langle \mathbf{g}^{(k,i)} \rangle_0^\mathsf{A}$ and $\langle \mathbf{h}^{(k,i)} \rangle_0^\mathsf{A}$ to update the gradients of left child nodes $\langle \mathbf{g}^{(2k, i)} \rangle_0^\mathsf{A}$, $\langle \mathbf{h}^{(2k, i)} \rangle_0^\mathsf{A}$.
	Then update the gradients of right child nodes as Line~\ref{alg:final_sys:L21}.
	\item For tree $t$, $\mathcal{S}$ plays the role of $\mathcal{F}_{\mathsf{div}}$ with inputs $\langle \mathbf{g}^{(k, 0)} \rangle_0^\mathsf{A}, \langle \mathbf{h}^{(k, 0)} \rangle_0^\mathsf{A}$ to obtain the weight $\langle \mathbf{w} \rangle_0^\mathsf{A}$ as Line~\ref{alg:final_sys:L22}.
	\item For tree $t$, $\mathcal{S}$ plays the role of $\mathcal{F}_{\mathsf{mux}}$ with inputs $\langle \mathbf{b}^{(k, i)} \rangle_0^\mathsf{B}$, $\langle \mathbf{w}[k] \rangle_0^\mathsf{A}$ to update $\langle \tilde{\mathbf{y}}^i \rangle_0^\mathsf{A}$ ($i \in \{0, 1\}$) as Line~\ref{alg:final_sys:L23}.
\end{enumerate}

To prove the simulation is indistinguishable from the real protocol,
we consider the following hybrids.\\
$\mathbf{Hybrid}_0$.
The same as the real protocol.
\\
$\mathbf{Hybrid}_1$.
Lines~\ref{alg:final_sys:L2}--\ref{alg:final_sys:L3} ($\varPi_{\mathsf{GBDT}}^{\mathsf{OTSA}}$) use replaced $\mathcal{F}_{\mathsf{CPSI}}$.
\\
$\mathbf{Hybrid}_2$.
Lines~\ref{alg:final_sys:L5}--\ref{alg:final_sys:L6} use replaced $\mathcal{F}_{\mathsf{sigmoid}}$, $\mathcal{F}_{\mathsf{mul}}$, $\mathcal{F}_{\mathsf{mux}}$.
\\
$\mathbf{Hybrid}_3$.
Lines~\ref{alg:final_sys:L11}--\ref{alg:final_sys:L12} use replaced $\mathcal{F}_{\mathsf{BinMatVec}}$.
\\
$\mathbf{Hybrid}_4$.
Line~\ref{alg:final_sys:L15} use replaced $\mathcal{F}_{\mathsf{mul}}$ and $\mathcal{F}_{\mathsf{div}}$.
\\
$\mathbf{Hybrid}_5$.
Line~\ref{alg:final_sys:L16} uses replaced $\mathcal{F}_{\mathsf{argmax}}$.\\
$\mathbf{Hybrid}_6$.
Line~\ref{alg:final_sys:L17} uses replaced $\mathcal{F}_{\mathsf{greater}}$.
\\
$\mathbf{Hybrid}_7$.
Lines~\ref{alg:final_sys:L18}--\ref{alg:final_sys:L19} use replaced $\mathcal{F}_{\mathsf{BOIS}}$.
\\
$\mathbf{Hybrid}_8$.
Line~\ref{alg:final_sys:L20} uses replaced $\mathcal{F}_{\mathsf{mux}}$.
\\
$\mathbf{Hybrid}_9$.
Line~\ref{alg:final_sys:L22} uses replaced $\mathcal{F}_{\mathsf{div}}$.
\\
$\mathbf{Hybrid}_{10}$.
Line~\ref{alg:final_sys:L23} uses replaced $\mathcal{F}_{\mathsf{mux}}$.
\\
$\mathbf{Hybrid}_{11}$.
The same as the execution of $\mathcal{S}$ above.

Note that $\mathbf{Hybrid}_1 \approx_c \mathbf{Hybrid}_0$ as $\mathcal{F}_{\mathsf{CPSI}}$ is secure; $\mathbf{Hybrid}_2 \approx_c \mathbf{Hybrid}_1$ as $\mathcal{F}_{\mathsf{sigmoid}}$ and $\mathcal{F}_{\mathsf{mul}}$ are secure; $\mathbf{Hybrid}_3 \approx_c \mathbf{Hybrid}_2$ as $\mathcal{F}_{\mathsf{BinMatVec}}$ is secure based on the semantic security of RLWE ciphertexts;
$\mathbf{Hybrid}_4 \approx_c \mathbf{Hybrid}_3$ as $\mathcal{F}_{\mathsf{mul}}$ and $\mathcal{F}_{\mathsf{div}}$ is secure; $\mathbf{Hybrid}_5 \approx_c \mathbf{Hybrid}_4$ as $\mathcal{F}_{\mathsf{argmax}}$ is secure; $\mathbf{Hybrid}_6 \approx_c \mathbf{Hybrid}_5$ as $\mathcal{F}_{\mathsf{greater}}$ is secure, and the recovered $z_*^{(k)}$ and $u_*^{(k)}$ are programmed into ideal functionality as the identical output between real execution and simulation; $\mathbf{Hybrid}_7 \approx_c \mathbf{Hybrid}_6$ as $\mathcal{F}_{\mathsf{BOIS}}$ is secure; $\mathbf{Hybrid}_8 \approx_c \mathbf{Hybrid}_7$ as $\mathcal{F}_{\mathsf{mux}}$ is secure; $\mathbf{Hybrid}_9 \approx_c \mathbf{Hybrid}_8$ as $\mathcal{F}_{\mathsf{div}}$ is secure; $\mathbf{Hybrid}_{10} \approx_c \mathbf{Hybrid}_9$ as $\mathcal{F}_{\mathsf{mux}}$ is secure.
$\mathbf{Hybrid}_{11}$ is identical to $\mathbf{Hybrid}_{10}$.
Thus, $\mathcal{S}_0 \approx_c \mathcal{V}_0^{\mathbf{Hybrid}_0}$.

\textbf{Case 2: $P_1$ is corrupted}:
The simulator plays the role of $P_0$ and the interaction with $P_1$ is similar for the symmetric tree building of $\Pi_{\mathsf{GBDT}}^{\mathsf{OTSA}}$, except for the simulation of Lines~\ref{alg:final_sys:L2}--\ref{alg:final_sys:L3} in $\Pi_{\mathsf{GBDT}}^{\mathsf{OTSA}}$ as input without labels, which can be simulated as $\mathcal{F}_{\mathsf{CPSI}}$ is secure.

In both cases, under the 
$(\mathcal{F}_{\mathsf{argmax}}$-, 
$\mathcal{F}_{\mathsf{BOIS}}$-, 
\mbox{$\mathcal{F}_{\mathsf{BinMatVec}}$-,}
$\mathcal{F}_{\mathsf{CPSI}}$-, 
$\mathcal{F}_{\mathsf{div}}$-, 
$\mathcal{F}_{\mathsf{greater}}$-,
$\mathcal{F}_{\mathsf{mul}}$-, 
$\mathcal{F}_{\mathsf{mux}}$-,
$\mathcal{F}_{\mathsf{sigmoid}})$-hybrid, and with semantically-secure RLWE-based encryption,
$\Pi_{\mathsf{GBDT}}^{\mathsf{OTSA}}$ is secure in the semi-honest model.
\end{proof}

\newpage 


\section{Meta-Review}

The following meta-review was prepared by the program committee for the 2026
IEEE Symposium on Security and Privacy (S\&P) as part of the review process as
detailed in the call for papers.

\subsection{Summary}
This paper studies efficient two-party computation for the task of gradient boosting decision tree training and inference. The paper assumes that the data is vertically distributed and thus alignment is required. The most interesting optimization is how the paper optimizes circuit-PSI using homomorphic encryption to reduce the cost of alignment when the payload can be large.

\subsection{Scientific Contributions}
\begin{itemize}
\item Provides a Valuable Step Forward in an Established Field
\end{itemize}

\subsection{Reasons for Acceptance}
The paper makes a valuable step towards practical two-party GBDT by using homomorphic encryption to reduce the communication complexity in PSI. This leads to significant savings in their experience.

\end{document}